\documentclass{article}

\PassOptionsToPackage{numbers,compress}{natbib}
\usepackage[preprint]{neurips_2021}
\usepackage[utf8]{inputenc} % allow utf-8 input
\usepackage[T1]{fontenc}    % use 8-bit T1 fonts
\usepackage{hyperref}       % hyperlinks
\usepackage{url}            % simple URL typesetting
\usepackage{booktabs}       % professional-quality tables
\usepackage{amsfonts}       % blackboard math symbols
\usepackage{nicefrac}       % compact symbols for 1/2, etc.
\usepackage{microtype}      % microtypography
\usepackage{wrapfig}
\pdfoutput=1

\usepackage{url}
\usepackage{amsmath}
 
\usepackage{amsthm}
\usepackage{algorithm}
\usepackage{algorithmicx}
\usepackage[noend]{algpseudocode}
\usepackage{verbatim}
\usepackage{graphicx}
\usepackage[space]{grffile}
\usepackage{tikz}
\usepackage{subfigure}
\usetikzlibrary{arrows}
\newtheorem{theorem}{Theorem}

\newtheorem{lemma}{Lemma}

\theoremstyle{definition}

\newtheorem*{deff*}{Definition}
\newtheorem{claim}{Claim}

\graphicspath{{figures/}}
\usepackage{amsmath}
\usepackage{mathrsfs}
\usepackage{amssymb}
\usepackage{booktabs}
\usepackage{changepage}
\usepackage{xspace}
\usepackage{hyperref}
\usepackage{thmtools}
\usepackage{thm-restate}

\usepackage{tikz}
\usetikzlibrary{arrows}

\usepackage{xcolor}
\hypersetup{
    colorlinks,
    linkcolor={red!80!black},
    citecolor={blue!60!black},
    urlcolor={blue!80!black}
}

\newcommand{\ie}{\textit{i.e.}\xspace}
\newcommand{\eg}{\textit{e.g.}\xspace}

\newcommand{\ex}[1]{ \mathbb{E} \left[ #1 \right] }

\newcommand{\prob}[1]{ Pr \left( #1 \right) }

\newcommand{\Lamb}[0]{\Lambda}
\newcommand{\lamb}[0]{\lambda}
\newcommand{\epsi}[0]{ \varepsilon }
\newcommand{\reals}[0]{ \mathbb{R}^+ }

\newcommand{\ff}[1]{ f \left( #1 \right) }

\newcommand{\fracc}[2]{ #1 / #2 }

\newcommand{\threshold}{\textsc{ThresholdSeq}\xspace}
\newcommand{\ts}{\textsc{TS}\xspace}
\newcommand{\fastlinear}{\textsc{LinearSeq}\xspace}
\newcommand{\bsm}{\textsc{BinarySearchMaximization}\xspace}
\newcommand{\exm}{\textsc{ExhaustiveMaximization}\xspace}
\newcommand{\fastlinearshort}{\textsc{LS}\xspace}
\newcommand{\ls}{\textsc{LS}\xspace}
\newcommand{\linearseq}{\textsc{LinearSeq}\xspace}
\newcommand{\low}{\textsc{LowAdapLinearSeq}\xspace}
\newcommand{\boost}{\textsc{ParallelGreedyBoost}\xspace}
\newcommand{\boostshort}{\textsc{PGB}\xspace}
\newcommand{\flsabr}{\textsc{LS+PGB}\xspace}

\newcommand{\plg}{\texttt{ParallelLazyGreedy}\xspace}
\newcommand{\oh}[1]{O\left( #1 \right)}

\newcommand\numberthis{\addtocounter{equation}{1}\tag{\theequation}}

\newcommand{\fast}{\textsc{FAST}\xspace}

\newcommand{\opt}{\textsc{OPT}\xspace}

\newcommand{\sm}{\textsc{SM}\xspace}
\newcommand{\marge}[2]{\Delta \left( #1 \, | \, #2 \right) }

\newcommand{\rev}{black}

\allowdisplaybreaks[4]

\usepackage[utf8]{inputenc} % allow utf-8 input
\usepackage[T1]{fontenc}    % use 8-bit T1 fonts
\usepackage{hyperref}       % hyperlinks
\usepackage{url}            % simple URL typesetting
\usepackage{booktabs}       % professional-quality tables
\usepackage{amsfonts}       % blackboard math symbols
\usepackage{nicefrac}       % compact symbols for 1/2, etc.
\usepackage{microtype}      % microtypography
\title{Best of Both Worlds: Practical and Theoretically Optimal Submodular Maximization in Parallel}

\author{%
  Yixin Chen \\ 
  Department of Computer Science \& Engineering\\
  Texas A\&M University\\
  College Station, Texas \\
  \texttt{chen777@tamu.edu} \\
  \And
  Tonmoy Dey \\ 
  Department of Computer Science\\
  Florida State University\\
  Tallahassee, Florida \\
  \texttt{td18d@my.fsu.edu} \\
  \And 
  Alan Kuhnle \\
  Department of Computer Science \& Engineering\\
  Texas A\&M University\\
  College Station, Texas \\
  \texttt{kuhnle@tamu.edu} \\
}

\begin{document}

\maketitle

%-----------------------------------------
\begin{abstract}
    For the problem of maximizing
    a monotone, submodular function
    with respect to a cardinality constraint $k$ on a ground set of size $n$,
    we provide an algorithm
    that achieves the state-of-the-art in both
    its empirical performance and its theoretical properties,
    in terms of adaptive complexity,
    query complexity, and approximation ratio;
    that is, it obtains, with high probability,
    query complexity of $\oh{n}$ in expectation,
    adaptivity of $\oh{\log(n)}$, 
    and approximation ratio of nearly $1-1/e$.
    The main algorithm is assembled from two
    components which may be of independent
    interest.
    The first component of our algorithm, \textsc{LinearSeq},
    is useful as a preprocessing algorithm to improve the query
    complexity of many algorithms.
    Moreover, a
    variant of \textsc{LinearSeq} is shown to have adaptive complexity
    of $O( \log (n / k) )$ which is smaller than that of any previous algorithm
    in the literature.
    The second component is a parallelizable thresholding procedure
    \threshold for adding elements
    with gain above a constant threshold.
    Finally, we demonstrate that our main algorithm empirically
    outperforms, in terms of runtime, adaptive rounds, total queries, and objective values, 
    the previous state-of-the-art algorithm \textsc{FAST} 
    in a comprehensive evaluation with six submodular objective functions.
\end{abstract}
%%% Local Variables:
%%% mode: latex
%%% TeX-master: "./main.tex"
%%% End:

\color{\rev}
\textbf{Version v3.}
This version fixes two issues in the previous version. 
Firstly, the prefix selection step in both \linearseq (Alg.~\ref{alg:fastlinear}) 
and \threshold (Alg.~\ref{alg:threshold}) might result in a significant loss in the objective value. Specifically, when the size of the prefix selected is small ($< 1/\epsi$), 
adding another bad block can ruin the objective value. 
To address this, we no longer include a bad block if the prefix is small. 
This is a minor change to the algorithm.

Secondly, the analysis of \linearseq incorrectly uses Wald’s equation 
to bound the summation of dependent random variables, 
specifically in the proof of Inequality~\ref{ineq:linear-fail-1}, 
as Wald’s equation does not apply. 
In this version, we bound the probabilities of these events in a different way.
% Recently, we observed two issues in the original NeurIPS 2021 manuscript.
% Firstly, the prefix selection step in both \linearseq (Alg.~\ref{alg:fastlinear}) 
% and \threshold (Alg.~\ref{alg:threshold}) might result in a significant loss in the objective value, 
% contrary to the original analysis.
% Specifically, when the size of the prefix selected is too small,
% adding another bad block can ruin the objective value.
% To address this, we now discard the bad good block if
% the index selected is too small.
% This change does not impact our main results a lot.

% Secondly, the analysis of \linearseq incorrectly uses Wald’s equation 
% to bound the summation of dependent random variables, 
% specifically in the proof of Inequality 3, 
% as Wald’s equation does not apply to all sequences of dependent random variables.
% In this updated manuscript, we construct another event 
% with a fixed number of indices to bound the original event
% that has random indices.
\color{black}

\section{Introduction}
\label{sec:intro}
The cardinality-constrained optimization of a monotone, submodular function
$f:2^{\mathcal N} \to \reals$,
defined on subsets of a ground set $\mathcal N$ of size $n$,
is a general problem formulation that is ubiquitous
in wide-ranging applications, \eg 
video or image summarization
\citep{Mirzasoleiman2018}, 
network monitoring \citep{Leskovec2007},
information gathering \citep{Krause2007},
and MAP Inference for Determinantal
Point Processes \citep{Gillenwater2012},
among many others.
The function $f:2^{\mathcal N} \to \reals$ is submodular iff
for all $S \subseteq T \subseteq \mathcal N$, $x \not \in T$, 
$\marge{x}{T} \le \marge{x}{S}$\footnote{$\marge{x}{S}$ denotes the \textit{marginal gain} of $x$ to $S$: $f(S \cup \{x\}) - f(S)$.}; and
the function $f$ is monotone if $f( S ) \le f(T)$ for all $S \subseteq T$.
In this paper, we study the following submodular maximization problem (SM)
\begin{equation*}
  \text{maximize} f(S), \text{ subject to } |S| \le k, \tag{SM}
\end{equation*}
where $f$ is a monotone, submodular function;
\sm is an NP-hard problem. There has been extensive
effort into the design of approximation algorithms
for \sm over the course of
more than 45 years, \eg \citep{Nemhauser1978a,Conforti1984,Calinescu2007,Horel2016,Kuhnle2020a}. 
For \sm, the optimal ratio
has been shown to be $1 - 1/e \approx 0.63$ \citep{Nemhauser1978a}.

As instance sizes have grown very large,
there has been much effort into the design
of efficient, parallelizable algorithms
for \sm.
Since queries to the objective function
can be very expensive,
the overall efficiency of
an algorithm for \sm is typically measured
by the \emph{query complexity}, or number of
calls made to the objective function $f$ \citep{Badanidiyuru2014,Breuer2019}.
The degree of parallelizability can be measured by the
\emph{adaptive complexity} of an algorithm, which is the minimum
number of rounds into which the queries to $f$ may be 
organized, such that within each round, the queries
are independent and hence may be arbitrariliy parallelized.
Observe that the lower the adaptive complexity, the
more parallelizable an algorithm is.
To obtain a constant approximation factor,
a lower bound of $\Omega(n)$ has been shown on the query
complexity \citep{Kuhnle2020a} and a lower bound
of $\Omega( \log(n) / \log \log (n) )$ has been shown
on the adaptive complexity \citep{Balkanski2018}.

Several algorithms have
been developed recently that are nearly optimal in terms
of query and adaptive complexities
\citep{Ene,Chekuri2018a,Fahrbach2018,Balkanski};
that is, these algorithms achieve $\oh{ \log n }$
adaptivity and $\oh{n \text{polylog}(n)}$ query 
complexity (see Table \ref{table:cmp}).
However, these algorithms use
sampling techniques that result in very large
constant factors that make these algorithms
impractical. This fact is discussed
in detail in \citet{Breuer2019}; as an illustration,
to obtain ratio $1 - 1/e - 0.1$ with 95\%
confidence, all of these algorithms require more than $10^6$ queries
of sets of size $k / \log(n)$ in every adaptive round \citep{Breuer2019};
moreover, even if these algorithms are run as heuristics using a single
sample, other inefficiencies preclude these algorithms of
running even on moderately sized instances \citep{Breuer2019}.
For this reason, the \fast algorithm of \citet{Breuer2019} has been
recently proposed, which uses an entirely different sampling technique
called adaptive sequencing. Adaptive sequencing
 was originally introduced in
\citet{Balkanski2018c}, but the original version has quadratic query
complexity in the size of the ground set and hence is still impractical
on large instances.
To speed it up, the \fast algorithm
sacrifices theoretical guarantees
to yield an algorithm that parallelizes well and
is faster than all previous algorithms
for \sm in an extensive experimental evaluation.
The theoretical sacrifices of \fast include:
the adaptivity of \fast is $\Omega( \log(n)\log^2( \log n ))$, which is higher than the state-of-the-art,
and more significantly,
the algorithm obtains no approximation ratio for $k < 850$\footnote{The approximation ratio $1 - 1/e - 4\epsi$ of \fast
holds with probability
$1 - \delta$ for $k \ge \theta( \epsi, \delta, k ) = \fracc{2\log( 2 \delta^{-1} \log( \frac{1}{\epsi}\log (k) ) )}{\epsi^2 (1 - 5 \epsi) }$.};
since many applications require small choices for $k$, this limits the practical utility of \fast. A natural question is thus: 
\emph{is it possible to design an algorithm that is both practical and theoretically optimal in
  terms of adaptivity, ratio, and total queries?}

\subsection{Contributions}
\begin{table*}[t] \caption{Theoretical comparison to algorithms that achieve
    nearly optimal adaptivity, ratio, and query complexity: 
    the algorithm of \citet{Ene}, 
    \textsc{randomized-parallel-greedy} (\textsc{RPG}) of \citet{Chekuri2018a}, 
    \textsc{Binary-Search-Maximization} (\textsc{BSM}) 
    and \textsc{Subsample-Maximization} (\textsc{SM}) of \citet{Fahrbach2018},
    \fast of \citet{Breuer2019}.
    The symbol $\dagger$ indicates the result holds with constant probability or in expectation;
    the symbol $\ddagger$ indicates the result does not hold on all instances of \sm;
    while no symbol indicates the result holds with probability greater than $1 - O(1/n)$. 
    Observe that our algorithm \flsabr dominates in at least one category 
    when compared head-to-head with any one of the previous algorithms.
    }\label{table:cmp}
\begin{center} %\small
  \resizebox{\textwidth}{!}{
\begin{tabular}{lllllll} \toprule
Reference  &  Ratio & Adaptivity & Queries \\ %& Deterministic? \\
  \midrule
  % Ene, Nguyen \citep{Ene} & $1 - 1/e - \epsi$ & $\oh{ \frac{1}{\epsi^2} \log (n) }$ & $\oh{n \text{poly}(\log n, 1/ \epsi ) }$\\
  % Chekuri, Quanrud \citep{Chekuri2018a} & $1 - 1/e - \epsi$ & $\oh{ \frac{1}{\epsi^2}\log(n) } \, \dagger$  & $\oh{ \frac{n}{\epsi^4} \log(n) } \, \dagger$   \\
  % Fahrbach, Mirrokni, Zadimoghaddam \citep{Fahrbach2018}  & $1 - 1/e - \epsi \, \dagger$  & $\oh{ \frac{1}{\epsi^2} \log \left( n \right)} $ & $\oh{ \frac{n}{\epsi^3} \log \log k }\, \dagger$ \\
  % Fahrbach, Mirrokni, Zadimoghaddam \citep{Fahrbach2018} & $1 - 1/e - \epsi \, \dagger$ & $O \left( \frac{1}{\epsi^2} \log \left( n \right) \right)$ & $\oh{ \frac{n}{\epsi^3} \log(1/\epsi) } \, \dagger$ \\
  % Breuer, Balkanski, Singer \citep{Breuer2019} & $1 - 1/e - \epsi \, \dagger$ & $O \left(\frac{1}{\epsi^2} \log(n)\log^2\left(\frac{\log(k)}{\epsi}\right) \right)$ & $\oh{ \frac{n}{\epsi^2} \log\left(\frac{\log(k)}{\epsi}\right) } $ \\
  \citet{Ene} & $1 - 1/e - \epsi$ & $\oh{ \frac{1}{\epsi^2} \log (n) }$ & $\oh{n \text{poly}(\log n, 1/ \epsi ) }$\\
  \citet{Chekuri2018a} (\textsc{RPG}) & $1 - 1/e - \epsi$ & $\oh{ \frac{1}{\epsi^2}\log(n) } \, \dagger$  & $\oh{ \frac{n}{\epsi^4} \log(n) } \, \dagger$   \\
  \citet{Fahrbach2018} (\textsc{BSM})  & $1 - 1/e - \epsi \, \dagger$  & $\oh{ \frac{1}{\epsi^2} \log \left( n \right)} $ & $\oh{ \frac{n}{\epsi^3} \log \log k }\, \dagger$ \\
  \citet{Fahrbach2018} (\textsc{SM}) & $1 - 1/e - \epsi \, \dagger$ & $O \left( \frac{1}{\epsi^2} \log \left( n \right) \right)$ & $\oh{ \frac{n}{\epsi^3} \log(1/\epsi) } \, \dagger$ \\
  \citet{Breuer2019} (\fast) & $1 - 1/e - \epsi  \, \ddagger$ & $O \left(\frac{1}{\epsi^2} \log(n)\log^2\left(\frac{\log(k)}{\epsi}\right) \right)$ & $\oh{ \frac{n}{\epsi^2} \log\left(\frac{\log(k)}{\epsi}\right) } $ \\
  \midrule
  %\fastlinearshort  & $(4 + \oh{\epsi } )^{-1}$ (wh & $\oh{\frac{1}{\epsi^3} \log (n) }$ & $\oh{ \frac{n}{\epsi^3 } }$ (exp) \\
  \flsabr [Theorem \ref{theorem:boost}] & $1 - 1/e - \epsi$ & $O \left( \frac{1}{\epsi^2} \log \left( n / \epsi \right) \right)$ & $\oh{ \frac{n}{\epsi^2} } \, \dagger$ \\
  \bottomrule
%  ? & $2e/(e - 1)+ \alpha + \epsi$ & $\qsmem$ & $O(1)$ & $\mathcal T ( \mathcal A )$ & Depends on $\mathcal A$ \\\bottomrule
\end{tabular}}
\setcounter{footnote}{0}
\end{center}
\end{table*}
\begin{table*}[t] \caption{Empirical comparison of our algorithm with \fast, 
  with each algorithm using 75 CPU threads on a broad range of $k$ values and six applications; 
  details provided in Section \ref{sec:exp}. Parameter settings are favorable to \fast, which
    is run without theoretical guarantees 
    while \flsabr enforces ratio $\approx 0.53$ with high probability.
    For each application, we report the arithmetic mean of each metric over all instances.
    Observe that \flsabr outperforms \fast in runtime on all applications.
    } \label{table:cmp-exp}
\begin{center} %\small
\begin{tabular}{lllllll} \toprule
    & \multicolumn{2}{c}{Runtime (s)} & \multicolumn{2}{c}{Objective Value} & \multicolumn{2}{c}{Queries}  \\ %& Deterministic? \\
  Application       & \fast & \flsabr & \fast & \flsabr & \fast & \flsabr \\
  \midrule
  TrafficMonitor & $3.7 \times 10^{-1}$ & $ \textbf{2.1} \times \textbf{10}^\textbf{-1}$ & $4.7 \times 10^8$ & $\textbf{5.0} \times \textbf{10}^\textbf{8}$ & $3.5 \times 10^3$ & $\textbf{2.4} \times \textbf{10}^\textbf{3}$  \\
  InfluenceMax & $4.4 \times 10^0$ & $\textbf{2.3} \times \textbf{10}^\textbf{0}$ & $1.1 \times 10^3$ & $1.1 \times 10^3$ & $7.7 \times 10^4$ & $\textbf{4.0} \times \textbf{10}^\textbf{4}$ \\ 
  TwitterSumm & $1.6 \times 10^1$ & $\textbf{3.5} \times \textbf{10}^\textbf{0}$ & $3.8 \times 10^5$ & $3.8 \times 10^5$ & $1.5 \times 10^5$ & $\textbf{6.2} \times \textbf{10}^\textbf{4}$ \\  
  RevenueMax & $3.9 \times 10^2$ & $\textbf{5.4} \times \textbf{10}^\textbf{1}$ & $1.4 \times 10^4$ & $1.4 \times 10^4$ & $7.6 \times 10^4$ & $\textbf{2.7} \times \textbf{10}^\textbf{4}$ \\ 
  MaxCover (BA) & $3.7 \times 10^2$ & $\textbf{7.6} \times \textbf{10}^\textbf{1}$ & $6.0 \times 10^4$ & $\textbf{6.3} \times \textbf{10}^\textbf{4}$ & $5.8 \times 10^5$ & $\textbf{1.8} \times \textbf{10}^\textbf{5}$ \\
  ImageSumm & $1.6 \times 10^1$ & $\textbf{8.1} \times \textbf{10}^\textbf{-1}$ & $9.1 \times 10^3$ & $9.1 \times 10^3$ & $1.3 \times 10^5$ & $\textbf{4.8} \times \textbf{10}^\textbf{4}$ \\ \bottomrule
\end{tabular}
\setcounter{footnote}{0}
\end{center}
\end{table*}
% In summary, we present an 
% algorithm \flsabr for parallel submodular maximization that
% achieves the state-of-the-art both theoretically (Table \ref{table:cmp})
% and empirically (Table \ref{table:cmp-exp}). The algorithm \flsabr
% is composed of three procedures: \linearseq, \threshold,
% and \boost.
% This organization breaks the final algorithm down
% into intuitively distinct components. Of these,
% \linearseq and \threshold are broadly applicable
% to other algorithms and problems and hence constitute
% contributions of independent interest.
%znTo obtain this 
%algorithm, we introduce a linear-time, optimally adaptive pre-processing
%technique (\fastlinear),
In this paper, we provide three main contributions.
The first contribution is the
algorithm \fastlinear (\fastlinearshort, Section \ref{section:fastlinear})
that achieves with probability $1 - 1/n$ a constant factor $(4 + O(\epsi))^{-1}$
in expected linear query complexity and with $O( \log n )$
adaptive rounds (Theorem \ref{theorem:fastlinear}).
  %This is the first highly parallelizable
  %algorithm to achieve linear time complexity in expectation
  %and a constant ratio with high probability.
Although the ratio of $\approx 0.25$ is smaller than the optimal $1 - 1/e \approx 0.63$, this algorithm can be used to improve the 
query complexity of
many extant algorithms, as we decribe in the related work section below.
%Appendix \ref{apdix:fastlinear}.
  Interestingly, \linearseq can be modified to have
  adaptivity $\oh{ \log( n / k ) }$ at a small cost to its
  ratio as discussed in
  Appendix \ref{apx:ls-low-ada}.
  This version of \linearseq is a constant-factor algorithm for \sm with
  smaller adaptivity than any previous algorithm in the literature,
  especially for values of $k$ that are large relative to $n$.

  Our second contribution is an improved
  parallelizable thresholding procedure \threshold (\ts, Section \ref{section:boost}) for a commonly recurring task in submodular optimization:
  namely, add all elements that have
  a gain of a specified threshold $\tau$ to the solution.
  This subproblem arises not only in \sm, but also
  \eg in submodular cover \citep{Fahrbach2018} and
  non-monotone submodular maximization 
  \citep{DBLP:conf/nips/BalkanskiBS18,Fahrbach2018a,Ene2020,Kuhnle2020b}.
  Our \ts accomplishes this task
  with probability $1 - 1/n$
  in $O( \log n)$ adaptive rounds and
  expected $O(n)$ query complexity
  (Theorem \ref{theorem:threshold}),
  while previous procedures
  for this task
  only add elements with an expected gain of $\tau$ and use
  expensive sampling techniques \citep{Fahrbach2018};
  have $\Omega( \log^2 n )$ adaptivity \citep{Kazemi2019}; or have 
  $\Omega( kn )$ query complexity \citep{Balkanski2018c}.
  % More discussion of the relationship between \threshold
  % and related algorithms is given in the Related Work discussion below. %in Appendix \ref{apdix:threshold}.
  
  Finally, we present in Section \ref{section:boost} the
  parallelized greedy algorithm \boost (\boostshort), which is
  used in conjunction with \linearseq and \threshold to
  yield the final algorithm \flsabr, which answers the above question
  affirmatively:
  \flsabr
  obtains nearly the
  optimal $1 - 1/e$ ratio with probability $1 - 2/n$ in $O( \log n )$ adaptive rounds and $O(n)$
  queries in expectation; moreover, \flsabr is faster than \fast
  in an extensive empirical evaluation (see Table \ref{table:cmp-exp}).
  In addition, \flsabr improves theoretically on the previous algorithms
  in query complexity while obtaining nearly optimal adaptivity (see Table \ref{table:cmp}).
  %In addition, \flsabr
  %Moreover, the final algorithm is simple to understand and
  %implement and does not require independent repetitions for multiple
  %guesses of \opt. 
  % This algorithm does not require
  %independent repetitions with multiple guesses
  %of \opt, in contrast to \fast and most extant
  %highly parallelizable algorithms.
  % empirically outperforms \fast across
  % a broad variety of applications and instance sizes as summarized
  % in Table \ref{table:cmp-exp} and described in Section \ref{sec:exp}.
  % Our algorithm
  % obtains superiority in objective value, runtime, and total queries.
  %Since \fast was demonstrated
  %in \citet{Breuer2019} to outperform all previous parallelizable
  %algorithms, our \flsabr is the fastest available algorithm
  %for submodular maximization.
  % 
  % We compare 
  % against the implementation of \fast provided by \citet{Breuer2019} with the same parameter settings
  % as evaluated therein.
  % Our algorithm is implemented wihin the same Python codebase 
  % as the implementation of \fast to ensure the fairness of the comparison.

  \subsection{Additional Related Work}
  \textbf{Adaptive Sequencing.}
  The main inefficiency of the adaptive sequencing method of \citet{Balkanski2018c} (which causes
  the quadratic query complexity) is an explicit check that a constant
  fraction of elements will be filtered from the ground set.
  In this work, we adopt a similar sampling technique
  to adaptive sequencing, except that we design the algorithm to filter a
  constant fraction
  of elements with only
    constant probability.
  This method allows us to reduce the quadratic
  query complexity of adaptive sequencing to linear query complexity
  while only increasing the adaptive complexity by a small constant
  factor. In contrast, \fast of \citet{Breuer2019} speeds up adaptive sequencing by increasing the adaptive complexity
  of the algorithm through adaptive binary search procedures,
  which, in addition to the increasing the adaptivity by logarithmic
  factors, place restrictions
  on the $k$ values for which the ratio can hold.
  This improved adaptive sequencing technique is the core
  of our \threshold procedure, which has the additional benefit
  of being relatively simple
  to analyze.

  \textbf{Algorithms with Linear Query Complexity.}
  Our \linearseq algorithm also uses
  the improved adaptive sequencing technique,
  but in addition this algorithm integrates
  ideas from the $\Omega (n)$-adaptive linear-time streaming algorithm
  of \citet{Kuhnle2020a} to achieve a constant-factor algorithm
  with low adaptivity in expected linear time. Integration of
  the improved adaptive sequencing with the ideas of \citet{Kuhnle2020a}
  is non-trivial, and ultimately this integration enables the
  theoretical improvement in query complexity over previous algorithms
  with sublinear adaptivity that obtain a constant ratio with high probability (see Table \ref{table:cmp}).
%  \textbf{Linear-Time Low-Adaptive Algorithms.}
In \citet{Fahrbach2018}, 
a linear-time procedure \textsc{SubsamplePreprocessing} is described; 
this procedure is to the best of our knowledge the only 
algorithm in the literature that obtains a constant ratio
with sublinear adaptive rounds and linear query complexity
and hence is comparable to \linearseq.
However, \textsc{SubsamplePreprocessing} uses entirely different
ideas from our \linearseq
and has much weaker theoretical guarantees: for input $0< \delta < 1$,
it obtains ratio $\frac{\delta^2}{2 \times 10^6}$
with probability $1 - \delta$
in $O(\log (n) / \delta)$ adaptive rounds and $O(n)$ queries in expectation --
the small ratio renders \textsc{SubsamplePreprocessing} impractical;
also, its ratio holds only with constant probability.
By contrast, with $\epsi = 0.1$, our \linearseq obtains ratio $\approx 0.196$
with probability $1 - 1/n$ in $O(\log(n))$ adaptive rounds and $O(n)$
queries in expectation.
% As discussed above, this large improvement comes from improved
% adaptive sequencing ideas combined with ideas to obtain a constant
% factor with linear query complexity from \citet{Kuhnle2020a}.

\textbf{Using \ls for Preprocessing: Guesses of \opt.}
Many algorithms for \sm, including \fast and all of the
algorithms listed in Table \ref{table:cmp} except for \textsc{SM}
and our algorithm, use a strategy of guessing logarithmically
many values of \opt.
% By submodularity, the optimal solution value \opt must lie in the
% interval $I = [f(a_{max}), kf(a_{max})]$, where $f(a_{max})$ is the maximum
% value of a singleton. Therefore, the algorithm is designed assuming
% the value \opt is known. This knowledge is then relaxed by guessing
% $O( \log k / \epsi )$ values of \opt in the interval $I$.
% While this strategy only adds a logarithmic factor to the total queries
% and one adaptive round (as each guess for \opt may be attempted in parallel),
% it is undesirable in practice as the algorithm is repeated
% for each \opt guess.
Our \linearseq algorithm
reduces the interval containing \opt from size $k$ to a small constant size in expected
linear time. Thus, \linearseq could be used for  preprocessing prior to running \fast or one of the other algorithms in Table \ref{table:cmp}, which
would improve their query complexity without compromising their adaptive complexity or ratio; this illustrates the general utility of \linearseq.
%However, this change does not result in any of the algorithms of Table
%\ref{table:cmp} becoming practical.
For example, with this change,
the theoretical adaptivity of \fast improves, although
it remains worse than \flsabr: the adaptive complexity of \fast becomes
$\oh{\frac{1}{\epsi^2} \log(n) \log \left( \frac{1}{\epsi} \log(k)\right)}$
in contrast to the $O \left( \frac{1}{\epsi^2} \log \left( n / \epsi \right) \right)$ of \flsabr. Although \textsc{SubsamplePreprocessing} may be
used for the same purpose, its ratio only holds with constant probability which would then limit the probability of success of any following algorithm.
% Moreover, 
% the practical performance
% of \fast does not typically improve on the instances in our empirical evaluation by using \ls
% for the following reason. The main practical weakness of \fast is multiple independent
% repetitions of the algorithm 
% are required for multiple guesses of \opt; although fewer guesses would be needed
% with this preprocessing, multiple repetitions of \fast are still required
% in addition to the preprocessing run of \ls. 
% The \fast algorithm is engineered to
% reduce the number of \opt guesses required; for example,
% a sequential binary search of the guesses in the interval $I$
% is conducted, which increases the adaptivity of the algorithm.
% Despite these optimizations, on many instances in our empirical
% evaluation (Section \ref{sec:exp}), \fast repeats itself more than
% $10$ times with different guesses of \opt.

\textbf{Relationship of \threshold to Existing Methods.}
%Our \threshold is the main component of our \boost algorithm. 
% As discussed in more detail below (Section \ref{section:threshold}), \threshold is a parallelizable
% algorithm that greedily adds all elements
% with gain above an input threshold $\tau$.
% This task is an important component of many algorithms
% for \sm, and thus \threshold is not the first algorithm
% for this purpose.
The first procedure in the literature to perform the same task
is the \textsc{ThresholdSampling} procedure of \citet{Fahrbach2018};
however, \textsc{ThresholdSampling} only ensures that the \textit{expected} marginal
gain of each element added is at least $\tau$ and has large constants
in its runtime that make it impractical \citep{Breuer2019}.
In contrast, \threshold ensures that added elements contribute
a gain of at least $\tau$ \textit{with high probability} and is highly
efficient empirically.
A second procedure in the literature to perform the same task is
the \textsc{Adaptive-Sequencing}
method of \citet{Balkanski2018c}, which similarly to
\threshold uses random permutations
of the ground set; however,
\textsc{Adaptive-Sequencing} focuses on explicitly
ensuring a constant fraction of elements will be filtered in the next
round, which is expensive to check: the query complexity of
\textsc{Adaptive-Sequencing} is $O(kn)$. In contrast, our \threshold
algorithm ensures this property with a constant probability, which is sufficient to ensure the adaptivity with the high probability of $1 - 1/n$
in $O(n)$ expected queries.
Finally, a third related procedure in the literature is
\textsc{ThresholdSampling} of \citet{Kazemi2019}, which also
uses random permutations to sample elements. 
However, this algorithm has the higher adaptivity
of $\oh{ \log (n) \log( k ) }$,
in contrast to the $\oh{ \log (n) }$ of \threshold.

\textbf{MapReduce Framework.}
% The study of constrained submodular maximization from
% the adaptive complexity standpoint was initiated recently by
% \citet{Balkanski2018}; since then, a large amount of
% work studying adaptive complexity in this context
% has appeared; a non-exhaustive list includes
% the works discussed above,
% works studying non-monotone submodular maximization \citep{Balkanskia,Fahrbach2018a,Ene2019,Kuhnle2020b},
% maximization under matroid and knapsack constraints \citep{Balkanski2018c,Chekuri2018a}, among others.
% Previously,
Another line of work studying parallelizable algorithms for \sm
has focused on the MapReduce framework \citep{Dean2008}
in a distributed setting, \eg
\citep{Barbosa2015,Barbosa2016,Epasto2017,Mirrokni2015}.
These algorithms divide the dataset over a large number
of machines and are intended for a setting in which
the data does not fit on a single machine.
None of these algorithms has sublinear adaptivity and hence
all have potentially large numbers of
sequential function queries on each machine. 
In this work, our empirical evaluation is
on a single machine with a large number of CPU cores; % , which yields
% a relatively low cost for communication among parallel threads of execution
we do not evaluate our algorithms in a distributed setting.

% The overall strategy of \linearseq to obtain its ratio is discussed
% in Section \ref{section:fastlinear}. This strategy is 
% forms the basis of the $\Omega(n)$-adaptive streaming algorithm
% of \citet{Kuhnle2020a}.

% \textbf{Limitations of \fast.}
% therefore, for the ratio to hold with high probability\footnote{In this work, high probability means at least probability $1 - \oh{1/n}$ for ground set of size $n$.}, we must
% have $k \ge O( \log(n) )$ or \fast must be repeated independently
% many times. Moreover, even if the failure probability $\delta$
% is set arbitrarily close to $1$, \textsc{FAST} cannot obtain any ratio
% for $k < 850$.
% Next, the adaptivity of \fast is 
% $O( \log(n) \log^2( \log k ) )$, which is larger
% than the state-of-the-art adaptivity $O( \log(n) )$.
% Finally, although \fast has been shown to be faster
% than previous algorithms empirically,
% \fast must sequentially repeat
% itself for different guesses of
% the optimal solution value \opt. 

\textbf{Organization.} The constant-factor algorithm \linearseq
is described and analyzed in Section \ref{section:fastlinear};
% this algorithm is the most novel and significant of our contributions
% so substantial discussion of the intuition and analysis is
% included in the main text although
the details of the
analysis are presented in Appendix \ref{apdix:fastlinear}.
The variant of \linearseq with lower adaptivity is described in
Appendix \ref{apx:ls-low-ada}.
The algorithms \threshold and \boost
are discussed at a high level in Section \ref{section:boost},
with detailed descriptions of these algorithms and theoretical
analysis presented in Appendices \ref{apdix:threshold} and
\ref{apdix:boost}. Our empirical
evaluation is summarized in Section \ref{sec:exp}
with more results and discussion
in Appendix \ref{apx:exp}.
% \textbf{Lower Bounds} \citet{Li2020}
%%% Local Variables:
%%% mode: latex
%%% TeX-master: "main"
%%% End:

\section{A Parallelizable Algorithm with Linear Query Complexity: \linearseq} \label{section:fastlinear}
\begin{algorithm}[t]
  \caption{The algorithm that obtains ratio $(4 + \oh{ \epsi )}^{-1}$ in
  $\oh{\log(n) / \epsi^3}$ adaptive rounds and expected $\oh{ n / \epsi^3 }$ queries.}
  \label{alg:fastlinear}
  \begin{algorithmic}[1]
  \Procedure{\fastlinear}{$f, \mathcal N, k, \epsi$}

  \State \textbf{Input:} evaluation oracle $f:2^{\mathcal N} 
  \to \reals$, constraint $k$, error $\epsi$

  \State $a = \arg \max_{u \in \mathcal{N}} f(\{u\})$

  \State Initialize $A \gets \{a\}$ , $V \gets \mathcal N$, 
  $\ell = \lceil 4(1+1/(\beta \epsi))\log(n) \rceil$, 
  \textcolor{\rev}{$\beta = \epsi/(24\log(8/(1-e^{-\epsi/2})))$}\label{line:beta}

  \For{ $j \gets 1$ to $\ell$ } \label{line:fastOuterForStart}

        \State Update $V \gets \{ x \in V : \marge{x}{A} \ge 
        f(A) / k \}$ and filter out the rest \label{line:fastFilterV} 

        \State \textbf{if} $|V| = 0$ \textbf{then break} \label{line:fastIf1}

      \State $V = \{v_1, v_2, \ldots, v_{|V|}\} \gets $\textbf{random-permutation}$(V)$ \label{line:fastPermutation}
      \color{\rev}
      \State $\Lamb \gets\left[\left\lceil \frac{1}{\epsi} \right\rceil\right] \cup$
      \color{black}
      $ \{\lfloor(1+\epsi)^u \rfloor: 
      1 \leq \lfloor(1+\epsi)^u \rfloor \leq k, 
      u \in \mathbb{N}\} $
      \Statex $\qquad \qquad \qquad \cup 
      \{\lfloor k+u\epsi k \rfloor: \lfloor k+u\epsi k \rfloor \leq 
      |V|, u \in \mathbb{N}\} \cup \{|V|\}$

      \State $B[\lamb_i] = \textbf{false}$, for $\lamb_i \in \Lambda$

      \For{$\lamb_i \in \Lamb$ in parallel }\label{line:fastInnerForStart}
            %
            %\If{$\lamb_i = 1$}
            %\State $T_{\lamb_{i-1}} \gets \emptyset$
            %\Else
            %
            \State $T_{\lamb_{i-1}} \gets \{v_1, v_2, \ldots, v_{\lamb_{i-1}}\}$
      % \EndIf
            ; $T_{\lamb_i} \gets \{v_1, v_2, \ldots, v_{\lamb_i}\}$ ;
            $T_{\lamb_i}' \gets T_{\lamb_i} \backslash T_{\lamb_{i-1}}$

            \State \textbf{if} $\marge{T_{\lamb_i}'}{A \cup T_{\lamb_{i-1}}}/|T_{\lamb_i}'| \geq 
          (1-\epsi) f(A \cup T_{\lamb_{i-1}})/k$ \textbf{then} $B[\lamb_i] \gets \textbf{true}$           \label{line:fastIf2}

      \EndFor \label{line:fastinnerForEnd}
      \color{\rev}
      \State $\lamb^* \gets \max \left\{\lamb_i \in \Lamb: 
      \left(\lamb_i \le \left\lceil \frac{1}{\epsi} \right \rceil \text{ and } B[1]\text{ to }B[\lamb_i]\text{ are all true }\right)\right.$
      \Statex \hspace*{8em} or $\left(\left\lceil \frac{1}{\epsi} \right \rceil<\lamb_i \leq k \text{ and } B[\lamb_i] = \textbf{false} \text{ and } B[1] \text{ to } B[\lamb_{i-1}] 
                  \text{ are all } \textbf{true}\right)$
      \Statex \hspace*{8em}or $\left(\lamb_i > k \text{ and } B[\lamb_i] = \textbf{false} \text{ and } \exists m \geq 1 \text{ s.t. } 
            |\bigcup_{u=m}^{i-1} T_{\lamb_u}'| \geq k \right.$
            and 
            $\left.\left.B[\lamb_m] \text{ to } B[\lamb_{i-1}] \text{ are all } \textbf{true}\right) \right\}$
      \label{line:rule}
      % \State \textbf{if} $\lamb^* \le \left\lceil \frac{1}{\epsi} \right \rceil$ \textbf{then} $\lamb^*\gets \lamb^*-1$\label{line:rule2}
      \color{black}
      \State $A \gets A \cup T_{\lamb^*}$ 

    \EndFor

    \State \textbf{if} { $|V| > 0$ } \textbf{then return} \textit{failure}
    \State \textbf{return} $A' \gets$ last $k$ elements added to $A$

    \EndProcedure
\end{algorithmic}
\end{algorithm}
In this section, we describe the  algorithm \linearseq for \sm
(Alg. \ref{alg:fastlinear}) that 
obtains ratio $(4 + \oh{\epsi})^{-1}$ in 
$\oh{\frac{1}{\epsi^3} \log (n) }$ adaptive rounds
and expected $\oh{ \frac{n}{\epsi^3 } }$ queries. 
If $\epsi \le 0.21$, the 
ratio of \linearseq is lower-bounded by \textcolor{\rev}{$(4 + 25\epsi)^{-1} \ge 0.108$},
which shows that a relatively large constant ratio is obtained
even at large values of $\epsi$.
An initial run of this algorithm is required
for our main algorithm \flsabr. 

\textbf{Description of \ls.} The work of \ls is done within iterations of
a sequential outer \textbf{for} loop (Line \ref{line:fastOuterForStart}); this loop
iterates at most $\oh{\log(n)}$ times, and each
iteration requires two adaptive rounds; thus, the adaptive complexity
of the algorithm is $\oh{\log(n)}$. 
Each iteration adds more elements to the set $A$, which is initially empty. 
Within each iteration, there
are four high-level steps: 1) filter elements from $V$ that have gain less than $f(A) / k$ (Line \ref{line:fastFilterV}); 2) randomly permute $V$ (Line \ref{line:fastPermutation}); 3) compute in parallel the marginal gain of adding blocks of the sequence of remaining elements in $V$ to $A$ (\textbf{for} loop on Line \ref{line:fastInnerForStart}); 4) select a prefix of the sequence $V = \left( v_1,v_2,\ldots,v_{|V|} \right)$ to add to $A$ (Line \ref{line:rule}). The selection of the prefix to add is carefully chosen to approximately
satisfy, on average, Condition \ref{ls:cond} for elements
added; and also to
ensure that, with constant probability,
a constant fraction of elements of $V$ are filtered on the
next iteration. 
% The two adaptive rounds are first the filtering, which is arbitrarily parallelizable, and then the inner \textbf{for} loop. 

The following theorem states the theoretical results for \linearseq.
The remainder of this section proves this theorem, with intuition
and discussion of the proof.
The omitted proofs for all lemmata are provided
in Appendix \ref{apdix:fastlinear}.
\begin{theorem} \label{theorem:fastlinear}
Let $(f,k)$ be an instance of \sm.
For any constant $0<\epsi< 1/2$,
the algorithm $\fastlinear$
has adaptive complexity $\oh{\log(n)/\epsi^3}$ and
outputs $A' \subseteq \mathcal N$ with $|A'| \le k$ such that the following properties hold:
1) The algorithm succeeds with probability at least $1 - 1/n$.
2) There are $O\left( (1/(\epsi k)+1)n /\epsi^3\right)$ 
oracle queries in expectation.
3) If the algorithm succeeds, \textcolor{\rev}{$\left[ 4+\frac{2(5-4\epsi)}{(1-2\epsi)^2}\cdot\epsi \right]f(A') \ge f(O)$}, where $O$ is an optimal
solution to the instance $(f,k)$.
\end{theorem}
\textbf{Overview.} The goal of this section is
to produce a constant factor, parallelizable algorithm
with linear query complexity. As a starting point,
consider an algorithm\footnote{This algorithm is a simplified version of the streaming algorithm of \citet{Kuhnle2020a}.} that takes one pass through
the ground set, adding each element $e$ to candidate set $A$
iff
\begin{equation}
  \marge{e}{A} \ge f(A) / k. \label{ls:cond}
\end{equation}
Condition \ref{ls:cond} ensures two properties:
1) the last $k$ elements in $A$ contain a constant fraction of
the value $f(A)$; and 2) $f(A)$ is within a constant fraction
of \opt. By these two properties, the last $k$ elements of $A$
are a constant factor approximation to \sm with exactly one
query of the objective function per element of the ground set.
% To achieve a constant factor in linear time,
% the following approach may be used:
% produce a set 
% $A$, which satisfies the following two properties:
% \begin{enumerate}
% \item[P1.] The element $a_{i}$ added to $A_i = \{a_1, a_2, \ldots, a_{i-1}\}$
% contributes a marginal gain of at least $f(A_i)/k$; that is,
% $\marge{a_{i}}{A_i} \ge f(A_i)/k$ for all $i \in [|A|]$.
% \item[P2.]  All $x \not \in A$ satisfy $\marge{x}{A} < \ff{A}/k$.
% \end{enumerate}
% If both of these conditions are satisfied, then the last $k$ elements added
% to $A$ is a $0.25$ approximation to \sm$(f,k)$.
% These properties may be obtained in linear time in a 
% single pass loop through the ground set; a fact
% which was used for a linear-time streaming algorithm
% in \citet{Kuhnle2020a}.
For completeness, we give
a pseudocode (Alg. \ref{alg:adaptive-linear})
and proof in Appendix \ref{sec:adaptive-linear}.
However, each query depends on
all of the previous ones and thus there are $n$
adaptive rounds.
Therefore, the challenge is to approximately simulate Alg. 
\ref{alg:adaptive-linear} in
a lowly adaptive (highly parallelizable) manner, which is what \linearseq accomplishes.
\subsection{Approximately Satisfying Condition \ref{ls:cond}} \label{sec:linearseq-add}
\textbf{Discarding Elements.} In one adaptive round during each iteration $j$ of
the outer \textbf{for} loop, 
all elements with gain to $A$ of 
less than $\ff{A}/k$ are discarded from $V$ (Line \ref{line:fastFilterV}).
Since the size of $A$ increases as the algorithm runs,
by submodularity,
the gain of these elements can only decrease and hence these elements
cannot satisfy Condition \ref{ls:cond}
and can be safely discarded from consideration. The process of filtering
thereby ensures the following lemma at termination.
\begin{restatable}{lemma}{fastSubOne}
\label{lemma:fastSubOne}
At successful termination of \linearseq, 
$f(O) \leq 2f(A)$, where $O \subseteq \mathcal{N}$ is an optimal 
solution of size $k$.
\end{restatable}  

\textbf{Addition of Elements.} Next, we describe the details of how elements are added to the set $A$. 
The random permutation of remaining elements on Line \ref{line:fastPermutation}
constructs a sequence $\left( v_1,v_2,\ldots,v_{|V|} \right)$ such that each element
is uniformly randomly sampled from the remaining elements. 
By testing the marginal gains along the sequence
%(\ie whether $\marge{v_i}{ A \cup V_i } \ge \ff{ A \cup V_i } / k$)
in parallel, it is possible to determine a good prefix of the sequence $(v_i)$ to add to $A$
to ensure the following: 1) Condition \ref{ls:cond}
is approximately satisfied; and 2)
We will discard a constant fraction of $V$ in the next iteration with constant probability.
Condition \ref{ls:cond} is important for the approximation ratio and discarding a constant fraction
of $V$ is important for the adaptivity and query complexity. Below, we discuss how to choose 
the prefix such that both are achieved. 
% \begin{itemize}
%   \item At least $(1 - \epsi)$ fraction of the last (up to) $k$ elements of the prefix 
%     selected satisfy the desired condition on the marginal gain:
%     $\marge{v_i}{ A \cup \{v_1, \ldots v_{i-1} \} } \ge \ff{ A \cup V_i } / k$.
%     This ensures an approximate form of Property P1 above holds. 
%   \item With constant probability, we have added enough elements to ensure that a constant fraction of elements will be discarded at the next iteration of the outer \textbf{for} loop. Discarding a constant fraction of elements with constant probability at each iteration ensures the expected linear runtime and adaptivity of the algorithm with high probability. 
% \end{itemize}
To speed up the  algorithm, we do not test the marginal gain at each point in
the sequence $(v_i)$, but rather test blocks of elements at once as determined by
the index set $\Lambda$ defined in the pseudocode.
% That is, we break the sequence
% $(v_i)$ into blocks $(T_\lamb', \lamb \in \Lamb)$ and test the average marginal gain of a block,
% \ie whether 
% \begin{equation}
% \marge{T_\lamb'}{A \cup \bigcup_{p < \lamb, p \in \Lamb} T_p'} \ge |T_\lamb'| (1 - \epsi) \ff{ A \cup \bigcup_{p < \lamb, p \in \Lamb} T_p' } / k.\tag{*}
% \end{equation}
% For the first $k$ elements of the sequence $(v_i)$, 
% the block size exponentially increases to $\epsi k$, and
% after the first $k$ elements of the sequence, the block size is set to $\epsi k$. 

\textbf{Prefix Selection.} Say a block is \textit{bad} if this block does not
satisfy the condition checked on Line \ref{line:fastIf2} (which is an approximate, average form of Condition \ref{ls:cond}); otherwise,
the block is \textit{good}. At the end of an iteration, 
we select the largest block index $\lamb^*$, 
where \textcolor{\rev}{$\lamb^*\le \left\lceil \frac{1}{\epsi}\right\rceil$ and all the blocks are good up to and including $\lamb^*$; or 
$\left\lceil \frac{1}{\epsi}\right\rceil < \lamb^* \le k$,
this block is bad, and
all the previous blocks are good blocks;
or this block is bad and the previous consecutive 
blocks which together have at least $k$ elements are all good.}
Then, we add the prefix $T_{\lamb^*} = (v_1, v_2, \ldots, v_{\lamb^*})$ into $A$.
Now, the relevance of Condition \ref{ls:cond} for the approximation
ratio is that it implies $f(A) \ge 2f(A \setminus A')$,
where $A'$ are the last $k$ elements added to $A$. Lemma
\ref{lemma:fastSubTwo} shows that
the conditions required on the marginal gains of blocks added
imply an approximate form of this fact is satisfied by \linearseq.
Indeed,
the proof of Lemma \ref{lemma:fastSubTwo} informs
the choice $\Lambda$ of blocks evaluated and the computation
of $\lambda^*$.
% If the blocks we add to $A$ are all \textit{good} blocks,
% we would have $\marge{A'}{A\backslash A'}\ge(1-\epsi)f(A\backslash A')$.
% The marginal gain with respect to the last $k$ elements almost 
% equals to the objective value of the set without the last $k$ elements.
% In the current circumstances, there are some \textit{bad} blocks in $A$.
% The following Lemma \ref{lemma:fastSubTwo} shows that the marginal
% gain with respect to the last $k$ elements still approaches 
% $f(A\backslash A')$.
\begin{restatable}{lemma}{fastSubTwo}
\label{lemma:fastSubTwo}
Suppose \linearseq terminates successfully. Then
$f(A) \geq \textcolor{\rev}{\left(1+\frac{(1-2\epsi)^2}{1+\epsi}\right)} f(A\backslash A')$.
\end{restatable}
\begin{proof}
\begin{figure}[t]
\centering
\includegraphics[width=0.7\textwidth]{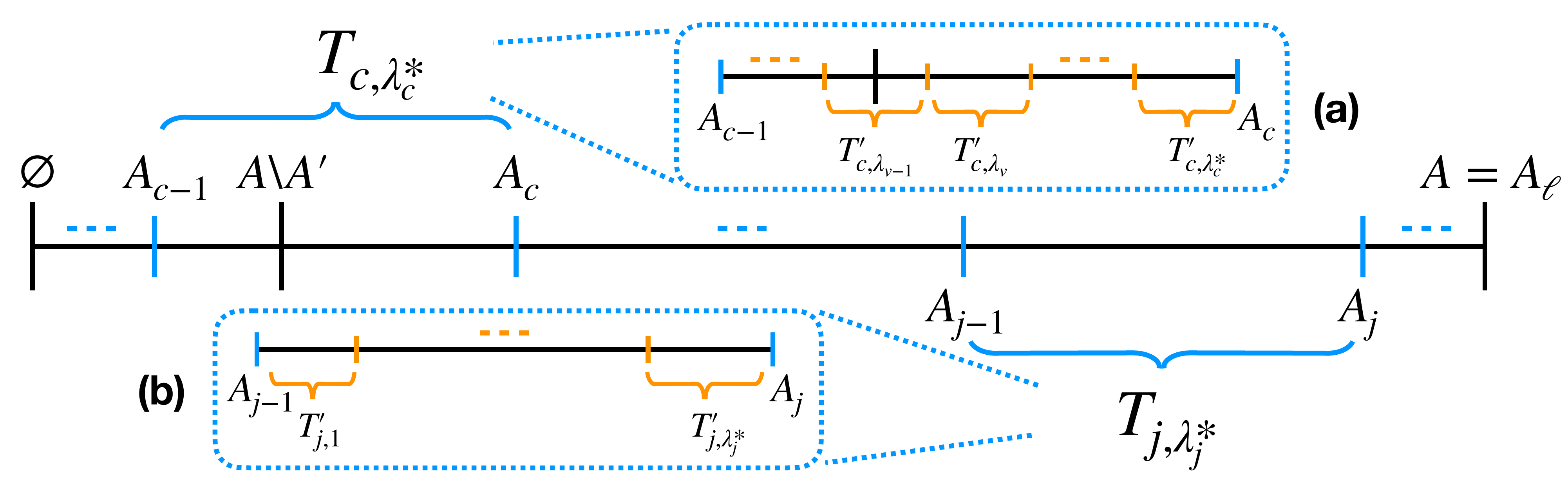}
\caption{This figure depicts the way elements are added to $A$, with
  $A_i$ denoting the value of $A$ after iteration $i$ 
  of the outer \textbf{for} loop; the prefix added at iteration
  $i$ is denoted by $T_{i,\lambda_i^*}$.
  The orange lines mark the blocks comprising each prefix. Also,
  the set $A'$ of the last $k$ elements added to $A$ is shown. Parts
  \textbf{(a)}
  and \textbf{(b)} depict two separate cases in the proof of Claim \ref{claim:LSblockOne}.}
\label{fig:block}
\end{figure}
If $|A| \le k$, the lemma is immediate, so assume
$|A| > k$. For iteration $j$ \color{\rev} of the outer for loop on Line~\ref{line:fastOuterForStart}, 
\color{black}
let $T_{j,\lamb_j^*}$
denote the set added to $A$ during iteration $j$;
and let
$T_{j,\lamb_j^*}=\emptyset$ if the algorithm terminates before
iteration $j$. Let $A_j$ denote the value of $A$ after iteration $j$.
Define $c = \max \{c \in \mathbb{N}: A' \subseteq
(\cup_{j=c}^{\ell} T_{j,\lamb_j^*} ) \}$. 
Then, $|T_{c,\lamb_c^*}| > 0$; and for any 
$j > c$, $|T_{j,\lamb_j^*}| < k$.
It holds that $(\cup_{j=c+1}^{\ell} T_{j,\lamb_j^*}) 
\subset A' \subseteq (\cup_{j=c}^{\ell} T_{j,\lamb_j^*})$.
Figure \ref{fig:block} shows how $A$ is composed of these
sets $T_{j,\lamb_j^*}$ and how each set is composed of blocks.
The following claim is proven in Appendix \ref{apx:claim1}.
\begin{restatable}{claim}{LSblockOne}
  \label{claim:LSblockOne}
  \color{\rev}
  Let $T_{j,\lamb_j^*}$ denote the set added to $A$ during iteration $j$ of the outer for loop on Line~\ref{line:fastOuterForStart},
  $A_j$ denote the value of $A$ after iteration $j$.
  Define $c = \max \{c \in \mathbb{N}: A' \subseteq
  (\cup_{j=c}^{\ell} T_{j,\lamb_j^*} ) \}$. 
  \color{black}
  It holds that
  $\marge{T_{c,\lamb_c^*}}{A\backslash A'}\ge 
(1-\epsi)\max\{0,|T_{c,\lamb_c^*}\cap A'|-2\epsi k \}\cdot f(A\backslash A')/k.$ For $j>c$, it holds that
$\marge{T_{j,\lamb_j^*}}{A_{j-1}}\ge 
\textcolor{\rev}{\frac{1-2\epsi}{1+\epsi}}|T_{j,\lamb_j^*}|\cdot f(A\backslash A')/k.$
\end{restatable}
From Claim \ref{claim:LSblockOne},
\begin{equation}
  f(A)-f(A \backslash A') 
= \marge{T_{c,\lamb_c^*}}{A\backslash A'} + 
\sum_{j=c+1}^{\ell} \marge{T_{j,\lamb_j^*}}{A_{j-1}} 
\geq \textcolor{\rev}{\frac{(1-2\epsi)^2}{1+\epsi}}\cdot f(A\backslash A') \label{ineq:claim1},
\end{equation}
where Inequality \ref{ineq:claim1} is proven
in Appendix \ref{apx:ineq-claim1}.
\end{proof}
% \begin{proof} 

% Let's say, at the end of an iteration $j$, a set 
% $T_{j,\lamb_j^*}$ is added to $A$. 
% $T_{j,\lamb_j^*}$ could contain several blocks with only one 
% \textit{bad} block.
% For any $j$ that $|T_{j,\lamb_j^*} \cap A'|>0$, there are two cases.

% First, $T_{j,\lamb_j^*}$ may be a subset of $A'$. In this case, 
% $|T_{j,\lamb_j^*}| \le k$. There are at most $\epsi/(1+\epsi)$ of 
% elements in the \textit{bad} block.
% Second, $T_{j,\lamb_j^*}$ isn't a subset of $A'$. Then, the number of 
% elements in $T_{j,\lamb_j^*}$ is uncertain. 
% For the set $T_{j,\lamb_j^*} \cap A'$, there is a \textit{bad} block
% with no more than $\epsi k$ elements and an incomplete block with no
% more than $\epsi k$ elements.

% Totally, there are at most $3\epsi k/(1+\epsi)$
% elements in the \textit{bad} block or incomplete block.

% \end{proof}
%\textbf{Probably Filtering out a Constant Fraction of $V$ with the Selection of $\lamb^*$.} 
% To successfully filter out a constant fraction of $V$ with high
% probability, $T_{\lambda^*}$ is chosen to be a 
% union of consecutive good blocks and also one bad
% block at the end of every iteration.
In the remainder of this subsection, 
we will show that a $(\beta \epsi)$-fraction of $V$ is discarded at each iteration $j$ of the outer \textbf{for} loop with
probability at least $1/2$, 
where $\beta$ is a constant in terms of $\epsi$ as defined
on Line \ref{line:beta}
in the pseudocode. The remainder of the proofs in this
section are implicitly conditioned on the behavior of the
algorithm prior to iteration $j$.
The next lemma describes the behavior of the
number of elements that will be filtered at iteration $j + 1$.
Observe that the set $S_i$
defined in the next lemma
is the set of elements that would be filtered at the
next iteration if prefix $T_i$ is added to $A$.
\begin{restatable}{lemma}{FLSnumBad}
  \label{lemma:FLSnumBad}
  Let $S_i = \left\{x \in V: \marge{x}{A \cup T_i} <f(A\cup T_i)/k\right\}.$
It holds that $|S_0|= 0$, $|S_{|V|}| = |V|$, 
and $|S_i| \leq |S_{i+1}|$.
\end{restatable}
By Lemma \ref{lemma:FLSnumBad}, we know the number of elements in 
$S_i$ increases from 0 to $|V|$ with $i$. 
Therefore, there exists a $t$ such that $t = 
\min \{i \in \mathbb{N}: |S_i| \geq \beta \epsi |V|\}$. 
If $\lamb^* \geq t$, $|S_{\lamb^*}| \geq \beta \epsi|V|$,
and we will successfully filter out more than $(\beta \epsi)$-fraction 
of $V$ at the next iteration. In this case, we say that the iteration
$j$ \emph{succeeds}.
Otherwise, if $\lamb^* < t$, the iteration may fail. The
remainder of the proof bounds the probability
that $\lamb^* < t$, which is an upper bound on
the probability that iteration $j$ fails.
Let 
% \color{\rev}
% $\lamb_t = t$ if $t\le \left\lceil \frac{1}{\epsi}\right \rceil$ or
% \color{black}
$\lamb_{t} = \max\{\lamb \in \Lambda: \lamb \textcolor{\rev}{\le} t \}$, and let 
$\lamb_{t}' = \max ( \{\lamb' \in \Lambda:
\sum_{\lamb \in \Lambda, \lamb' \leq \lamb \leq \lamb_{t}}
|T'_{\lamb}| \geq k\} \cup \{ 1 \} )$.
If $\lamb^* < t$, there must be at least one index
$\lamb$ between $\lamb_t'$ and $\lamb_t$ such that the
block $T'_\lamb$ is bad.
% The next lemma bounds the probability that any block
% $T_\lamb',$ with \textcolor{\rev}{given} $\lamb  \textcolor{\rev}{\le} \lamb_t$,
% is bad.
% \begin{restatable}{lemma}{oneBadBlock}
%   \label{lemma:oneBadBlock}
%   Let $t = \min \{i \in \mathbb{N}: |S_i| \geq \beta \epsi |V|\}$;
%   $\lamb_t = \max \{\lamb \in \Lamb: \lamb \textcolor{\rev}{\le} t\}$;
%   $(Y_i)$ be a sequence of independent and identically distributed Bernoulli trials,
%   where the success probability is $\beta\epsi$.
%   Then for any \textcolor{\rev}{given} $\lamb \textcolor{\rev}{\le} \lamb_t$,
%   $\prob{B[\lamb]=\textbf{false}}
%   \leq \prob{\sum_{i=1}^{|T'_\lamb|}Y_i > \epsi |T'_\lamb|}$.
% \end{restatable}
% With Lemma \ref{lemma:oneBadBlock}, the probability of a \textit{bad}
% block appears before $t$ has an upper bound.
% If there is a \textit{bad} block between $\lamb''$ and $\lamb'$,
% we may fail to filter out $\beta \epsi$-fraction of $V$ at that iteration.
% For any $\lamb_i$ that $\lamb_{i''} \leq \lamb_i \leq \lamb_{i'}$, $\lamb_i$ can be any 
% one which is less than or equal to $k$. And there should be no more 
% than $1/\epsi$ $\lamb_i$s, that  $\lamb_i>k$. 
% Therefore, the probability of $\lamb^*<t$ can be bounded as follows,
% Finally, we bound the probability that an iteration $j$
% of the outer \textbf{for} loop fails.
Let $B_1 = \{ \lamb \in \Lambda : \lamb \le k \text{ and } \lamb \textcolor{\rev}{\le} \lamb_t \}$, 
$B_2 = \{ \lamb \in \Lambda : | \Lambda \cap [ \lamb, \lamb_t ] | \le \lceil 1/ \epsi \rceil \}$.
\color{\rev}
$B_1\cup B_2$ includes all indices between $\lamb_t'$ and $\lamb_t$.
Then, the failure probability of iteration $j$ can be bounded as below,
\color{black}
\begin{equation}
  \prob{\text{iteration $j$ fails}  }
  \leq \prob{\exists \lambda \in B_1 \cup B_2 \text{ with }B[\lamb] = \textbf{false}}  \le 1/2, \label{ineq:linear-fail-1}
\end{equation}
where the proof of Inequality \ref{ineq:linear-fail-1} is in Appendix 
\ref{apx:ls-fail}.
% By the Law of Total Probability, $Pr( \text{iteration $j$ fails} ) \le 1/2 $ after unfixing the event $\lamb_t = \alpha$ and the behavior of the algorithm prior to iteration $j$. 
%Therefore, with the design of blocks and the selection of $\lamb^*$,
%the probability of successfully filtering out an 
%$\beta \epsi$-fraction of $V$ at any iteration is more than $1/2$.

\subsection{Proof of Theorem \ref{theorem:fastlinear}}
\label{sec:linearseq-proof}
From Section \ref{sec:linearseq-add}, the probability 
at any iteration 
of the outer \textbf{for} loop of successful
filtering of an $(\beta \epsi)$-fraction of $V$ is at least $1/2$. 
We can model the success of the iterations as a sequence of dependent Bernoulli
random variables, with success probability that depends on the results of previous
trials but is always at least $1/2$.
% To select a set with large marginal gain, the easiest way is to campare 
% all the first few elements of $V$ and select the one with most elements. 
% In this case, we need $|V|$ query calls. To reduce the query calls made
% in selection, we set blocks and calculate the marginal gain of each
% block. For the first $k$ elements, the block size is exponentially 
% increasing with $1+\epsi$, where the block size is no more than $\epsi k$.
% After the first $k$ elements, the block size is set to $\epsi k$.

\textbf{Success Probability of \linearseq.} 
If there are at least $m = \lceil \log_{1-\beta\epsi}(1/n) \rceil$ successful iterations,
the algorithm \linearseq will succeed.
The number of successful iterations $X_\ell$ up to
and including the $\ell$-th iteration is a sum of
dependent Bernoulli random variables.
With some work (Lemma \ref{lemma:indep} in Appendix \ref{apx:prob}),
the Chernoff bounds can be applied to ensure the algorithm
succeeds with probability at least $1 - 1/n$, as shown
in Appendix \ref{apx:ls-success}.

\textbf{Adaptivity and Query Complexity.}
Oracle queries are made on Lines
\ref{line:fastFilterV} and \ref{line:fastIf2} of \linearseq.
The filtering on Line \ref{line:fastFilterV} is in one adaptive round,
and the inner \textbf{for} loop 
is also in one adaptive round. % with no more than 
%$|V|/(\epsi k)+\log_{1+\epsi} (k)-1/\epsi+3$ iterations. 
Thus, the adaptivity is
proportional to the number of iterations of the outer
\textbf{for} loop, $\oh{\ell} = \oh{\log(n)/\epsi^3}.$
For the query complexity,
let $Y_i$ be the number of iterations between the $(i-1)$-th
and $i$-th successful iterations of the outer \textbf{for} loop.
By Lemma \ref{lemma:indep} in Appendix \ref{apx:prob}, 
$\ex{Y_i} \le 2$. From here, we show in Appendix \ref{apx:ls-query}
that there are at most $\oh{n / \epsi^3}$ queries in expectation.

\textbf{Approximation Ratio.} Suppose \linearseq terminates
successfully.
We have the approximation ratio as follows:
\color{\rev}
$$f(A') \overset{(a)}{\geq} f(A) - f(A\backslash A') \overset{(b)}{\geq} f(A) - \frac{1+\epsi}{1+\epsi+(1-2\epsi)^2} f(A) \overset{(c)}{\geq} \frac{1}{4+\frac{2(5-4\epsi)}{(1-2\epsi)^2}\cdot\epsi} f(O),$$
\color{black}
where Inequality (a) is from submodularity of $f$, Inequality (b) is
from Lemma \ref{lemma:fastSubTwo}, and Inequality (c) is from
Lemma \ref{lemma:fastSubOne}.
%%% Local Variables:
%%% mode: latex
%%% TeX-master: "main.tex"
%%% End:

\section{Improving to Nearly the Optimal Ratio} \label{section:boost}
In this section, we describe how to obtain the nearly optimal
ratio in nearly optimal query and adaptive complexities (Section
\ref{sec:flsabr}).
First, in Section \ref{section:threshold},
we describe \threshold, a parallelizable procedure to add all
elements with gain above a constant threshold to the solution.
In Section \ref{sec:flsabr}, we describe \boost and
finally the main algorithm \flsabr.
Because of space constraints, the algorithms are described
in the main text at
a high level only, with detailed descriptions and proofs
deferred to Appendices \ref{apdix:threshold} and \ref{apdix:boost}.
\subsection{The \threshold Procedure}
\label{section:threshold}
In this section, we discuss the algorithm
\threshold,
which adds all elements with
gain above an input threshold $\tau$ up to accuracy $\epsi$
in $O(\log n)$ adaptive rounds and $O(n)$ queries in expectation. 
Pseudocode is given in Alg. \ref{alg:threshold} in Appendix
\ref{apdix:threshold}.

\textbf{Overview.}
The goal of this algorithm is, given an input threshold $\tau$
and size constraint $k$,
to produce a set of size at most $k$
such that the average gain of elements added
is at least $\tau$.
As discussed in Section \ref{sec:intro}, this task
is an important subroutine of many algorithms for submodular
optimization (including our final algorithm),
although by itself it does not produce any
approximation ratio for \sm.
% As a starting point,
% consider an algorithm that takes one pass through
% the ground set, adding each element $e$ to candidate set $A$
% iff
% \begin{equation}
%   \marge{e}{A} \ge \tau \text{ and } |A \cup \{ e \}| \le k. \label{ts:cond}
% \end{equation}
% Clearly, such an algorithm requires at most $n$ queries and meets
% the desired requirements; unfortunately, the adaptivity of this
% approach is also $n$ rounds. Therefore, the challenge is to parallelize
% this algorithm.
%Observe that Conditions \ref{ls:cond} and \ref{ts:cond} share
%a common form,
%although the threshold $\tau$ in Condition \ref{ts:cond} is fixed,
%while the threshold in Condition \ref{ls:cond} is the variable $f(A)/k$ which
%depends on $A$.
The overall strategy of our parallelizable algorithm
\threshold is analagous to that of \linearseq, although
\threshold is considerably simpler to analyze.
% The detailed algorithm description and analysis are deferred to
% Appendix \ref{apdix:threshold}.
The following theorem summarizes 
the theoretical guarantees
of \threshold and
the proofs are
in Appendix \ref{apdix:threshold}.
\begin{theorem} \label{theorem:threshold}
Suppose \threshold is run with input $(f,k, \epsi, \delta, \tau )$.
Then, the algorithm 
has adaptive complexity $O(\log (n/\delta)/\epsi)$ and
outputs $A \subseteq \mathcal N$ with $|A| \le k$ such that the following properties hold:
1) The algorithm succeeds with probability at least $1 - \delta/n$.
2) There are $O(n/\epsi)$ oracle queries in expectation.
\color{\rev}
3) It holds that $f(A)/|A| \geq (1-2\epsi)\tau / (1+\epsi)$.
\color{black}
4) If $|A| < k$, then $\marge{x}{A} < \tau$ for all $x \in \mathcal{N}$.
\end{theorem}

\subsection{The \boost Procedure and the Main Algorithm} \label{sec:flsabr}
\begin{wrapfigure}{r}{0.63\textwidth}
	\begin{minipage}{0.63\textwidth}
		\vspace{-2em}
	  \begin{algorithm}[H]
		\caption{The \boost procedure.}
		\label{alg:boost}
		\begin{algorithmic}[1]
		  %	\Procedure{ParallelGreedyBoost}{$f, \mathcal N, k, \alpha, \Gamma, \epsi$}
	  \State \textbf{Input:} evaluation oracle $f:2^{\mathcal N} \to \reals$, constraint $k$, constant $\alpha$, value $\Gamma$ such that 
	  $\Gamma \leq f(O) \leq \Gamma/\alpha$, accuracy parameter $\epsi$
	  \State Initialize $\tau$ $\gets$ $\Gamma/(\alpha k)$, 
	  $\delta$ $\gets$ $1/(\log_{1-\epsi}(\alpha/3) + 1)$, 
	  $A \gets \emptyset$
	  \While{$\tau \geq \Gamma/(3 k)$} \label{line:boostWhileStart}
		  \State $\tau \gets \tau (1-\epsi)$
		  \State $S \gets \threshold(f_A, \mathcal{N}, k-|A|, \delta, \epsi/3, \tau)$\label{line:boostQuery}
		  \State $A \gets A \cup S$ \label{line:boostUpdateA}
		  \If{$|A| = k$}
		  \State \textbf{return} $A$
		  \EndIf
		  \EndWhile\label{line:boostWhileEnd}
	  \State \textbf{return} \textit{A}
  \end{algorithmic}
  \end{algorithm}
  \vspace{-2em}
  \end{minipage}
  \end{wrapfigure}
In this section, we describe the greedy algorithm
\boost (\boostshort, Alg. \ref{alg:boost}) that uses multiple calls to
\threshold with descending thresholds. Next, our state-of-the-art
algorithm \flsabr is specified.

\textbf{Description of \boost.}
This procedure
takes as input the results from running an $\alpha$-approximation
algorithm on the instance $(f,k)$ of \sm;
thus, \boost is not meant to be used as a standalone algorithm.
Namely, \boost takes as input $\Gamma$, the solution
value of an $\alpha$-approximation algorithm for \sm; this solution
value $\Gamma$ is then boosted to ensure the ratio $1 - 1/e - \epsi$
on the instance.
The values of $\Gamma$ and $\alpha$ are used to produce an
initial threshold value $\tau$ for \threshold. Then,
the threshold value is iteratively decreased by a factor of $(1 - \epsi)$ and the call to \threshold is iteratively repeated
to build up a solution, until a minimum value for the threshold
of $\Gamma /(3k)$ is reached.
Therefore, \threshold is called at most $\oh{ \log( 1 / \alpha ) / \epsi }$
times. We remark that $\alpha$ is not required to be a constant
approximation ratio.

\begin{theorem} \label{theorem:boost}
Let $(f,k)$ be an instance of \sm. 
Suppose an $\alpha$-
approximation algorithm
for \sm
is used to obtain
$\Gamma$,
where the approximation ratio $\alpha$ holds with probability $1 - p_\alpha$.
For any constant $\epsi >0$, the algorithm \boost has 
adaptive complexity 
$\oh{ \frac{ \log{ \alpha^{-1} } }{\epsi^2} \log \left( \frac{ n \log \left( \alpha^{-1} \right) }{\epsi } \right) } $
%$\oh{ \frac{ \log{ \alpha^{-1} } }{\epsi^2} \log \left( \frac{ \epsi n }{ \log \left( \alpha^{-1} \right) } \right) } $ %$O(\log(n/(\alpha \epsi))/(\alpha\epsi^2))$ 
and
outputs $A \in \mathcal{N}$ with $|A| \leq k$ such that the following 
properties hold:
1) The algorithm succeeds with probability at least $1 - 1/n - p_\alpha$.
2) If the algorithm succeeds, there are $\oh{ n \log \left( \alpha^{-1} \right) / \epsi^2}$ oracle queries in expectation.
3) If the algorithm succeeds, $f(A) \ge (1-1/e-\epsi) f(O)$, where $O$ is an optimal solution to the instance $(f,k)$.
\end{theorem}
\begin{proof}
	\textbf{Success Probability.}
	For the \textbf{while} loop in Line \ref{line:boostWhileStart}-\ref{line:boostWhileEnd}, 
there are no more than 
$\lceil \log_{1-\epsi}(\alpha/3)\rceil$ iterations. 
If \threshold completes successfully at every iteration,
Algorithm \ref{alg:boost} also succeeds.
The probability that this occurs is lower bounded in Appendix \ref{apdx:boost-suc}.
For the remainder of the proof of Theorem \ref{theorem:boost},
we assume that
every call to \threshold succeeds.

\textbf{Adaptive and Query Complexity.}
There are at most $\lceil \log_{1-\epsi}(\alpha/3)\rceil$ iterations of the \textbf{while} loop. 
Since $\log(x) \leq x-1$, 
$\lceil \log_{1-\epsi}(\alpha/3)\rceil = \lceil \frac{\log(\alpha/3)}{\log(1-\epsi)}\rceil$, and $\epsi < 1 - 1/e$,
it holds that 
$ \lceil \log_{1-\epsi}(\alpha/3)\rceil \le \lceil \frac{\log(3/\alpha)}{\epsi}\rceil. $
And for each iteration, 
queries to the oracle happen only on Line \ref{line:boostQuery},
the call to \threshold.
Since the adaptive and query complexity of \threshold is $\oh{\log(n/ \delta )/\epsi}$
and $\oh{n/\epsi}$,
the adaptive and query complexities for Algorithm \ref{alg:boost} are
$\oh{ \frac{ \log{ \alpha^{-1} } }{\epsi^2} \log \left( \frac{ n \log \left( \alpha^{-1} \right) }{\epsi } \right) },
\oh{ \frac{ \log{ \alpha^{-1} } }{\epsi^2} n}, $ respectively.

\textbf{Approximation Ratio.}
Let $A_j$ be the set $A$ we get after Line \ref{line:boostUpdateA}, 
and let $S_j$ be the set returned by \threshold in iteration $j$
of the \textbf{while} loop.
Let $\ell$ be the number of iterations of the \textbf{while} loop.

First, in the case that $|A| < k$ at termination, \threshold returns
$0\leq|S_\ell| < k-|A_{\ell-1}|$ at the last iteration.
From Theorem \ref{theorem:threshold}, for any $o \in O$, 
$\marge{o}{A} < \tau < \Gamma/(3k)$. By submodularity and monotonicity,
$
f(O) - f(A)  \leq f(O\cup A) - f(A) 
\leq \sum_{o \in O\backslash A}\marge{o}{A}
\leq \sum_{o \in O\backslash A} \Gamma/(3k) 
\leq f(O)/3,
$ and the ratio holds.

Second, consider the case that $|A| = k$.
Suppose in iteration $j+1$, 
\threshold returns a nonempty set $S_{j+1}$. 
Then, in the previous iteration $j$,
\threshold returns a set $S_j$ that $0 \leq |S_j|<k-|A_{j-1}|$.
From Theorem \ref{theorem:threshold}, 
\color{\rev}
\begin{equation} \label{ineq:boost1}
	f(O)-f(A_{j+1}) \le \left(1-\frac{(1-2\epsi/3)(1-\epsi)}{(1+\epsi/3)k} |A_{j+1}\backslash A_j|\right)(f(O)-f(A_j)).
\end{equation}
\color{black}
The above inequality also holds when $A_{j+1}=A_j$. 
Therefore, it holds that
\color{\rev}
\begin{equation}\label{ineq:boost2}
	f(O)-f(A) 
\le e^{-\frac{(1-2\epsi/3)(1-\epsi)}{1+\epsi/3}} f(O)
\le (1/e+\epsi)f(O).
\end{equation}
\color{black}
The detailed proof of 
Inequality \ref{ineq:boost1} and \ref{ineq:boost2}
can be found in Appendix \ref{apdx:boost-ratio}.
% Since \threshold has $\oh{n/\epsi}$ 
% oracle queries in expectation, and there are no more than 
% $\oh{ \log \left( \alpha^{-1} \right) / \epsi }$
% iterations, the total query calls for Algorithm \ref{alg:boost} 
% is $\oh{ n \log \left( \alpha^{-1} \right) / \epsi^2}$ 
% in expectation. 
\end{proof}
\textbf{Main Algorithm: \flsabr.}
To obtain the main algorithm of this paper (and its nearly optimal
theoretical guarantees), we use \boost with the solution value $\Gamma$ and ratio $\alpha$
given by \linearseq. Because this choice requires an initial run of \linearseq,
we denote this algorithm by \flsabr.
Thus, %without loss on adaptive or query complexity,
\flsabr integrates \linearseq and \boost to get nearly the optimal $1-1/e$ ratio
with query complexity of $\oh{n}$ and adaptivity of $\oh{\log(n)}$.
% To obtain the
% main algorithm of this paper (and its nearly optimal theoretical guarantees),
% we use \boost with
% the solution value and ratio of \ls for $\Gamma$ and $\alpha$, respectively.
% Because this choice requires an initial run of \ls, we denote this algorithm
% by \flsabr. 
% Other choices for $\Gamma, \alpha$ could be used. To make the preprocessing
% step extremely simple, one may use 
% the maximum singleton value for $\Gamma$ and $1/k$ for $\alpha$ to
% obtain an algorithm with nearly the optimal ratio and
% with $\oh{ \log(k)\log(n) }$ adaptivity and $\oh{ n \log (k) }$ expected query complexity.

%%% Local Variables:
%%% mode: latex
%%% TeX-master: "main.tex"
%%% End:

\section{Empirical Evaluation} \label{sec:exp}
\begin{figure}[t]
  \subfigure[]{
    \includegraphics[width=0.23\textwidth, height=0.11\textheight]{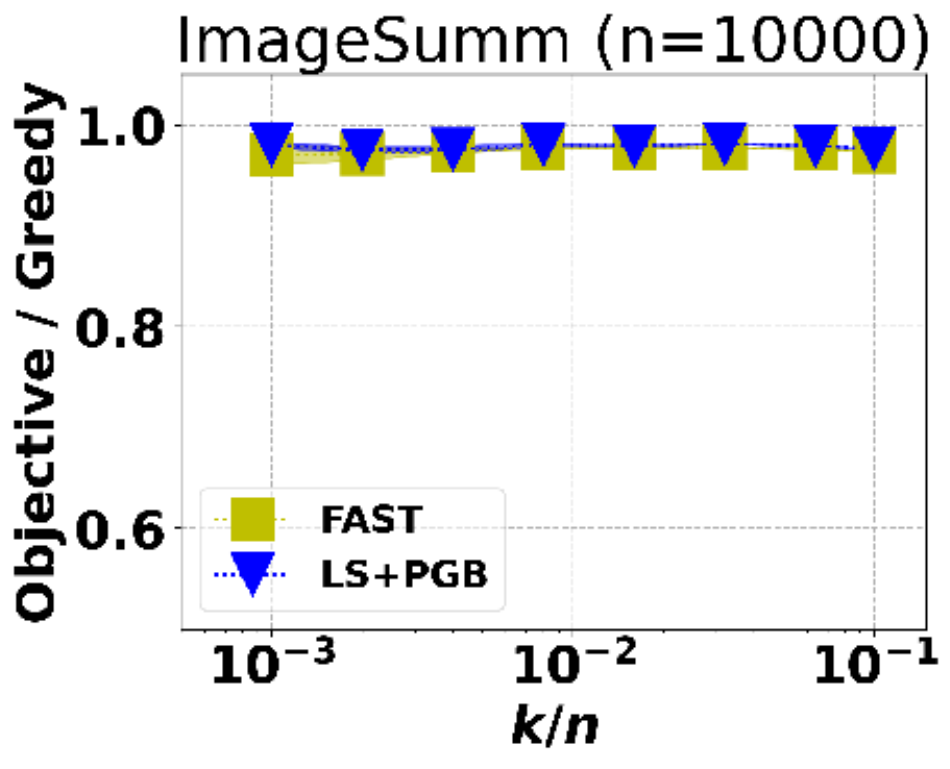}\label{fig:objC}
  }
  \subfigure[]{
    \includegraphics[width=0.23\textwidth, height=0.11\textheight]{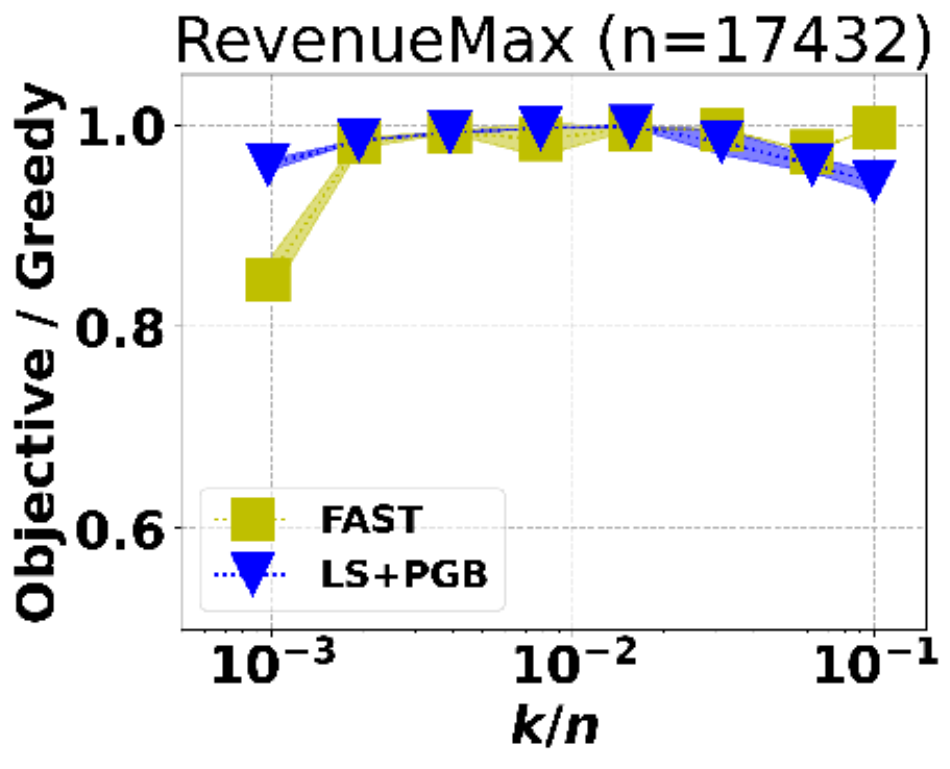} \label{fig:objD}
  }
  \subfigure[]{
    \includegraphics[width=0.23\textwidth, height=0.11\textheight]{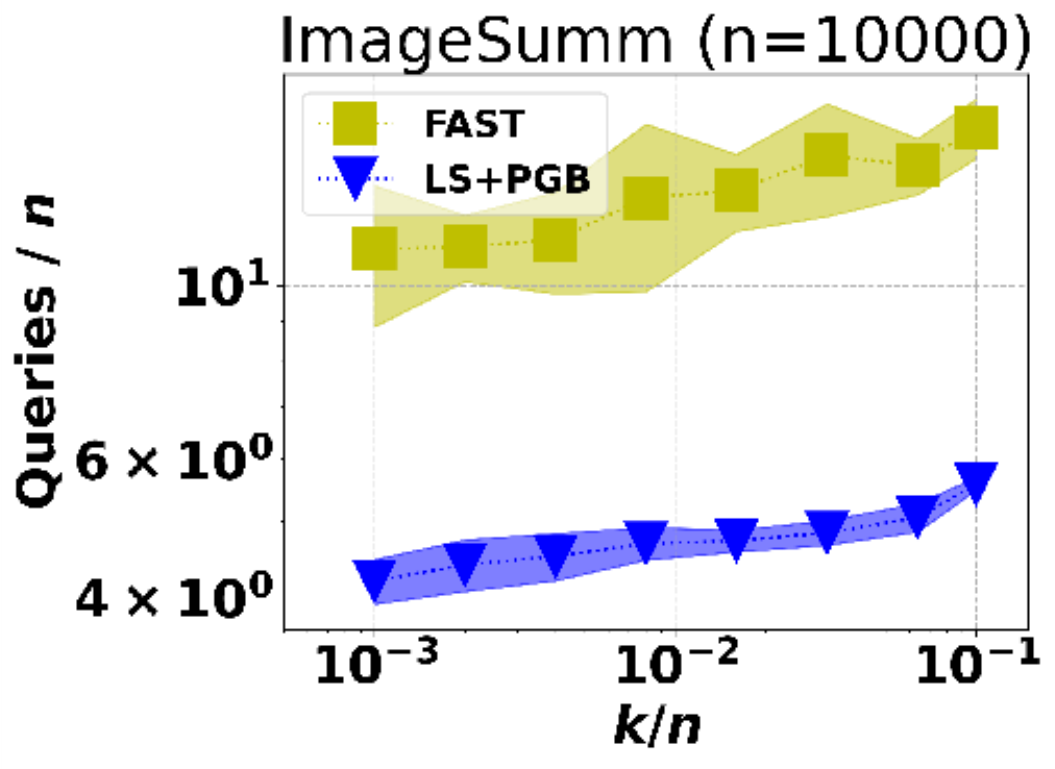}\label{fig:qryC}
  }
  \subfigure[]{
    \includegraphics[width=0.23\textwidth, height=0.11\textheight]{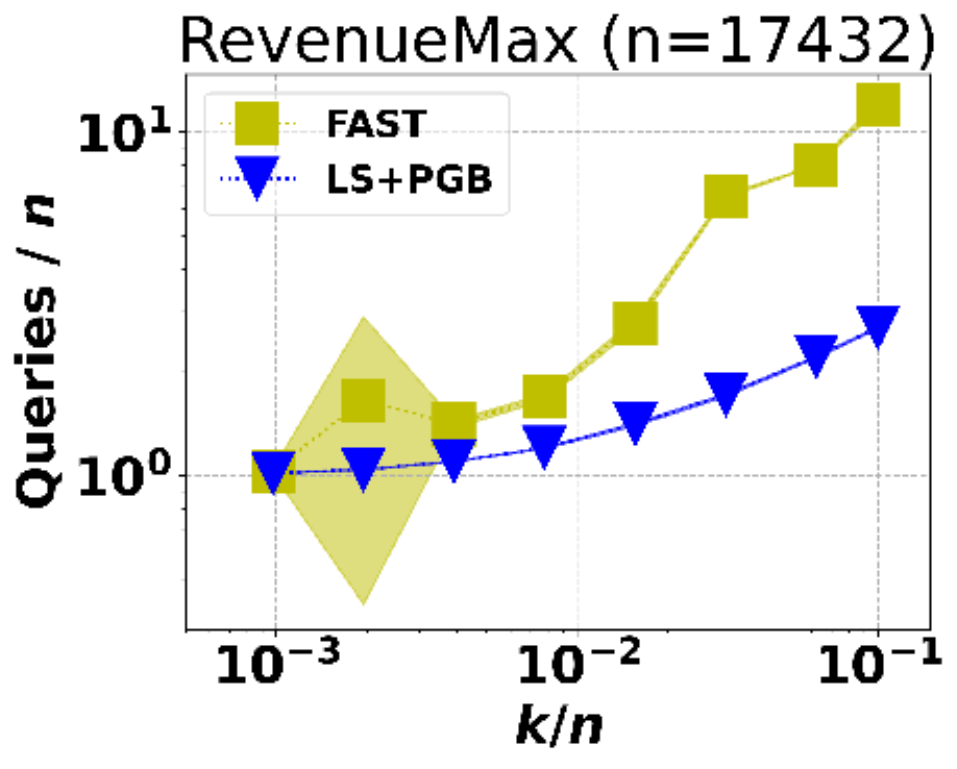} \label{fig:qryD}
  }
  
  \subfigure[]{
    \includegraphics[width=0.23\textwidth, height=0.11\textheight]{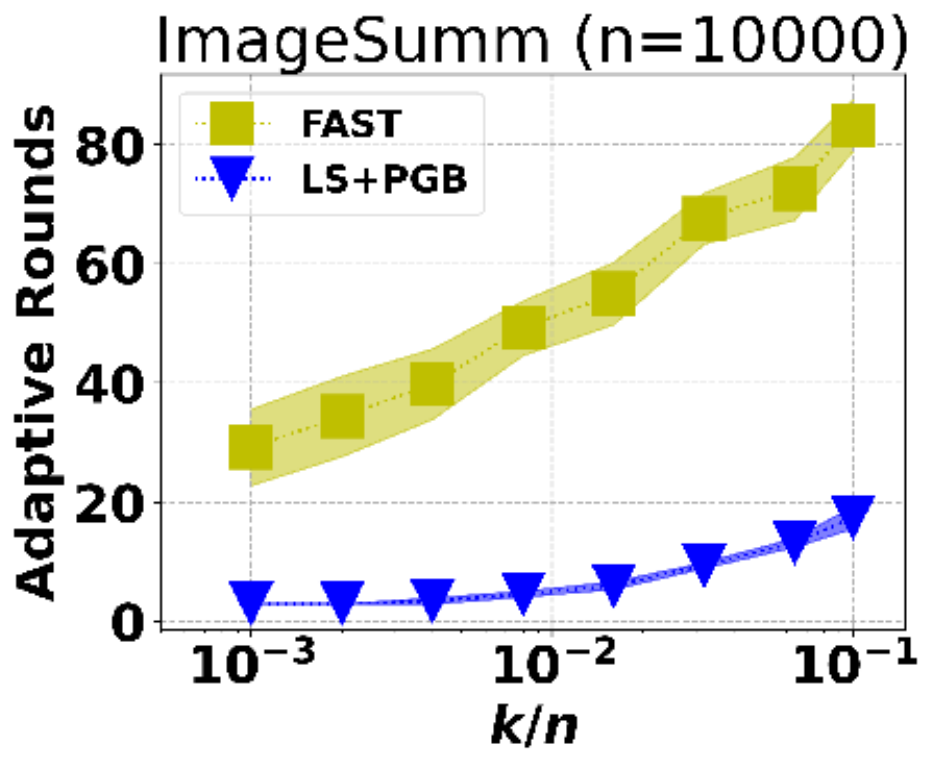}\label{fig:adaC}
  }
  \subfigure[]{
    \includegraphics[width=0.23\textwidth, height=0.11\textheight]{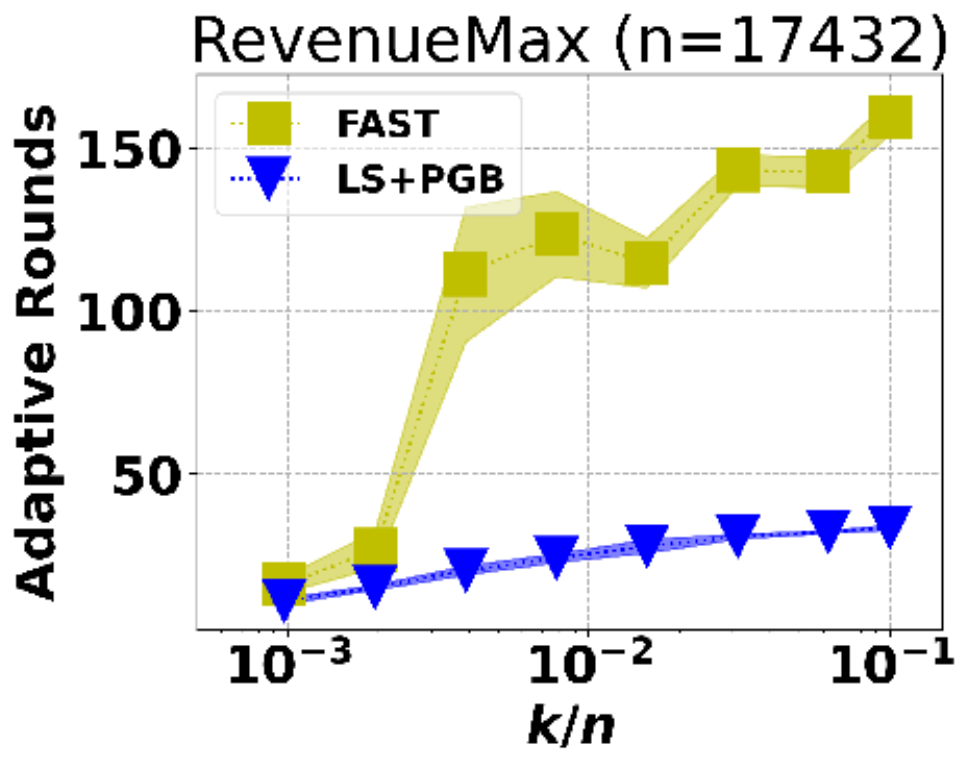}\label{fig:adaD}
  }
  \subfigure[]{
    \includegraphics[width=0.23\textwidth, height=0.11\textheight]{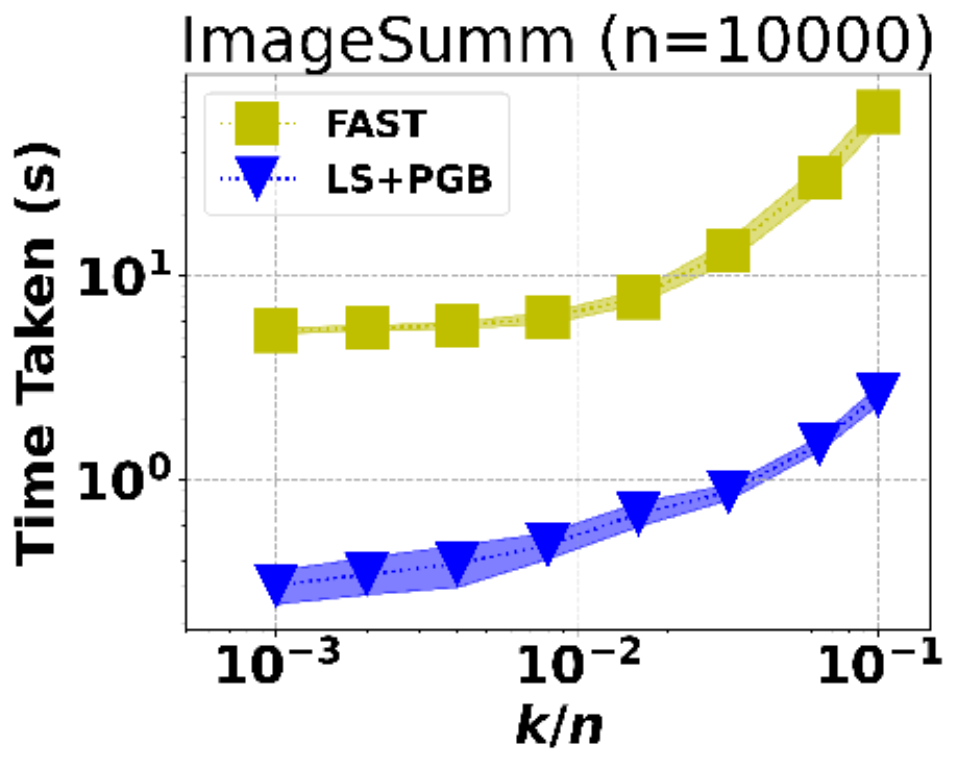}\label{fig:timeC}
  }
  \subfigure[]{
    \includegraphics[width=0.23\textwidth, height=0.11\textheight]{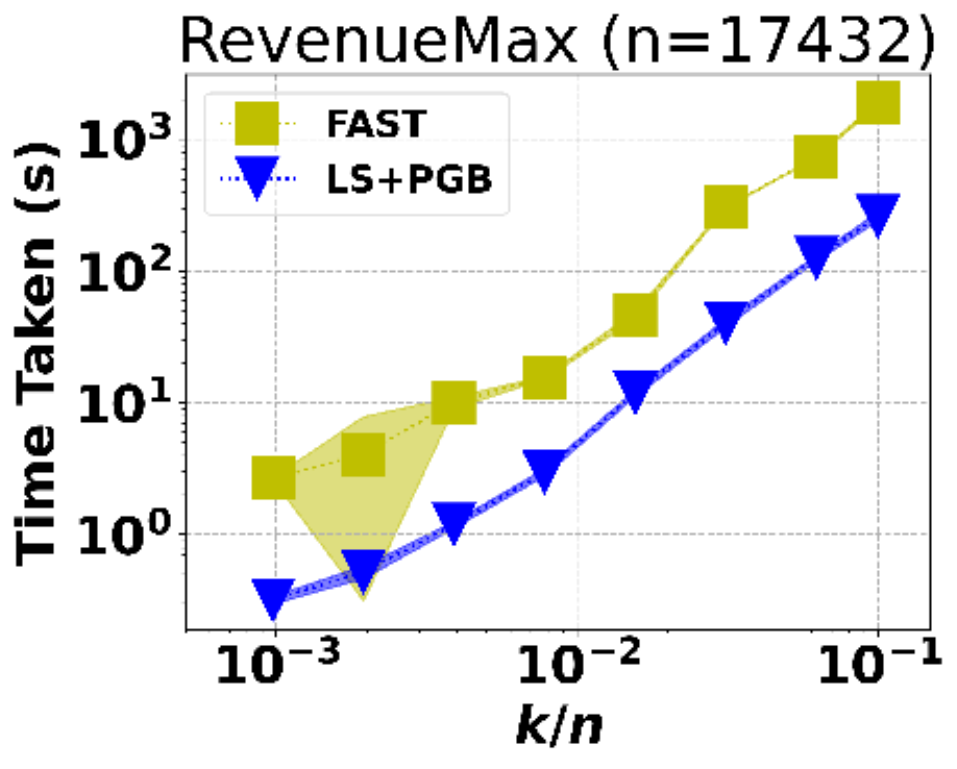}\label{fig:timeD}
  }
  % \subfigure[]{
  %   \includegraphics[width=0.23\textwidth, height=0.11\textheight]{figs/val-BA_exp1.pdf}
  % }
  % \subfigure[]{
  %   \includegraphics[width=0.23\textwidth, height=0.11\textheight]{figs/val-WS_exp1.pdf}
  % }
  % \subfigure[]{
  %   \includegraphics[width=0.23\textwidth, height=0.11\textheight]{figs/val-ER_exp1.pdf}
  % }
  % \subfigure[]{
  %   \includegraphics[width=0.23\textwidth, height=0.11\textheight]{figs/val-TWITTERSUMM_exp1.pdf}
  % }
  % \subfigure[]{
  %   \includegraphics[width=0.23\textwidth, height=0.11\textheight]{figs/val-IMAGESUMM_exp1.pdf}
  % }
  % \subfigure[]{
  %   \includegraphics[width=0.23\textwidth, height=0.11\textheight]{figs/val-INFLUENCEEPINIONS_exp1.pdf}
  % }
  % \subfigure[]{
  %   \includegraphics[width=0.23\textwidth, height=0.11\textheight]{figs/val-YOUTUBE2000_exp1.pdf}
  % }
  % \subfigure[]{
  %   \includegraphics[width=0.23\textwidth, height=0.11\textheight]{figs/val-CAROADFULL_exp1.pdf}
  % }
  \caption{Evaluation of adaptive algorithms on ImageSumm and RevenueMax in terms of objective value normalized by the standard greedy value (Figure \ref{fig:objC} -  \ref{fig:objD}), total number of queries (Figure \ref{fig:qryC} -  \ref{fig:qryD}), total adaptive rounds (Figure \ref{fig:adaC} -  \ref{fig:adaD}) and the time required by each algorithm (Figure \ref{fig:timeC} -  \ref{fig:timeD})} \label{fig:evalMain}
\end{figure} 
In this section, we demonstrate that the empirical performance
of \flsabr outperforms that of \fast
for the metrics of total time, total queries, adaptive rounds,
and objective value
across six applications of
\sm 
% defined in Section \ref{subsec:obj}. 
: maximum cover on random graphs (MaxCover),
twitter feed summarization (TweetSumm), image summarization (ImageSumm),
influence maximization (Influence), revenue maximization (RevMax),
and Traffic Speeding Sensor Placement (Traffic). 
See
Appendix \ref{subsec:obj} for the definition of the objectives.
The sizes $n$ of the ground sets range from $n=1885$ to $100000$.
% We report the results in comparison with \fast since an
% extensive evaluation in
% \citet{Breuer2019} demonstrated \fast is the
% state-of-the-art algorithm for \sm empirically
% across the metrics of time, objective value, and total
% queries. However, \flsabr also outperforms optimized
% implementations of
% \textsc{Lazier-than-Lazy-Greedy} \citep{Mirzasoleiman2014},
% \textsc{Randomized-Parallel-Greedy} \citep{Chekuri2018a},
% \textsc{Binary-Search-Maximization} \citep{Fahrbach2018},
% and \textsc{Amortized-Filtering} \citep{Balkanski}.

% the following: 1) \fast was the
% only algorithm for \sm with sublinear adaptivity capable
% of running on large instances; 2) \fast
% outperformed highly optimized implementations
% of the previously fastest algorithms for \sm,
% including single-threaded and multi-threaded versions of
% the linear-time lazier-than-lazy-greedy algorithm of
% \citet{Mirzasoleiman2014}.

\textbf{Implementation and Environment.}
We evaluate the same implementation
of \fast used in \citet{Breuer2019}.
Our implementation of \flsabr
is parallelized using
the Message Passing Interface (MPI) within the same Python
codebase as \fast (see the 
Supplementary Material for source code).
Practical optimizations to \linearseq
are made, which do not compromise the theoretical
guarantees, which are discussed in Appendix \ref{apx:practical-opt}.
The hardware of the system consists of 40 Intel(R) Xeon(R) Gold 5218R CPU @ 2.10GHz cores (with 80 threads available), of which up to 75 threads are made available to the algorithms for the experiments.
On each instance,
the algorithms are repeated independently for five repetitions,
and the mean and standard deviation of the objective value, total queries, adaptive rounds and parallel (wall clock) runtime to the submodular function is plotted.
% for each instance.

\textbf{Parameters.}
The parameters $\epsi, \delta$ of \fast
are set to enforce the nominal ratio of
$1 - 1/e - 0.1 \approx 0.53$ with probability $0.95$;
these are the same parameter
settings for \fast as in the \citet{Breuer2019} evaluation.
The $\epsi$ parameter of \flsabr is set to enforce the same
ratio with probability $1 - 2/n$.
With these parameters, \fast ensures
its ratio only if $k \ge \theta( \epsi, \delta, k ) = \fracc{2\log( 2 \delta^{-1} \log( \frac{1}{\epsi}\log (k) ) )}{\epsi^2 (1 - 5 \epsi) } \ge 7103$.
Since $k < 7103$ on many of our instances,
\fast is evaluated in these instances
as a theoretically motivated heuristic. In contrast,
the ratio of \flsabr holds on all instances evaluated.
% These settings are advantageous to \fast
% in two ways:
% 1) the ratio of \fast is obtained with
% probability $1 - \delta$ whereas the
% ratio of \flsabr holds with probability
% $1 - 1/n$; we use the same setting 
% $\delta = 0.025$ as in \citet{Breuer2019}; and 2)
We use exponentially increasing $k$ values from $n/1000$ to $n/10$
for each application to explore the behavior of
each algorithm across a broad range of instance sizes.

\textbf{Overview of Results.}
%Even on those instances where
  % \fast requires only a single guess for \opt, our algorithm is faster, and our algorithm obtains an order of magnitude improvement in runtime when \fast must run its entire sequence of \opt guesses. Moreover, our algorithm yields much more stable results than \fast, which can exhibit large fluctuations in runtime and objective value because of the inconsistent nature
  % of some of its heuristic optimizations.
Figure \ref{fig:evalMain} illustrates the comparison with \fast across the ImageSumm and RevenueMax application;
results on other applications are shown in Appendix \ref{apx:exp}.  
\textbf{Runtime:} \flsabr is faster than \fast by more than $1\%$ on $80\%$ 
% \footnote{For all metrics, improvements are only considered when performance gap is $\ge$ 1\%.} 
of instances
evaluated; and is faster by an order of magnitude on $14\%$ of instances.
\textbf{Objective value:} \flsabr achieves higher objective 
by more than $1\%$ on $50\%$ of instances,
whereas \fast achieves higher objective 
by more than $1\%$ on $8\%$ of instances.
\textbf{Adaptive rounds:} \flsabr achieves more than $1\%$ fewer adaptive rounds
on $75\%$ of instances,
while \fast achieves more than $1\%$ fewer adaptive rounds on $22\%$ of instances. 
\textbf{Total queries:} \flsabr uses more than 1\% fewer
queries on $84\%$ of scenarios with \fast using
more than 1\% fewer
queries on $9\%$ of scenarios. 
In summary, \flsabr frequently gives substantial improvement in objective value, queries, adaptive rounds,
and parallel runtime. 
Comparison of the arithmetic means of the metrics over all instances is given 
in Table \ref{table:cmp-exp}.
Finally, \fast and \flsabr show very
similar linear speedup with the number of processors employed:
as shown in Fig. \ref{fig:parallel}.

\section{Concluding Remarks} \label{sec:concl}
In this work, we have introduced the
algorithm \flsabr, which is highly
parallelizable and achieves state-of-the-art
empirical performance over any previous algorithm
for \sm; also, \flsabr is nearly optimal theoretically
in terms of query complexity, adaptivity, and approximation
ratio. An integral component of \flsabr is
our preprocessing algorithm \linearseq,
which reduces the interval
containing \opt to a small constant size
in expected linear time and low adaptivity,
which may be independently useful. Another
component of \flsabr is the \threshold
procedure, which adds all elements with
gain above a threshold in a parallelizable
manner and improves existing algorithms
in the literature for the same task.

\clearpage

\section*{Acknowledgements}
The work of Yixin Chen, Tonmoy Dey, and Alan Kuhnle was partially 
supported by Florida State University.
The authors have received no third-party funding in direct support of this work. 
The authors have no additional revenues from other sources related to this work.
\bibliographystyle{plainnatfixed}
\bibliography{mend,bibG1_G5}

\begin{thebibliography}{33}
\providecommand{\natexlab}[1]{#1}
\providecommand{\url}[1]{\texttt{#1}}
\expandafter\ifx\csname urlstyle\endcsname\relax
  \providecommand{\doi}[1]{doi: #1}\else
  \providecommand{\doi}{doi: \begingroup \urlstyle{rm}\Url}\fi

\bibitem[cal()]{caltrans}
{CalTrans. Pems: California performance measuring system}.
\newblock \url{http://pems.dot.ca.gov/}.
\newblock [Online; accessed 1-May-2018].

\bibitem[Badanidiyuru and Vondr{\'{a}}k(2014)]{Badanidiyuru2014}
Ashwinkumar Badanidiyuru and Jan Vondr{\'{a}}k.
\newblock {Fast algorithms for maximizing submodular functions}.
\newblock In \emph{ACM-SIAM Symposium on Discrete Algorithms (SODA)}, 2014.

\bibitem[Balkanski and Singer(2018)]{Balkanski2018}
Eric Balkanski and Yaron Singer.
\newblock {The adaptive complexity of maximizing a submodular function}.
\newblock In \emph{ACM SIGACT Symposium on Theory of Computing (STOC)}, 2018.

\bibitem[Balkanski et~al.(2018)Balkanski, Breuer, and
  Singer]{DBLP:conf/nips/BalkanskiBS18}
Eric Balkanski, Adam Breuer, and Yaron Singer.
\newblock Non-monotone submodular maximization in exponentially fewer
  iterations.
\newblock In \emph{Advances in Neural Information Processing Systems 31: Annual
  Conference on Neural Information Processing Systems 2018, NeurIPS 2018,
  December 3-8, 2018, Montr{\'{e}}al, Canada}, pages 2359--2370, 2018.

\bibitem[Balkanski et~al.(2019{\natexlab{a}})Balkanski, Rubinstein, and
  Singer]{Balkanski}
Eric Balkanski, Aviad Rubinstein, and Yaron Singer.
\newblock {An Exponential Speedup in Parallel Running Time for Submodular
  Maximization without Loss in Approximation}.
\newblock In \emph{ACM-SIAM Symposium on Discrete Algorithms (SODA)},
  2019{\natexlab{a}}.

\bibitem[Balkanski et~al.(2019{\natexlab{b}})Balkanski, Rubinstein, and
  Singer]{Balkanski2018c}
Eric Balkanski, Aviad Rubinstein, and Yaron Singer.
\newblock {An optimal approximation for submodular maximization under a matroid
  constraint in the adaptive complexity model}.
\newblock In \emph{Proceedings of the Annual ACM Symposium on Theory of
  Computing}, pages 66--77, nov 2019{\natexlab{b}}.

\bibitem[Barbosa et~al.(2015)Barbosa, Ene, {Le Nguyen}, and Ward]{Barbosa2015}
Rafael Barbosa, Alina Ene, Huy {Le Nguyen}, and Justin Ward.
\newblock {The Power of Randomization: Distributed Submodular Maximization on
  Massive Datasets}.
\newblock In \emph{International Conference on Machine Learning (ICML)}, 2015.

\bibitem[Barbosa et~al.(2016)Barbosa, Ene, Nguyen, and Ward]{Barbosa2016}
Rafael Da~Ponte Barbosa, Alina Ene, Huy~L. Nguyen, and Justin Ward.
\newblock {A New Framework for Distributed Submodular Maximization}.
\newblock In \emph{IEEE Symposium on Foundations of Computer Science (FOCS)},
  2016.

\bibitem[Breuer et~al.(2019)Breuer, Balkanski, and Singer]{Breuer2019}
Adam Breuer, Eric Balkanski, and Yaron Singer.
\newblock {The FAST Algorithm for Submodular Maximization}.
\newblock In \emph{International Conference on Machine Learning (ICML)}, 2019.

\bibitem[Calinescu et~al.(2007)Calinescu, Chekuri, P{\'{a}}l, and
  Vondr{\'{a}}k]{Calinescu2007}
Gruia Calinescu, Chandra Chekuri, Martin P{\'{a}}l, and Jan Vondr{\'{a}}k.
\newblock {Maximizing a Submodular Set Function subject to a Matroid
  Constraint}.
\newblock In \emph{Integer Programming and Combinatorial Optimization (IPCO)},
  pages 182--196, 2007.

\bibitem[Chekuri and Quanrud(2019)]{Chekuri2018a}
Chandra Chekuri and Kent Quanrud.
\newblock {Submodular function maximization in parallel via the multilinear
  relaxation}.
\newblock In \emph{Proceedings of the Annual ACM-SIAM Symposium on Discrete
  Algorithms (SODA)}, pages 303--322, jul 2019.

\bibitem[Conforti and Cornu{\'{e}}jols(1984)]{Conforti1984}
Michele Conforti and G{\'{e}}rard Cornu{\'{e}}jols.
\newblock {Submodular set functions, matroids and the greedy algorithm: Tight
  worst-case bounds and some generalizations of the Rado-Edmonds theorem}.
\newblock \emph{Discrete Applied Mathematics}, 7\penalty0 (3):\penalty0
  251--274, 1984.

\bibitem[Dean and Ghemawat(2008)]{Dean2008}
Jeffrey Dean and Sanjay Ghemawat.
\newblock {MapReduce: Simplified Data Processing on Large Clusters}.
\newblock \emph{Communications of the ACM}, 51\penalty0 (1):\penalty0 107--113,
  2008.

\bibitem[Ene and Nguyen(2019)]{Ene}
Alina Ene and Huy~L Nguyen.
\newblock {Submodular Maximization with Nearly-optimal Approximation and
  Adaptivity in Nearly-linear Time}.
\newblock In \emph{ACM-SIAM Symposium on Discrete Algorithms (SODA)}, 2019.

\bibitem[Ene and Nguy{\^{e}}n(2020)]{Ene2020}
Alina Ene and Huy~L. Nguy{\^{e}}n.
\newblock {Parallel Algorithm for Non-Monotone DR-Submodular Maximization}.
\newblock In \emph{International Conference on Machine Learning (ICML)}, 2020.

\bibitem[Epasto et~al.(2017)Epasto, Mirrokni, and Zadimoghaddam]{Epasto2017}
Alessandro Epasto, Vahab Mirrokni, and Morteza Zadimoghaddam.
\newblock {Bicriteria Distributed Submodular Maximization in a Few Rounds}.
\newblock In \emph{Symposium on Parallelism in Algorithms and Architectures
  (SPAA)}, 2017.

\bibitem[Fahrbach et~al.(2019{\natexlab{a}})Fahrbach, Mirrokni, and
  Zadimoghaddam]{Fahrbach2018}
Matthew Fahrbach, Vahab Mirrokni, and Morteza Zadimoghaddam.
\newblock {Submodular Maximization with Nearly Optimal Approximation,
  Adaptivity, and Query Complexity}.
\newblock In \emph{ACM-SIAM Symposium on Discrete Algorithms (SODA)}, pages
  255--273, 2019{\natexlab{a}}.

\bibitem[Fahrbach et~al.(2019{\natexlab{b}})Fahrbach, Mirrokni, and
  Zadimoghaddam]{Fahrbach2018a}
Matthew Fahrbach, Vahab Mirrokni, and Morteza Zadimoghaddam.
\newblock {Non-monotone Submodular Maximization with Nearly Optimal Adaptivity
  Complexity}.
\newblock In \emph{International Conference on Machine Learning (ICML)},
  2019{\natexlab{b}}.

\bibitem[Fahrbach et~al.(2019{\natexlab{c}})Fahrbach, Mirrokni, and
  Zadimoghaddam]{fahrbach2019non}
Matthew Fahrbach, Vahab Mirrokni, and Morteza Zadimoghaddam.
\newblock Non-monotone submodular maximization with nearly optimal adaptivity
  and query complexity.
\newblock In \emph{International Conference on Machine Learning}, pages
  1833--1842. PMLR, 2019{\natexlab{c}}.

\bibitem[Gillenwater et~al.(2012)Gillenwater, Kulesza, and
  Taskar]{Gillenwater2012}
Jennifer Gillenwater, Alex Kulesza, and Ben Taskar.
\newblock {Near-Optimal MAP Inference for Determinantal Point Processes}.
\newblock In \emph{Advances in Neural Information Processing Systems
  (NeurIPS)}, 2012.

\bibitem[Horel and Singer(2016)]{Horel2016}
Thibaut Horel and Yaron Singer.
\newblock {Maximization of Approximately Submodular Functions}.
\newblock In \emph{Advances in Neural Information Processing Systems
  (NeurIPS)}, 2016.

\bibitem[Kazemi et~al.(2019)Kazemi, Mitrovic, Zadimoghaddam, Lattanzi, and
  Karbasi]{Kazemi2019}
Ehsan Kazemi, Marko Mitrovic, Morteza Zadimoghaddam, Silvio Lattanzi, and Amin
  Karbasi.
\newblock {Submodular Streaming in All its Glory: Tight Approximation, Minimum
  Memory and Low Adaptive Complexity}.
\newblock In \emph{International Conference on Machine Learning (ICML)}, 2019.

\bibitem[Krause and Guestrin(2007)]{Krause2007}
Andreas Krause and Carlos Guestrin.
\newblock {Near-optimal observation selection using submodular functions}.
\newblock \emph{AAAI Conference on Artificial Intelligence}, 2007.

\bibitem[Kuhnle(2021{\natexlab{a}})]{Kuhnle2020a}
Alan Kuhnle.
\newblock {Quick Streaming Algorithms for Maximization of Monotone Submodular
  Functions in Linear Time}.
\newblock In \emph{Artificial Intelligence and Statistics (AISTATS)},
  2021{\natexlab{a}}.

\bibitem[Kuhnle(2021{\natexlab{b}})]{Kuhnle2020b}
Alan Kuhnle.
\newblock {Nearly Linear-Time, Parallelizable Algorithms for Non-Monotone
  Submodular Maximization}.
\newblock In \emph{AAAI Conference on Artificial Intelligence},
  2021{\natexlab{b}}.

\bibitem[Leskovec et~al.(2007)Leskovec, Krause, Guestrin, Faloutsos,
  VanBriesen, and Glance]{Leskovec2007}
Jure Leskovec, Andreas Krause, Carlos Guestrin, Christos Faloutsos, Jeanne
  VanBriesen, and Natalie Glance.
\newblock {Cost-effective Outbreak Detection in Networks}.
\newblock In \emph{ACM SIGKDD International Conference on Knowledge Discovery
  and Data Mining (KDD)}, 2007.

\bibitem[Minoux(1978)]{minoux1978accelerated}
Michel Minoux.
\newblock Accelerated greedy algorithms for maximizing submodular set
  functions.
\newblock In \emph{Optimization techniques}, pages 234--243. Springer, 1978.

\bibitem[Mirrokni and Zadimoghaddam(2015)]{Mirrokni2015}
Vahab Mirrokni and Morteza Zadimoghaddam.
\newblock {Randomized Composable Core-Sets for Distributed Submodular
  Maximization}.
\newblock In \emph{ACM Symposium on Theory of Computing (STOC)}, 2015.

\bibitem[Mirzasoleiman et~al.(2016)Mirzasoleiman, Badanidiyuru, and
  Karbasi]{mirzasoleiman2016fast}
Baharan Mirzasoleiman, Ashwinkumar Badanidiyuru, and Amin Karbasi.
\newblock Fast constrained submodular maximization: Personalized data
  summarization.
\newblock In \emph{International Conference on Machine Learning}, pages
  1358--1367. PMLR, 2016.

\bibitem[Mirzasoleiman et~al.(2018)Mirzasoleiman, Jegelka, and
  Krause]{Mirzasoleiman2018}
Baharan Mirzasoleiman, Stefanie Jegelka, and Andreas Krause.
\newblock {Streaming Non-Monotone Submodular Maximization: Personalized Video
  Summarization on the Fly}.
\newblock In \emph{AAAI Conference on Artificial Intelligence}, 2018.

\bibitem[Mitzenmacher and Upfal(2017)]{mitzenmacher2017probability}
Michael Mitzenmacher and Eli Upfal.
\newblock \emph{Probability and computing: Randomization and probabilistic
  techniques in algorithms and data analysis}.
\newblock Cambridge university press, 2017.

\bibitem[Nemhauser and Wolsey(1978)]{Nemhauser1978a}
G~L Nemhauser and L~A Wolsey.
\newblock {Best Algorithms for Approximating the Maximum of a Submodular Set
  Function}.
\newblock \emph{Mathematics of Operations Research}, 3\penalty0 (3):\penalty0
  177--188, 1978.

\bibitem[Rossi and Ahmed(2015)]{rossi2015network}
Ryan Rossi and Nesreen Ahmed.
\newblock The network data repository with interactive graph analytics and
  visualization.
\newblock In \emph{Proceedings of the AAAI Conference on Artificial
  Intelligence}, volume~29, 2015.

\end{thebibliography}

\clearpage
\appendix
\section{Probability Lemma and Concentration Bounds} \label{apx:prob}
In this section, we state the Chernoff bound and prove a useful
lemma (Lemma \ref{lemma:indep})
for working with a sequence of dependent Bernoulli trials.
Lemma \ref{lemma:indep} is applied in the analysis of both
\threshold and \linearseq.
\begin{lemma}\label{lemma:chernoff}
    (Chernoff bounds \cite{mitzenmacher2017probability}). 
    Suppose $X_1$, ... , $X_n$ are independent binary random variables such that 
    $\prob{X_i = 1} = p_i$. Let $\mu = \sum_{i=1}^n p_i$, and 
    $X = \sum_{i=1}^n X_i$. Then for any $\delta \geq 0$, we have
    \begin{align}
        \prob{X \ge (1+\delta)\mu} \le e^{-\frac{\delta^2 \mu}{2+\delta}}.
    \end{align}
    Moreover, for any $0 \leq \delta \leq 1$, we have
    \begin{align}
        \prob{X \le (1-\delta)\mu} \le e^{-\frac{\delta^2 \mu}{2}}.
    \end{align}
\end{lemma}

\begin{lemma} \label{lemma:indep}
  Suppose there is a sequence of $n$ Bernoulli trials:
  $X_1, X_2, \ldots, X_n,$
  where the success probability of $X_i$
  depends on the results of
  the preceding trials $X_1, \ldots, X_{i-1}$.
  Suppose it holds that $$\prob{X_i = 1 | X_1 = x_1, X_2 = x_2, \ldots, X_{i-1} = x_{i-1} } \ge \eta,$$ where $\eta > 0$ is a constant and $x_1,\ldots,x_{i-1}$ are arbitrary.

  Then, if $Y_1,\ldots, Y_n$ are independent Bernoulli trials, each with probability $\eta$ of
  success, then $$\prob {\sum_{i = 1}^n X_i \le b } \le \prob{\sum_{i = 1}^n Y_i \le b }, $$
  where $b$ is an arbitrary integer.

  Moreover, let $A$ be the first occurrence of success in sequence $X_i$.
  Then, $$\ex{A} \le 1/\eta.$$
\end{lemma}

\begin{proof}
Let $Z_j = \sum_{i=1}^j X_i + \sum_{i=j+1}^n Y_i$, where 
$Z_0 = \sum_{i=1}^n Y_i$ and $Z_n = \sum_{i=1}^n X_i$.
Our goal is to prove that $\prob{Z_n\le b} \le \prob{Z_0\le b}$.
This lemma holds, if for any $j = 1,\ldots,n$, 
$\prob{Z_j\le b} \le \prob{Z_{j-1}\le b}$.
With Bayes Theorem and Total Probability Theorem,
\begin{align*}
\prob{Z_j\le b}&=\prob{X_j=0,Z_j-X_j\le b-1}+
\prob{X_j=1,Z_j-X_j\le b-1}\\
&\quad+\prob{X_j=0,Z_j-X_j= b}\\
&=\prob{Z_j-X_j\le b-1}+ \sum_{Z_j-X_j=b} 
\prob{X_j=0|X_1,\ldots,X_{j-1},Y_{j+1},\ldots,Y_n}\\
&\quad\cdot
\prob{X_1,\ldots,X_{j-1},Y_{j+1},\ldots,Y_n}\\
&\le \prob{Z_{j-1}-Y_j\le b-1} \\
&\quad+ \sum_{Z_{j-1}-Y_j=b} 
\prob{Y_j=0}\cdot \prob{X_1,\ldots,X_{j-1},Y_{j+1},\ldots,Y_n}\numberthis \label{ineq:dep-to-indep}\\
&= \prob{Z_{j-1}-Y_j\le b-1} +\prob{Y_j=0,Z_{j-1}-Y_j=b}\\
&= \prob{Z_{j-1}\le b},
\end{align*}
where inequality \ref{ineq:dep-to-indep} follows from 
$\prob{X_j=0|X_1 = x_1, X_2 = x_2, \ldots, X_{j-1} = x_{j-1} } \leq 
1-\eta=\prob{Y_j=0}$ and $Z_j-X_j=Z_{j-1}-Y_j$.

Following the first inequality, we can prove the second one as follows:
\begin{align*}
\ex{A}&= \sum_{a\ge1} a \prob{A=a}\\
&=\sum_{a\ge1} a \prob{X_1=\ldots=X_{a-1}=0, X_a=1}\\
&=\sum_{a\ge1} a\left(\prob{X_1=\ldots=X_{a-1}=0}-\prob{X_1=\ldots=X_a=0}\right)\\
&=1+\sum_{a\ge1} \prob{X_1=\ldots=X_a=0}\\
&=1+\sum_{a\ge1} \prob{\sum_{i=1}^a X_a\le 0}\\
&\le1+\sum_{a\ge1} \prob{\sum_{i=1}^a Y_a\le 0}\\
&=1+\sum_{a\ge1}\prob{Y_1=\ldots=Y_a=0}\\
&=1+\sum_{a\ge1}(1-\eta)^a\\
&=1/\eta.
\end{align*}
\end{proof}
\begin{lemma} \label{lemma:indep2}
  Suppose there is a sequence of $n+1$ Bernoulli trials:
  $X_1, X_2, \ldots,X_{n+1},$
  where the success probability of $X_i$
  depends on the results of
  the preceding trials $X_1, \ldots, X_{i-1}$,
  and it decreases from 1 to 0.
  Let $t$ be a random variable based on the $n+1$ Bernoulli trials.
  Suppose it holds that 
  $$\prob{X_i = 1 | X_1 = x_1, X_2 = x_2, \ldots, X_{i-1} = x_{i-1}, i\le t } \ge \eta,$$ 
  where $x_1,\ldots,x_{i-1}$ are arbitrary and $0 < \eta < 1$ is a constant.
  Then, if $Y_1,\ldots, Y_{n+1}$ are independent Bernoulli trials, each with probability $\eta$ of
  success, then 
  $$\prob {\sum_{i = 1}^t X_i \le bt } \le \prob{\sum_{i = 1}^t Y_i \le bt }, $$
  where $b$ is an arbitrary integer.
\end{lemma}
\begin{proof}
  Let $Z_j=\sum_{i=1}^j X_i \cdot 1_{\{i\le t\}} + 
  \sum_{i=j+1}^{n+1} Y_i \cdot 1_{\{i\le t\}}$,
  where $$Z_0=\sum_{i=1}^{n+1} Y_i \cdot 1_{\{i\le t\}} = 
  \sum_{i = 1}^{t} Y_i,$$
  and $$Z_{n+1}=\sum_{i=1}^{n+1} X_i \cdot 1_{\{i\le t\}} = 
  \sum_{i = 1}^{t} X_i.$$
  If for any $1\le j \le n+1$, 
  \begin{equation} \label{ineq:seq}
    \prob{Z_j \le bt} \le 
  \prob{Z_{j-1} \le bt},
  \end{equation}
  then,
  $$\prob{Z_{n+1} \le bt} \le
  \prob{Z_0 \le bt}.$$

  We prove Inequality \ref{ineq:seq} as follows,
  \begin{align*}
    &\prob{Z_j \le bt}\\
    &= \prob{X_j\cdot 1_{\{j\le t\}}=0, Z_j-X_j\cdot 1_{\{j\le t\}} \le bt-1}\\
    & \quad +\prob{X_j\cdot 1_{\{j\le t\}}=1, Z_j-X_j\cdot 1_{\{j\le t\}} \le bt-1}\\
    & \quad +\prob{X_j\cdot 1_{\{j\le t\}}=0, Z_j-X_j\cdot 1_{\{j\le t\}} = bt}\\
    &= \prob{Z_j-X_j\cdot 1_{\{j\le t\}} \le bt-1}\\
    &\quad + \prob{1_{\{j\le t\}}=0,Z_j-X_j\cdot 1_{\{j\le t\}} = bt}\\
    &\quad + \prob{X_j=0,1_{\{j\le t\}}=1,Z_j-X_j\cdot 1_{\{j\le t\}} = bt}\\
    &= \prob{Z_j-X_j\cdot 1_{\{j\le t\}} \le bt-1}\\
    &\quad + \prob{1_{\{j\le t\}}=0,Z_j-X_j\cdot 1_{\{j\le t\}} = bt}\\
    &\quad + \sum_{Z_j-X_j\cdot 1_{\{j\le t\}} = bt,j\le t}
    \prob{X_j=0 | X_1,\cdots, X_{j-1},Y_{j+1},\cdots, Y_{n+1}, j\le t} \\
    &\quad \cdot \prob{X_1,\cdots, X_{j-1},Y_{j+1},\cdots, Y_{n+1}, j\le t}\\
    &\le \prob{Z_{j-1}-Y_j\cdot 1_{\{j\le t\}} \le bt-1}\\
    &\quad + \prob{1_{\{j\le t\}}=0,Z_{j-1}-Y_j\cdot 1_{\{j\le t\}} = bt}\\
    &\quad + \sum_{Z_{j-1}-Y_j\cdot 1_{\{j\le t\}} = bt,j\le t}
    \prob{Y_j=0}\cdot \prob{X_1,\cdots, X_{j-1},Y_{j+1},\cdots, Y_{n+1}, j\le t}\numberthis \label{ineq:dep-to-indep2}\\
    &= \prob{Z_{j-1}-Y_j\cdot 1_{\{j\le t\}} \le bt-1}\\
    &\quad + \prob{1_{\{j\le t\}}=0,Z_{j-1}-Y_j\cdot 1_{\{j\le t\}} = bt}\\
    &\quad + \prob{Y_j=0,1_{\{j \le t\}}=1,Z_{j-1}-Y_j\cdot 1_{\{j\le t\}} = bt}\\
    &= \prob{Z_{j-1}-Y_j\cdot 1_{\{j\le t\}} \le bt-1}\\
    &\quad + \prob{Y_j\cdot 1_{\{j\le t\}}=0,Z_{j-1}-Y_j\cdot 1_{\{j\le t\}} = bt}\\
    &= \prob{Z_{j-1}\le bt},
  \end{align*}
  where Inequality \ref{ineq:dep-to-indep2} follows from 
  $\prob{X_i = 1 | X_1 = x_1, X_2 = x_2, \ldots, X_{i-1} = x_{i-1}, i\le t } \ge \eta$
  and $Z_j-X_j\cdot 1_{\{j\le t\}} = Z_{j-1}-Y_j\cdot 1_{\{j\le t\}}$.
\end{proof}
% \begin{lemma} \label{lemma:wald}
%   (Wald's Equation \cite{wald1945some}).
%   Let $(X_n)_{n \in \mathbb{N}}$ be an infinite sequence of real-valued random variables 
%   and let $N$ be a nonnegative integer-valued random variable.
%   Assume that: 1) $(X_n)_{n \in \mathbb{N}}$ are all integrable (finite-mean) random variables,
%   2) $\mathrm{E}\left[X_{n} 1_{\{N \geq n\}}\right]=\mathrm{E}\left[X_{n}\right] \mathrm{P}(N \geq n)$
%   for every natural number $n$,
%   3) the infinite series satisfies 
%   $\sum_{n=1}^{\infty} \mathrm{E}\left[\left|X_{n}\right| 1_{\{N \geq n\}}\right]<\infty$.
%   Then the random sums $S_{N}=\sum_{n=1}^{N} X_{n}$ and 
%   $T_{N}=\sum_{n=1}^{N} \mathrm{E}\left[X_{n}\right]$ are integrable and
%   $\mathrm{E}\left[S_{N}\right]=\mathrm{E}\left[T_{N}\right]$.
% \end{lemma}
\section{Highly Adaptive $0.25$-Approximation} \label{sec:adaptive-linear}
In this section, we show that Alg. \ref{alg:adaptive-linear} achieves
approximation ratio of $1/4$ in $n$ adaptive queries to the objective function $f$. Alg. \ref{alg:adaptive-linear} influences the design of \linearseq,
which uses similar ideas to obtain its constant ratio. See the discussion
in Section \ref{section:fastlinear}.

\textbf{Description of Alg. \ref{alg:adaptive-linear}.} This algorithm operates
in one \textbf{for} loop through the ground set. Each element $u$ is added to
a set $A$ iff. $\marge{u}{A} \ge f(A) / k$. The solution returned is the set
$A'$ of the last $k$ elements added to $A$.
\begin{algorithm}[t]
	\caption{Highly Adaptive Linear-Time Algorithm}
	\label{alg:adaptive-linear}
	\begin{algorithmic}[1]
	\State \textbf{Input:} evaluation oracle $f:2^{\mathcal N} \to \reals$, constraint $k$

	\State Initialize $A \gets \emptyset$ 
	\For{ $u \in \mathcal N$ } 
        \If{$\marge{u}{A} \ge \ff{A}/k$}
        \State $A \gets A \cup \{ u \}$
        \EndIf
        \EndFor
	\State \textbf{return} $A' \gets \{ \text{last } k \text{ elements added to } A \}$
\end{algorithmic}
\end{algorithm}
\begin{theorem} 
  \label{thm:adaptive-linear} Let $O$ be an optimal solution to \sm on instance $(f,k)$.
  Then the set $A'$ produced after running Alg. \ref{alg:adaptive-linear} on this instance
  satisfies $4f(A') \ge f(O)$.
\end{theorem}
\begin{proof}
  \begin{lemma} \label{lemm:al-1}
    At termination of Alg. \ref{alg:adaptive-linear},
    $f(O) \leq 2f(A)$.
  \end{lemma}
  \begin{proof}
    For each $o \in O\backslash A$, let $j(o)+1$ be the iteration
    of the \textbf{for} loop in which $o$ is processed,
    and $A_{j(o)}$ be $A$ after 
    iteration $j(o)$.  Thus, $\marge{o}{A_{j(o)}}) < f(A_{j(o)})/k$. 
    Then
\begin{align*}
f(O)-f(A) & \le f(O \cup A) - f(A) \numberthis \label{ineq:al-o1}\\
          & \le \sum_{o \in O\backslash A}\marge{o}{A} \numberthis \label{ineq:al-o2}\\
          & \le \sum_{o \in O\backslash A} \marge{o}{A_{j(o)}} \numberthis \label{ineq:al-o3}\\
          & \le \sum_{o \in O\backslash A} f(A_{j(o)})/k \\
          & \le f(A)\numberthis \label{ineq:al-o4},
\end{align*} 
where Inequalities \ref{ineq:al-o2} and \ref{ineq:al-o3} follow from submodularity, 
and Inequalities \ref{ineq:al-o1} and \ref{ineq:al-o4} follow from monotonicity.
\end{proof}
  \begin{lemma} \label{lemm:al-2}
    At termination of Alg. \ref{alg:adaptive-linear},
    $f(A) \leq 2f(A')$.
  \end{lemma}
  \begin{proof}
    If $|A| < k$, then $A' = A$ and there is nothing
    to show. So suppose $|A| \ge k$. Let $A'_i = \{a_1, \ldots, a_{i - 1} \}$.
    \begin{align*}
      f(A) - f(A \setminus A') &= \sum_{i = 1}^{k} \marge{a_i}{A \setminus A' \cup A'_i} \\
                               &\ge \sum_{i=1}^k \ff{ (A \setminus A') \cup A'_i } / k \\
                               &\ge \sum_{i=1}^k \ff{A \setminus A'} / k \\
                               &= \ff{A \setminus A'}.
    \end{align*}
    Therefore, $\ff{A \setminus A'} \le f(A) / 2.$ Now,
    \begin{align*}
      f(A') &\ge f(A) - f(A \setminus A') \numberthis \label{ineq:al-sm} \\
            &\ge f(A) - f(A) / 2 = f(A) / 2,
    \end{align*}
    where Inequality \ref{ineq:al-sm} is by submodularity of $f$.
  \end{proof}
  By Lemma \ref{lemm:al-1} and \ref{lemm:al-2}, we have
  $$f(O) \le 2f(A) \le 4f(A').$$
\end{proof}
\section{Analysis of \linearseq}
\label{apdix:fastlinear}
\subsection{Proof of Lemma \ref{lemma:fastSubOne}} 
\fastSubOne*
\begin{proof}
For each $o \in O\backslash A$, let $j(o)+1$ be the iteration 
where $o$ is filtered out, and $A_{j(o)}$ be $A$ after 
iteration $j(o)$.  Thus, $\marge{o}{A_{j(o)}}) < f(A_{j(o)})/k$. 
Then
\begin{align*}
f(O)-f(A) & \le f(O \cup A) - f(A) \numberthis \label{ineq:ls-o1}\\
          & \le \sum_{o \in O\backslash A}\marge{o}{A} \numberthis \label{ineq:ls-o2}\\
          & \le \sum_{o \in O\backslash A} \marge{o}{A_{j(o)}} \numberthis \label{ineq:ls-o3}\\
          & \le \sum_{o \in O\backslash A} f(A_{j(o)})/k \\
          & \le f(A)\numberthis \label{ineq:ls-o4},
\end{align*} 
where Inequalities \ref{ineq:ls-o2} and \ref{ineq:ls-o3} follow from submodularity, 
and Inequalities \ref{ineq:ls-o1} and \ref{ineq:ls-o4} follow from monotonicity.
\end{proof}
\subsection{Probability \linearseq is Successful} \label{apx:ls-success}
\begin{proof}
Let $Y_\ell$ be the number of successes in $\ell$ independent Bernoulli random
variables with success probability 1/2. Then,
\begin{align*}
\prob{\text{Algorithm \ref{alg:fastlinear} succeeds}}
& \ge \prob{Y_\ell \geq m} \numberthis \label{ineq:ls-prob1}\\
& \geq \prob{Y_\ell \geq \log(n) /(\beta \epsi)} \\
&=1- \prob{Y_\ell \leq \log(n) /(\beta \epsi)} \\
&\ge 1-e^{-\frac{1}{2}\left(\frac{2\beta \epsi +1}{2(\beta \epsi +1)}\right)^2
\cdot 2\left(1+\frac{1}{\beta \epsi}\right) \log(n)} \numberthis \label{ineq:ls-prob2}\\
& = 1-(\frac{1}{n})^{\frac{(2\beta \epsi+1)^2}{4\beta \epsi(\beta \epsi+1)}} \\
& \geq 1-\frac{1}{n},
\end{align*}
where Inequalities \ref{ineq:ls-prob1} and \ref{ineq:ls-prob2} follow from Lemma \ref{lemma:indep} and 
Lemma \ref{lemma:chernoff}, respectively.
\end{proof}
\subsection{Proof of Claim \ref{claim:LSblockOne}} \label{apx:claim1}
\LSblockOne*
\begin{proof}
In this proof, let $T_{j,i} = (v_1,\ldots,v_i)$ be the prefix
of $V$ of length $i$ at iteration $j$. Likewise, define
$T'_{j,i} = T_{j,i} \setminus T_{j,i -1}$.
Say block $T_{j,k}'$ is \textit{bad} if this block does not 
satisfy the condition in Line \ref{line:fastIf2} during
iteration $j$.

First, consider the case that $j>c$.
It holds that $T_{j,\lamb_j^*} \subseteq A'$. 
Thus, $|T_{j,\lamb_j^*}| \le k$.
If $|T_{j,\lamb_j^*}|=0$, the result holds \color{\rev}
immediately since $\marge{T_{j,\lamb_j^*}}{A_{j-1}}=0 = 
\frac{1-2\epsi}{1+\epsi}|T_{j,\lamb_j^*}|\cdot f(A\backslash A')/k$.
If $0<|T_{j,\lamb_j^*}|\le k$, consider two cases of $\lamb_j^*$
regarding its selection on Line~\ref{line:rule} in Alg.~\ref{alg:fastlinear}.

If $\lamb_j^* \le \left\lceil\frac{1}{\epsi} \right\rceil$,
all the blocks up to and including $T_{j,\lamb_j^*}'$ are good.
So,
\begin{align*}
\marge{T_{j,\lamb_j^*}}{A_{j-1}} 
&= \sum_{\lamb_i \le \lamb_j^*} \marge{T_{j,\lamb_i}\setminus T_{j,\lamb_{i-1}}}{A_{j-1} \cup T_{j,\lamb_{i-1}}}\\
&= \sum_{\lamb_i \le \lamb_j^*} \marge{T_{j,\lamb_i}'}{A_{j-1} \cup T_{j,\lamb_{i-1}}}\\
&\ge \sum_{\lamb_i \le \lamb_j^*}(1-\epsi) |T_{j,\lamb_i}'|\cdot f(A_{j-1} \cup T_{j,\lamb_{i-1}})/k\\
&\ge (1-\epsi)|T_{j,\lamb_j^*}|\cdot f(A\backslash A')/k,
\end{align*}
where the second from the last inequality follows from the property of good blocks,
and the last inequality follows from monotonicity.

If $\left\lceil\frac{1}{\epsi} \right\rceil < \lamb_j^* \le k$, 
the last block in $T_{j,\lamb_j^*}$ is bad while all the previous ones are good.
Let $|T_{j,\lamb_j^*}| = \lamb_j^* = \left\lfloor (1+\epsi)^{u+1} \right \rfloor$.
It holds that 
\begin{equation}\label{ineq:ls-prefix}
\frac{|T_{j,\lamb_j^*}\setminus T_{j,\lamb_j^*}'|}{|T_{j,\lamb_j^*}|} 
= \frac{\left\lfloor (1+\epsi)^u \right \rfloor}{\left\lfloor (1+\epsi)^{u+1} \right \rfloor}
\ge \frac{(1+\epsi)^u -1}{ (1+\epsi)^{u+1}}
\ge \frac{1}{1+\epsi} -\epsi,
\end{equation}
where the last inequality follows from $(1+\epsi)^{u+1} \ge \lamb_j^*>\left\lceil \frac{1}{\epsi}\right\rceil \ge \frac{1}{\epsi}$.
\color{black}
Then,
\begin{align*}
\marge{T_{j,\lamb_j^*}}{A_{j-1}} 
&=\color{\rev} \sum_{\lamb_i \le \lamb_j^*} \marge{T_{j,\lamb_i}\setminus T_{j,\lamb_{i-1}}}{A_{j-1} \cup T_{j,\lamb_{i-1}}}\\
&= \sum_{\lamb_i \le \lamb_j^*} \marge{T_{j,\lamb_i}'}{A_{j-1} \cup T_{j,\lamb_{i-1}}}\\
&\ge \sum_{\lamb_i < \lamb_j^*} \marge{T_{j,\lamb_i}'}{A_{j-1} \cup T_{j,\lamb_{i-1}}}\numberthis \label{ineq:ls-marge1}\\
&\ge \sum_{\lamb_i < \lamb_j^*}(1-\epsi) |T_{j,\lamb_i}'|\cdot f(A_{j-1} \cup T_{j,\lamb_{i-1}})/k \numberthis \label{ineq:ls-marge2}\\
&\ge (1-\epsi)|T_{j,\lamb_j^*}\backslash T_{j,\lamb_j^*}'|\cdot f(A\backslash A')/k \numberthis \label{ineq:ls-marge3}\\
&\ge \textcolor{\rev}{\frac{1-2\epsi}{1+\epsi}}|T_{j,\lamb_j^*}|\cdot f(A\backslash A')/k,\numberthis \label{ineq:ls-marge7}
\end{align*}
where Inequalities \ref{ineq:ls-marge1} and \ref{ineq:ls-marge3} follow from monotonicity, 
Inequality \ref{ineq:ls-marge2} follows from the property of good block,
\color{\rev}
and Inequality \ref{ineq:ls-marge7} follows from Inequality~\ref{ineq:ls-prefix}.
Therefore, the claim holds for $j > c$.
\color{black}

Secondly, consider the case $j = c$.
\color{\rev}
In this case, block $T_{c,\lamb_c^*}'$ can be a good or bad block.
Consider blocks before $T_{c,\lamb_c^*}'$.
By the selection of $\lamb_c^*$ on Line~\ref{line:rule},
if $|T_{c,\lamb_c^*}| \le k$, 
all the blocks before $T_{c,\lamb_c^*}'$ are good;
if $|T_{c,\lamb_c^*}| > k$,
several consecutive blocks before $T_{c,\lamb_c^*}'$ are good total with at least $k$ elements.
\color{black}
Let $\lamb_v = \min\{\lamb_v \in \Lamb:
\cup_{\lamb_v \le \lamb_i \le \lamb_c^*} T_{c,\lamb_i}' \subseteq A'\}$. 
\color{\rev}
Since $|T_{c,\lamb_c^*}' \cap A'| \le k$,
all the blocks with indices $\lamb_v \le \lamb_i \le \lamb_c^*$ are good.
\color{black}
For any $\lamb_i\in \Lamb$, it holds that $|T_{c,\lamb_i}'| \le \epsi k$. Then,
\begin{align*}
\marge{T_{c,\lamb_c^*}}{A\backslash A'} 
&=\sum_{\lamb_i \le \lamb_c^*} \marge{T_{c,\lamb_i}'}{A\backslash A'\cup T_{c, \lamb_{i-1}}} \\
&\ge \sum_{\lamb_v \le \lamb_i < \lamb_c^*} \marge{T_{c,\lamb_i}'}{A\backslash A'\cup T_{c, \lamb_{i-1}}} \numberthis \label{ineq:ls-marge4}\\
&\ge (1-\epsi) \sum_{\lamb_v \le \lamb_i < \lamb_c^*} |T_{c,\lamb_i}'|\cdot f(A\backslash A'\cup T_{c, \lamb_{i-1}}) /k \numberthis \label{ineq:ls-marge5}\\
& \ge (1-\epsi)\max\left\{0, \left(|T_{c,\lamb_c^*}\cap A'| - 
|T_{c,\lamb_{v-1}}'| - |T_{c,\lamb_c^*}'|\right)\right\}\cdot f(A\backslash A')  /k \numberthis \label{ineq:ls-marge6}\\
& \ge (1-\epsi)\max\{0,|T_{c,\lamb_c^*}\cap A'|-2\epsi k\}\cdot f(A\backslash A')/k,
\end{align*}
where Inequality \ref{ineq:ls-marge4} follows from monotonicity,
Inequality \ref{ineq:ls-marge5} follows from the property of good block,
Inequality \ref{ineq:ls-marge6} follows from monotonicity, 
$\cup_{\lamb_v \le \lamb_i < \lamb_c^*} T_{c,\lamb_i}' \subseteq \left(T_{c,\lamb_c^*}\cap A'\right)$,
and $$\left(\left(T_{c,\lamb_c^*}\cap A'\right)\backslash \cup_{u=v}^{i-1} T_{c,\lamb_i}' \right)\subseteq 
\left(T_{c,\lamb_{v-1}}' \cup T_{c,\lamb_c^*}'\right).$$
\end{proof}
\subsection{Proof of Inequality \ref{ineq:claim1}} \label{apx:ineq-claim1}
\begin{proof}[Proof of Inequality \ref{ineq:claim1}]
\begin{align*}
&f(A)-f(A \backslash A') 
= \marge{T_{c,\lamb_c^*}}{A\backslash A'} + 
\sum_{j=c+1}^{\ell} \marge{T_{j,\lamb_j^*}}{A_{j-1}} \\
& \geq (1-\epsi)\max\{0,|T_{c,\lamb_c^*}\cap A'|-2\epsi k \}\cdot f(A\backslash A')/k
+ \textcolor{\rev}{\frac{1-2\epsi}{1+\epsi}}\left(k-|T_{c,\lamb_c^*}\cap A'|\right)\cdot f(A\backslash A')/k\\
% & = (1-\epsi)\left(\max \left\{2\epsi-\frac{|T_{c,\lamb_c^*}\cap A'|}{(1+\epsi)k},
% \frac{\epsi|T_{c,\lamb_c^*}\cap A'|}{(1+\epsi)k} \right\} + \frac{1}{1+\epsi}-2\epsi \right)\cdot f(A\backslash A')\\
&\ge \textcolor{\rev}{\frac{(1-2\epsi)^2}{1+\epsi}}\cdot f(A\backslash A')
\end{align*}
\color{\rev}
where the second from the last inequality follows from Claim~\ref{claim:LSblockOne},
and equality in the last inequality holds if and only if $|T_{c,\lamb_c^*}\cap A'| = 2\epsi k$.
\color{black}
\end{proof}
\subsection{Proof of Lemma \ref{lemma:FLSnumBad}}
\FLSnumBad*
\begin{proof}
After Line \ref{line:fastFilterV}, for any $x \in V$, 
$\marge{x}{A} \geq f(A)/k$, $\marge{x}{V \cup A} = 0$. 
Since $T_0=\emptyset$ and 
$T_{|V|}=V$,  
\begin{alignat*}{5}
&|S_0| &=&|\{x \in V: \marge{x}{A} <f(A)/k\}| &=& 0,\\
&|S_{|V|}|& =& |\{x \in V: \marge{x}{A \cup V} <f(A\cup V)/k\}| &=&|V|.
\end{alignat*}
Due to submodularity and monotonicity, for any $x\in S_i$,
\begin{equation*}
\marge{x}{A\cup T_{i+1}} \leq \marge{x}{A\cup T_i} 
< f(A \cup T_i)/k \leq f(A \cup T_{i+1})/k.
\end{equation*}

So, $x \in S_{i+1}$, and $S_i$ is the subset of $S_{i+1}$, which means 
$|S_i| \leq |S_{i+1}|$. 
\end{proof}

\subsection{Proof of Inequality \ref{ineq:linear-fail-1}} \label{apx:ls-fail}
\begin{lemma} \label{lemma:blockProb}
  Let $(Y_i)$ be a sequence of independent and identically distributed Bernoulli trials,
  where the success probability is $\beta\epsi$.
  Then for a constant integer $\alpha$,
  $\prob{\sum_{i=1}^\alpha Y_i > \epsi \alpha}\le 
  \min\{\beta, e^{-\frac{(1-\beta)^2}{1+\beta}\epsi\alpha}\}$.
\end{lemma}
\begin{proof}
  Since $(Y_i)$ is a sequence of i.i.d. Bernoulli trails with success probability $\beta \epsi$,
  it holds that $\ex{\sum_{i=1}^\alpha Y_i}=\beta \epsi \alpha$.
  For small $\alpha$, by Markov's inequality, we can bound the probability as follows,
  $$\prob{\sum_{i=1}^\alpha Y_i > \epsi \alpha}
  \le \frac{\ex{\sum_{i=1}^\alpha Y_i}}{\epsi \alpha}=\beta.$$
  For large $\alpha$, there exists a tighter bound by the application of Lemma \ref{lemma:chernoff},
  $$\prob{\sum_{i=1}^\alpha Y_i > \epsi \alpha} \le e^{-\frac{(1-\beta)^2}{1+\beta}\epsi\alpha}.$$
\end{proof}

\color{\rev}
Here, we restate Inequality \ref{ineq:linear-fail-1} and provide its analysis below.
\begin{claim}\label{claim:linear-fail-1}
At iteration $j$ of outer for loop on Line~\ref{line:fastOuterForStart} in Alg.~\ref{alg:fastlinear},
let $B_1 = \{ \lamb \in \Lambda : \lamb \le k \text{ and } \lamb \textcolor{\rev}{\le} \lamb_t \}$, 
$B_2 = \{ \lamb \in \Lambda : | \Lambda \cap [ \lamb, \lamb_t ] | \le \lceil 1/ \epsi \rceil \}$.
It holds that
\begin{equation*}
  \prob{\text{iteration $j$ fails}  }
  \leq \prob{\exists \lambda \in B_1 \cup B_2 \text{ with }B[\lamb] = \textbf{false}}  \le 1/2.
\end{equation*}
\end{claim}
\textbf{Proof Overview.}
The goal of the proof is to bound the probability of the event 
($\exists \lambda \in B_1 \cup B_2 \text{ with }B[\lamb] = \textbf{false}$),
which is composed of a random number of sub-events.
Directly bounding this probability can be challenging,
so we construct an alternative event that encompasses the original one
but consists of a fixed number of sub-events.
In detail, the sizes of $B_1$ and $B_2$ are random and depend on the random variable $t$.
To address this, we enlarge $B_1$ to include all indices less than $k$
and expand $B_2$ to contain exactly $\left\lceil 1/\epsi \right\rceil$ indices.
Recall that $t = \min\{i \in \mathbb N: |S_i| \ge \beta \epsi |V|\}$.
Thus, for each $\lamb \ge t$, $B[\lamb] = \textbf{false}$ with probability greater than $\beta \epsi$.
To maintain the upper bound for each sub-event beyond $t$,
we construct another sequence of dependent Bernoulli trials
where the success probability for each trial is at most $\beta \epsi$.
This allows us to bound the probability of the original event 
by considering a more tractable event.
\begin{proof}[Proof of Claim~\ref{claim:linear-fail-1} (Inequality \ref{ineq:linear-fail-1})]
% Since $B_1$ and $B_2$ have random sizes, making the probability hard to analyze,
% we consider two other fixed-size sets $B'_1$ and $B'_2$ with another dependent Bernoulli trials $(X_i)$. 

Recall that, the random permutation of $V$ can be regarded as 
$|V|$ dependent Bernoulli trials, 
where $\prob{v_i \text{ is bad} | v_1, \ldots, v_{i-1}} = |S_{i-1}|/|V|$.
Here, $S_{i-1}$ is the set of elements that would be filtered at the next iteration
if prefix $T_{i-1}$ is added to the solution.
Thus, $S_{|V|} = V$ and we further define $S_i = V$ when $i > |V|$.

Instead of analyzing this sequence directly, 
we construct another sequence of Bernoulli trials $\{X_i\}_{i=1}^{\infty}$,
where $X_i = 1_{v_i \text{ is bad}}$ when $i \le t$,
and $X_i = 1$ with probability $\beta\epsi$ when $i > t$.
So, the success probability for each $X_i$ equals $\min\{|S_{i-1}|/|V|, \beta\epsi\}$,
where $S_i = V$ when $i > |V|$.
Let $\Lamb' = \left[\left\lceil \frac{1}{\epsi} \right\rceil\right] \cup\{\lfloor(1+\epsi)^u \rfloor: 
      1 \leq \lfloor(1+\epsi)^u \rfloor \leq k, 
      u \in \mathbb{N}\} \cup 
      \{\lfloor k+u\epsi k \rfloor: u \in \mathbb{N}\}$.
Then, it holds that $\Lamb \subseteq \Lamb'$.
For each $\lamb_i\in \Lamb'$,
define $B'[\lamb_i] =\textbf{false}$ if there are more than $\epsi(\lamb_i-\lamb_{i-1})$ successes in $\{X_{\lamb_{i-1}+1}, \ldots, X_{\lamb_i}\}$;
otherwise, define $B'[\lamb_i] =\textbf{true}$.
So, when $\lamb_i \le t$, it holds that $B[\lamb_i] =B'[\lamb_i]$.
Then, we construct two additional index sets $B_1'$ and $B_2'$, which
include $B_1$ and $B_2$ but have fixed sizes.
Let $B'_1 = \{\lamb \in \Lamb': \lamb \le k\}$
and $B'_2 = B_2\cup \{\lamb \in \Lamb': |\Lamb' \cap (\lamb_t, \lamb]| \le \lceil 1/\epsi\rceil -|B_2|\}$.
Then,
\begin{align*}
  \prob{\text{iteration $j$ fails}} &\le \prob{\exists \lambda \in B_1 \cup B_2 \text{ with }B[\lamb] = \textbf{false}}\\
  &= \prob{\exists \lambda \in B_1 \cup B_2 \text{ with }B'[\lamb] = \textbf{false}}\numberthis \label{ineq:ls-Bprime}\\
  &\le \prob{\exists \lambda \in B'_1 \cup B'_2 \text{ with }B'[\lamb] = \textbf{false}} \numberthis \label{ineq:ls-new}\\
  &\le \prob{\exists \lambda \in B'_1 \text{ with } B'[\lamb] = \textbf{false}}+\prob{\exists \lambda \in B'_2 \text{ with } B'[\lamb] = \textbf{false}},\numberthis \label{ineq:apx-lin-1}
\end{align*}
where Inequality~\ref{ineq:ls-Bprime} follows from $\max (B_1\cup B_2) \le t$
and $B[\lamb] =B'[\lamb]$ for each $\lamb \le t$,
Inequality~\ref{ineq:ls-new} follows from $B_1\subseteq B'_1$
and $B_2\subseteq B'_2$.
Let $(Y_j)$ be a sequence of independent and identically distributed Bernoulli trials,
where the success probability is $\beta \epsi$.
% where Equation \ref{ineq:apx-lin-1} holds, since the sequence $(1_{\{B[\lamb]=\textbf{false}\}})_\lamb$
% and the random variable $t$ follow the assumptions in Lemma \ref{lemma:wald}:
% 1) $1_{\{B[\lamb]=\textbf{false}\}}$s are all integrable random variables, 
% because they only take the value 0 and 1;
% 2) $t$ is a stopping time since it only depends on the previous $t-1$ selections;
% 3) $\prob{1_{\{t \ge n\}}=0}=1$ for any $n \ge |V|$.
% Furthermore, by Lemma~\ref{lemma:oneBadBlock}, we get the following corollary.
% \begin{corollary}\label{cor:oneBadBlock}
% Let $(X_i)$ be an infinite sequence of dependent Bernoulli trials where
% $\prob{X_i = 1}\le \beta\epsi$;
% $(Y_i)$ be a sequence of independent and identically distributed Bernoulli trials,
% where the success probability is $\beta\epsi$.
% Then, for any fixed $\lamb_i \in \Lamb'$,
% \[\prob{B'[\lamb_i] = \textbf{false}}
% \le \prob{\sum_{j=\lamb_{i-1}+1}^{\lamb_i} Y_j > \epsi(\lamb_i-\lamb_{i-1})}.\]
% \end{corollary}

To bound the first term of Inequality~\ref{ineq:apx-lin-1},
we have,
\begin{align*}
  &\prob{\exists \lambda \in B'_1 \text{ with } B'[\lamb] = \textbf{false}}
  \le \sum_{\lamb_i\in B'_1} \prob{ \sum_{j=\lamb_{i-1}+1 }^{\lamb_i}X_j > \epsi (\lamb_i-\lamb_{i-1})}\\
  &\le \sum_{\lamb_i\in B'_1} \prob{ \sum_{j=\lamb_{i-1}+1 }^{\lamb_i}Y_j > \epsi (\lamb_i-\lamb_{i-1})}\numberthis \label{ineq:ls-block6}\\
  &\le \sum_{\lamb_i \in \left[\left\lceil \frac{1}{\epsi} \right\rceil\right]\cup \{\lfloor(1+\epsi)^u\rfloor:u\ge 1\}} \prob{\sum_{j=\lamb_{i-1}+1 }^{\lamb_i}Y_j > \epsi (\lamb_i-\lamb_{i-1})} \numberthis \label{ineq:ls-block5}\\
  & \le \sum_{\lamb_i \in \left[\left\lceil \frac{1}{\epsi} \right\rceil\right]\cup \{\lfloor(1+\epsi)^u\rfloor:u\ge 1\}} \min\{\beta, e^{-\frac{(1-\beta)^2}{1+\beta}\epsi (\lamb_i-\lamb_{i-1})}\}\numberthis \label{ineq:ls-block4} \\ 
  & \le \beta \left\lceil \frac{1}{\epsi} \right\rceil+ \sum_{\lamb \ge 1} \min\{\beta, e^{-\epsi \lamb/2}\} \numberthis \label{ineq:ls-block3} \\
  & \le \beta \left( \left\lceil \frac{1}{\epsi} \right\rceil+a\right) + \sum_{\lamb = a+1}^\infty e^{-\epsi \lamb/2} 
\text{, where } a= 
\left \lfloor\frac{2}{\epsi} 
\log\left(\frac{8}{1-e^{-\epsi/2}}\right)\right \rfloor
\le \frac{1}{8\beta}-\left \lfloor \frac{1}{\epsi} \right \rfloor\\
& \le  \frac{1}{8} + \frac{e^{-\epsi(a+1)/2}}{1-e^{-\epsi/2}}\\
& \le \frac{1}{8} + \frac{1}{8} = \frac{1}{4},
\end{align*}
\color{black}
where Inequality~\ref{ineq:ls-block6} follows from Lemma~\ref{lemma:indep2};
\color{\rev}
Inequality \ref{ineq:ls-block5} follows from $B'_1 \subseteq \left[\left\lceil \frac{1}{\epsi} \right\rceil\right] \cup \{\lfloor(1+\epsi)^u\rfloor:u\ge 1\}$;
\color{black}
Inequality \ref{ineq:ls-block4} follows from Lemma \ref{lemma:blockProb};
and Inequality \ref{ineq:ls-block3} follows from $\lamb_i-\lamb_{i-1} < \lamb_i$ and $\beta < 1/5$.

For the second term in Inequality \ref{ineq:apx-lin-1}, 
\color{\rev}
from Lemma~\ref{lemma:indep2},
Lemma \ref{lemma:blockProb},
and the fact that $|B'_2| = \lceil 1 / \epsi \rceil$, we have
$$\prob{\exists \lambda \in B'_2 \text{ with } B'[\lamb] = \textbf{false}}
\le \sum_{\lamb_i \in B'_2} \prob{\sum_{j=\lamb_{i-1}+1 }^{\lamb_i}Y_j > \epsi (\lamb_i-\lamb_{i-1})} 
\le \beta \lceil 1 /\epsi \rceil< \frac{1}{4},$$
where the last inequality follows from
\[1\le \left \lfloor\frac{2}{\epsi} 
\log\left(\frac{8}{1-e^{-\epsi/2}}\right)\right \rfloor
\le \frac{1}{8\beta}-\left \lfloor \frac{1}{\epsi} \right \rfloor.\]
\color{black}
% \begin{align*}
%   &\ex{\sum_{\lamb \in B_2} \ex{1_{\{B[\lamb]=\textbf{false}\}}}}
%   \le \ex{\sum_{\lamb \in B_2} \prob{\sum_{i=1 }^{|T_\lamb'|}Y_i > \epsi |T_\lamb'|} }\\
%   & \le \beta \lceil 1 /\epsi \rceil \\
%   & \le \frac{1+\epsi}{16\log\left(\frac{8}{1-e^{-\epsi/2}}\right)} \le \frac{1}{4}.
% \end{align*}

From the above two inequalities, we have that 
$$\prob{\text{iteration $j$ fails}} \le 1/2.$$
\end{proof}
\subsection{Query Complexity of \linearseq} \label{apx:ls-query}
\begin{proof}[Query Complexity]
  Let $V_j$ be the value of $V$ after filtering on Line \ref{line:fastFilterV}
  during iteration $j$,
  and let $Y_i$ be the number of iterations between the $(i-1)$-th success 
  and $i$-th success.
  By Lemma \ref{lemma:indep} in Appendix \ref{apx:prob}, 
$\ex{Y_i} \le 2$.
  Observe that $|V_0| = |\mathcal{N}| = n$.
Then, the expected number of queries is bounded as follows:
\begin{align*}
&\ex{Queries} \leq \sum_{j=1}^\ell \ex{|V_{j-1}|+2|V_j|/(\epsi k)+
        2\log_{1+\epsi} (k)-2/\epsi + 7} \\
%&\le \left(\frac{2}{\epsi k}+1\right)\sum_{j=1}^{\ell} \ex{|V_j|}+n
%       + \left(2\log_{1+\epsi}(k) -\frac{2}{\epsi}+7\right)\ell \\
       &\le \left(\frac{2}{\epsi k}+1\right)\sum_{i=1}^{\infty} 
       \ex{Y_i(1-\beta \epsi)^i n}+n
       + \left(2\log_{1+\epsi}(k) -\frac{2}{\epsi}+7\right)\ell \\
       & \leq \left(1+\left(\frac{2}{\beta \epsi}-2\right)\left(\frac{2}{\epsi k}+1\right)\right)n + 4\left(1-\frac{1}{\beta \epsi}\right)
       \left(\frac{2(1+\epsi)}{\epsi}\log( k)-\frac{2}{\epsi}+7\right)\log(n).
\end{align*}
Thus, the total queries are $\oh{ (1/(\epsi k)+1)n /\epsi^3 } = \oh{ n / \epsi^3 }$
in expectation.
\end{proof}

%%% Local Variables:
%%% mode: latex
%%% TeX-master: "main.tex"
%%% End:

\section{Description and Analysis of \threshold}
\label{apdix:threshold}
\textbf{Description.} The algorithm takes as input oracle $f$, 
size constraint $k$, accuracy parameter $\epsi > 0$, and probability
parameter $\delta$ which influences the failure probability of
at most $\delta / n$.
The algorithm works in iterations of a sequential outer \textbf{for}
loop of length at most $\oh{\log n}$; a set $A$ is initially empty,
and elements are added to $A$ in each iteration of the \textbf{for} loop.
Each iteration has four parts: filtering low value elements
from $V$ (Line \ref{line:threFilterV}), randomly permuting $V$ (Line \ref{line:threPermute}),
computing in parallel the marginal gain of adding blocks of elements of $V$ to $A$
(Line \ref{line:innerforStart}), and adding a block slightly larger than the largest block
that had average gain at least $(1 - \epsi)\tau$ (Line \ref{line:updateA}).
% \begin{align*}
% \prob{\text{Algorithm \ref{alg:threshold} succeeds}}
% &\geq \prob{X_\ell \geq m} \ge 1-\frac{1}{n}, \numberthis \label{ineq:linear-success}
% \end{align*}
% where the proof of Inequality \ref{ineq:linear-success} is in
% Appendix \ref{apdix:fastlinear}.
\begin{algorithm}[t]
	\caption{A Parallelizable Greedy Algorithm for Fixed Threshold $\tau$}
	\label{alg:threshold}
	\begin{algorithmic}[1]
	\Procedure{\threshold}{$f, \mathcal N, k, \delta, \epsi, \tau$}

	\State \textbf{Input:} evaluation oracle $f:2^{\mathcal N} \to \reals$, constraint $k$, revision $\delta$, error $\epsi$, threshold $\tau$

	\State Initialize $A \gets \emptyset$ , $V \gets \mathcal N$, 
	$\ell = \lceil 4(1+2/\epsi)\log (n/\delta) \rceil$ 

	\For{ $j \gets 1$ to $\ell$ } \label{line:threForStart}

		\State Update $V \gets \{ x \in V : \marge{x}{A} \ge \tau \}$ 
		and filter out the rest \label{line:threFilterV} 

		\If{ $|V| = 0$ } 
			
			\State \textbf{return} $A$
		
		\EndIf
		
		\State $V \gets $ \textbf{random-permutation}$(V)$. \label{line:threPermute}

		\State $s \gets \min \{k-|A|, |V|\}$
		\color{\rev}
		\State $\Lambda \gets \left[\min\{s, \left\lceil \frac{1}{\epsi} \right\rceil\}\right] \cup \{ \lfloor (1+\epsi)^u \rfloor: 
		1\leq \lfloor (1+\epsi)^u \rfloor \leq s
		, u \in \mathbb{N}\} \cup \{s\}$ \label{line:lambda}
		\color{black}
		\State $B \gets \emptyset$

		\For{$\lamb_i \in \Lambda$ in parallel} \label{line:innerforStart}

			\State $T_{\lamb_i} \gets \{v_1, v_2, \ldots, v_{\lamb_i}\}$ 

			\If{$ \marge{T_{\lamb_i}}{A}/|T_{\lamb_i}| \geq  (1-\epsi) \tau $} \label{line:threIf2}

				\State $B \gets B \cup \{\lamb_i\}$
			
			\EndIf

		\EndFor\label{line:innerforEnd}
		\color{\rev}
		\State \textbf{if} $\max B \le \min\{s, \left\lceil \frac{1}{\epsi} \right\rceil\}$ \textbf{then} $\lamb^* \gets \max B$
		\State \textbf{else} $\lamb^* \gets \min\{\lamb_i \in \Lambda: \lamb_i >b, \forall b \in B \}$
		\color{black}
		\State  $A \gets A \cup T_{\lamb^*}$  \label{line:updateA}
		\If{ $|A| = k$ }

			\State \textbf{return} $A$ 

		\EndIf

	\EndFor \label{line:threForEnd}

	\State \textbf{return} \textit{failure}

	\EndProcedure
\end{algorithmic}
\end{algorithm}
\subsection{Proof of Theorem \ref{theorem:threshold}}
\subsubsection{Success Probability}
The algorithm \threshold will successfully terminate once $V$ is empty or $|A| = k$. 
If it fails to terminate, the outer for loop runs $\ell$ iterations;
it also holds that $|V|>0$ and $|A|<k$ at termination.

Fix an iteration $j$ of the outer \textbf{for} loop;
for this part, we will condition on the behavior
of the algorithm prior to iteration $j$.
% We will say iteration $j$ \textit{succeeds} iff
% an $\epsi / 2$ fraction of $V$ is filtered at the
% next iteration (iteration $j + 1$).
% Lemma \ref{lemma:numBad} shows an analogous result to that of Lemma \ref{lemma:FLSnumBad}.
\begin{restatable}{lemma}{numBad} 
	\label{lemma:numBad}
At an iteration $j$, let 
$S_i = \left\{x \in V: \marge{x}{T_i\cup A} <\tau\right\}$ 
after Line \ref{line:threFilterV}. 
It holds that $|S_0|= 0$, $|S_{|V|}| = |V|$, 
and $|S_i| \leq |S_{i+1}|$.
\end{restatable}
% \numBad*
\begin{proof}
After Line \ref{line:threFilterV}, for any $x \in V$, 
$\marge{x}{A} \geq \tau$, $\marge{x}{V \cup A} = 0$. 
Since $T_0=\emptyset$ and 
$T_{|V|}=V$,  
\begin{alignat*}{5}
&|S_0| &=&|\{x \in V: \marge{x}{A} <\tau\}| &=& 0,\\
&|S_{|V|}|& =& |\{x \in V: \marge{x}{V\cup A} <\tau\}| &=&|V|.
\end{alignat*}
Due to submodularity, 
$\marge{x}{T_i \cup A} > \marge{x}{T_{i+1} \cup A}$. 
So, $S_i$ is the subset of $S_{i+1}$, which means 
$|S_i| \leq |S_{i+1}|$.
\end{proof}
From Lemma \ref{lemma:numBad},
there exists a $t$ such that 
$t = \min\{i \in \mathbb{N}: |S_i| \ge  \epsi |V|/2\}$.
If $\lamb^* \ge t$, we will successfully discard at least
$\epsi/2$-fraction of $V$ at next iteration.
Suppose that the algorithm does not terminate
before the outer for loop ends.
In the case that $|A|=k$ or $|V|=0$, we continue the outer for loop
with $s=0$ and select the empty set.
Let $i' = \min \{s,t\}$.
To fail the algorithm, there should be no more than 
$m=\lceil\log_{1-\epsi/2}(1/n) \rceil$ iterations 
that $\lamb^* \ge i'$.
Otherwise, the algorithm terminates either with
$|A|=0$ or with more than $m$
iterations which successfully filter out at least
$\epsi/2$-fraction of $V$ resulting in that $|V|=0$.
The following lemma shows that, at each iteration,
there is a constant probability to successfully discard
at least $\epsi/2$-fraction of $V$ or have
that $|A|=k$.
Then, we will show that
the algorithm succeeds with a probability of at least $1-\delta/n$.

\begin{restatable}{lemma}{thresholdProb} 
	\label{lemma:threshold-prob}
	It holds that $\prob{\lamb^*<i'} \le 1/2$, where $i' = \min\{s,t\}$.
\end{restatable}
\begin{proof}
	Call an element $v_i \in V$ \textit{bad} iff
	$\marge{v_i}{A\cup T_{i-1}} < \tau$.
	Let 
	\color{\rev}
	\begin{equation*}
	\lamb_{i'} = \left\{
	\begin{aligned}
	&i', && \text{if } i' \le \left\lceil \frac{1}{\epsi}\right\rceil\\
	&\max\{\lamb \in \Lamb:\lamb < i'\}, && \text{o.w.}
	\end{aligned}\right. .
	\end{equation*}
	\color{black}
	The random permutation of $V$ can be regarded
	as $|V|$ dependent Bernoulli trials, with success
	iff the element is bad and failure otherwise.
	Observe that the probability that an element in $T_{\lambda_{i'}}$
	is bad is less than $\epsi / 2$, condition on the
	outcomes of the preceding trials.
	Further, if $\lambda_{i'} \not \in B$,
	there are at least $\epsi \lamb_{i'}$ bad elements in 
	$T_{\lamb_{i'}}$, for otherwise,
	the condition on Line \ref{line:threIf2} would hold
	and $\lambda_{i'}$ would be in $B$.
	Let $(Y_i)$ be a sequence of independent and identically distributed Bernoulli trails,
	each with success probability $\epsi/2$.
	The probability can be bounded as follows:
	\begin{align*}
		\prob{\lamb^* < i'}& \le \prob{\marge{T_{\lamb_{i'}}}{A}/\lamb_{i'} < (1-\epsi)\tau}\\
		& \le \prob{\#\text{ of bad elements in }T_{\lamb_{i'}} > \epsi \lamb_{i'}}\\
		& = \prob{\sum_{i=1}^{\lamb_{i'}} 1_{\{v_{i} \text{ is bad}\}} > \epsi \lamb_{i'}}\\
		& \le \prob{\sum_{i=1}^{\lamb_{i'}} Y_i > \epsi \lamb_{i'}}\numberthis \label{ineq:thre-prob1}\\
		&\le1/2,\numberthis \label{ineq:thre-prob2}
	\end{align*}
	where Inequality \ref{ineq:thre-prob1} follows from Lemma \ref{lemma:indep2},
	 Inequality \ref{ineq:thre-prob2} follows from
	Markov's inequality and Law of Total Probability.
\end{proof}

From the above discussion, to fail the algorithm, 
there should be no more than $m$ iterations that $\lamb^* \ge i'$.
Define a \textit{successful iteration} as an iteration that 
$\lamb^* \ge i'$.
Let $X$ be the number of successes during $\ell$ iterations.
From the definition and Lemma \ref{lemma:threshold-prob}, 
$X$ is the sum of $\ell$ dependent Bernoulli random variables,
where each variable has a success probability of more than $1/2$.
Then, let $Y$ be the sum of $\ell$ independent Bernoulli random variables
with success probability $1/2$.
Therefore, the failure probability of Algorithm \ref{alg:threshold} can be bounded as follows:
\begin{align*}
	\prob{\text{Algorithm \ref{alg:threshold} fails}}&\le \prob{X \le m} \\
	&\le \prob{Y \le m} \numberthis \label{ineq:dep2}\\ 
	&\le \prob{Y \le 2\log(n/\delta)/\epsi}\\ 
	&\le e^{-\left(\frac{\epsi+1}{\epsi+2}\right)^2\cdot\left(1+\frac{2}{\epsi}\right)
	\log\left(\frac{n}{\delta}\right)} \numberthis \label{ineq:chernoff}\\ 
	&= e^{-\frac{(\epsi+1)^2}{\epsi(\epsi+2)}\log(\frac{n}{\delta})} \le \delta/n,
\end{align*}
where Inequalities \ref{ineq:dep2} and \ref{ineq:chernoff} follow from 
Lemmas \ref{lemma:indep} and \ref{lemma:chernoff}, respectively.
% Therefore, at any iteration, the probability of successfully filtering out an 
% $\epsi/2$-fraction of $V$ is more than $1/2$, regardless of the
% outcomes (\ie success or failure) of the preceding iterations.
% Let $X_\ell$ be the number of successful iterations
% up to and including the $\ell$-th iteration.
% $X_\ell$ is a sum of dependent Bernoulli random variables.
% And, let $Y_\ell$ be a sum of $\ell$ independent Bernoulli random variables
% with success probability $1/2$.
% If there are more than $m = \lceil \log_{1-\epsi /2}(1/n) \rceil$ 
% successful iterations, the algorithm \threshold will succeed. 
% The probability that the algorithm \threshold succeeds is bounded as
% follows:
% \begin{align*}
% \prob{\text{Algorithm \ref{alg:threshold} succeeds}}
% &\geq \prob{X_\ell \geq m}\\
% & \ge \prob{Y_\ell \geq \lceil \log_{1-\epsi /2}(1/n) \rceil} \numberthis \label{ineq:ts-prob3}\\
% & \geq \prob{Y_\ell \geq 2\log (n/\delta)/\epsi} \\
% & = 1-\prob{Y_\ell \le 2\log (n/\delta)/\epsi} \\
% &\ge 1-e^{-\frac{1}{2} \cdot \left(\frac{\epsi+1}{\epsi+2}\right)^2 
% \cdot 2\left(1+\frac{2}{\epsi}\right)\log(\frac{n}{\delta})} \numberthis \label{ineq:ts-prob4}\\
% & = 1-(\frac{\delta}{n})^{\frac{(\epsi+1)^2}{\epsi(\epsi+2)}} \\
% & \geq 1-\frac{\delta}{n},
% \end{align*}
% where Inequality \ref{ineq:ts-prob3} and \ref{ineq:ts-prob4} follow from Lemma \ref{lemma:indep} and 
% Lemma \ref{lemma:chernoff}, respectively.

\subsubsection{Adaptivity and Qeury Complexity}
For Algorithm \ref{alg:threshold}, oracle queries incurred by computing
the marginal gain on Line \ref{line:threFilterV} and \ref{line:threIf2}.
The filtering on Line \ref{line:threFilterV} can be done in parallel.
And the inner \textbf{for} loop is also in parallel. 
Therefore, the adaptivity for Algorithm \ref{alg:threshold} is 
$\oh{\log(n/\delta)/\epsi}$.

With the same notations of $V_j$ and $Y_i$ in Section \ref{sec:linearseq-proof},
there are no more than $|V_{j-1}| +1$ and $2\log_{1+\epsi}(|V_j|)+4$
oracle queries on Line \ref{line:threFilterV} and in inner \textbf{for}
loop, respectively; by Lemma \ref{lemma:indep} in Appendix \ref{apx:prob},
$\ex{Y_i} \le 2$. 
Then, the expected number of queries is bounded as follows:
\begin{align*}
&\ex{Queries}
\le \sum_{j=1}^\ell \ex{|V_{j-1}| + 2\log_{1+\epsi}(|V_j|) + 5} \\
& \le n+\sum_{n(1-\epsi/2)^i\ge1}\ex{Y_i}\left(n(1-\epsi/2)^i+2\log(n(1-\epsi/2)^i)\right)+ 5\ell\\
& \le n+2n \sum_{i\ge 1}(1-\epsi/2)^i + 
4 \sum_{n(1-\epsi/2)^i\ge1}\log(n(1-\epsi/2)^i)+ 5\ell\\
& \le n\left(\frac{4}{\epsi}+1\right)+4\log\left(n\left(1-\frac{\epsi}{2}\right)\right)\frac{\log_{1-\epsi/2}\left(\frac{1}{n}\right)}{2}+5\lceil 4(1+2/\epsi)\log (n/\delta) \rceil.
\end{align*}
Thus, the total queries are $\oh{n/\epsi}$ in expectation.

\subsubsection{The Marginal Gain (Properties 3 and 4 of Theorem \ref{theorem:threshold})}
Let $A_j$ be $A$ after iteration $j$, and $T_{j, \lamb_j^*}$ 
be the subset added to $A$ at iteration $j$. 
\color{\rev}
If $\lamb_j^* \le \left\lceil \frac{1}{\epsi}\right\rceil$,
the if condition on Line~\ref{line:threIf2} holds for $T_{j,\lamb_j^*}$.
So, we have
\begin{equation*}
\marge{T_{j,\lamb_j^*}}{A_{j-1}}\ge (1-\epsi)\tau |T_{j, \lamb_j^*}|.
\end{equation*}
If $\lamb_j^* > \left\lceil \frac{1}{\epsi}\right\rceil$,
the if condition on Line~\ref{line:threIf2} holds for $T_{j,\lamb_j}$,
where $\lamb_j =\max B_j$ and $B_j$ is $B$ after inner for loop ends at iteration $j$ of the outer for loop.
Similarly, we have
\begin{equation*}
\marge{T_{j,\lamb_j}}{A_{j-1}}\ge (1-\epsi)\tau |T_{j, \lamb_j}|.
\end{equation*}
Let $\lamb_j = \left\lfloor (1+\epsi)^u \right \rfloor$.
Then, $\lamb_j^* = \left\lfloor (1+\epsi)^{u+1} \right \rfloor$.
\begin{equation*}
\frac{|T_{j, \lamb_j}|}{|T_{j, \lamb_j^*}|} = \frac{\lamb_j}{\lamb_j^*}
= \frac{\left\lfloor (1+\epsi)^u \right \rfloor}{\left\lfloor (1+\epsi)^{u+1} \right \rfloor}\ge \frac{(1+\epsi)^u -1}{ (1+\epsi)^{u+1}}\ge \frac{1}{1+\epsi} -\epsi,
\end{equation*}
where the last inequality follows from $(1+\epsi)^{u+1} \ge \lamb_j^*>\left\lceil \frac{1}{\epsi}\right\rceil \ge \frac{1}{\epsi}$.
Therefore, by above two inequalities and monotonicity of $f$,
\begin{equation*}
\marge{T_{j,\lamb_j^*}}{A_{j-1}}\ge\marge{T_{j,\lamb_j}}{A_{j-1}}
\ge (1-\epsi)\left(\frac{1}{1+\epsi}-\epsi\right)\tau |T_{j, \lamb_j^*}|
\ge (1-2\epsi)\tau |T_{j, \lamb_j^*}|/(1+\epsi).
\end{equation*}
% For each $T_{j,\lamb_j^*}$ being added to $A$, there exists a 
% $T_{j, \lamb_j}$, $|T_{j, \lamb_j}| \geq |T_{j, \lamb_j^*}|/(1+\epsi)$ and
% the if condition in Line \ref{line:threIf2} holds for $T_{j,\lamb_j}$. 
By summing up the marginal gain obtained by every iteration,
the value of the solution can be bounded as follows,
\begin{align*}
f(A) = \sum_{j=1}^\ell \marge{T_{j,\lamb_j^*}}{A_{j-1}} 
	\ge \sum_{j=1}^\ell (1-2\epsi)\tau |T_{j, \lamb_j^*}|/(1+\epsi) 
	= (1-\epsi)\tau|A| /(1+\epsi),
\end{align*}
Therefore, the average marginal $f(A)/|A| \geq (1-2\epsi)\tau / (1+\epsi)$.
\color{black}

If $|A| < k$ and the algorithm is successful,
the algorithm terminates with $|V|=0$. 
For any $x \in \mathcal{N}$, 
$x$ should be filtered out at an iteration $j$ with $A_{j-1}$, which
means $\marge{x}{A_{j-1}} < \tau$. Due to submodularity, 
$\marge{x}{A} \leq \marge{x}{A_{j-1}} < \tau$.

%%% Local Variables:
%%% mode: latex
%%% TeX-master: "main.tex"
%%% End:

\section{Pseudocode and Omitted Proofs for Section \ref{section:boost}}
\label{apdix:boost}
In Algorithm \ref{alg:boost}, detailed pseudocode of
\boost is provided. For any set $A$,
the notation $f_A(\cdot)$ represents
the function $f(A \cup \cdot)$; $f_A$ is submodular if $f$
is submodular.

\subsection{Proof of Theorem \ref{theorem:boost}}
\subsubsection{Success Probability} \label{apdx:boost-suc}
% For the \textbf{while} loop in Line \ref{line:boostWhileStart}-\ref{line:boostWhileEnd}, 
% there are no more than 
% $\lceil \log_{1-\epsi}(\alpha/3)\rceil$ iterations. 
% If the $\alpha$-approximation algorithm for \sm succeeds,
% and \threshold succeeds at every iteration,
% Algorithm \ref{alg:boost} also succeeds.
\begin{proof}
	The probability of succeess can be bounded as follows:
\begin{align*}
\prob{\text{Algorithm \ref{alg:boost} succeed}} 
&\geq 1-\prob{\text{\threshold fails at an iteratin}}\\
&\quad -\prob{\alpha \text{-approximation algorithm for \sm fails}}\\
&\geq 1- \sum_{i=1}^{\lceil \log_{1-\epsi}(\alpha/3)\rceil} 
\prob{\text{\threshold fails}} - p_\alpha\\
&\geq 1- \frac{\delta}{n}\cdot \lceil \log_{1-\epsi}(\alpha/3)\rceil- p_\alpha\\
& \ge 1-1/n- p_\alpha.
\end{align*}
\end{proof}

% For the remainder of the proof of Theorem \ref{theorem:boost},
% we assume that
% every call to \threshold succeeds.
% \subsubsection{Adaptive Complexity}
% There are at most $\lceil \log_{1-\epsi}(\alpha/3)\rceil$ iterations of the \textbf{while} loop. 
% Since $\log(x) \leq x-1$, 
% $\lceil \log_{1-\epsi}(\alpha/3)\rceil = \lceil \frac{\log(\alpha/3)}{\log(1-\epsi)}\rceil$, and $\epsi < 1 - 1/e$,
% we have 
% $$ \lceil \log_{1-\epsi}(\alpha/3)\rceil \le \lceil \frac{\log(3/\alpha)}{\epsi}\rceil. $$

% And for each iteration, the adaptivity is $\oh{\log(n/ \delta )/\epsi}$. 
% Thus, the adaptivity for Algorithm \ref{alg:boost} is 
% $$\oh{ \frac{ \log{ \alpha^{-1} } }{\epsi^2} \log \left( \frac{ n \log \left( \alpha^{-1} \right) }{\epsi } \right) } $$

% \subsubsection{Query Complexity}
% Queries to the oracle happen only on Line \ref{line:boostQuery},
% the call to \threshold.
% Since \threshold has $\oh{n/\epsi}$ 
% oracle queries in expectation, and there are no more than 
% $\oh{ \log \left( \alpha^{-1} \right) / \epsi }$
% iterations, the total query calls for Algorithm \ref{alg:boost} 
% is $\oh{ n \log \left( \alpha^{-1} \right) / \epsi^2}$ 
% in expectation. 

\subsubsection{Approximation Ratio}\label{apdx:boost-ratio}
% Let $A_j$ be the set $A$ we get after Line \ref{line:boostUpdateA}, 
% and let $S_j$ be the set returned by \threshold in iteration $j$
% of the \textbf{while} loop.
% Let $\ell$ be the number of iterations of the \textbf{while} loop.

% Consider the case that at termination $|A| < k$. Then, for the last 
% iteration, \threshold returns $S_\ell$, where $0\leq|S_\ell| < k-|A_{\ell-1}|$. 
% From Theorem \ref{theorem:threshold}, for any $o \in O$, 
% $\marge{o}{A} < \tau < \Gamma/(3k)$. By submodularity and monotonicity,
% \begin{align*}
% f(O) - f(A) & \leq f(O\cup A) - f(A) \\
% & \leq \sum_{o \in O\backslash A}\marge{o}{A}\\
% & \leq \sum_{o \in O\backslash A} \Gamma/(3k) \\
% & \leq f(O)/3,
% \end{align*}
% from which $f(A) \geq 2f(O)/3 \geq (1-1/e-\epsi)f(O)$.

% Next, consider the case that $|A| = k$. Suppose in iteration $j+1$, 
% \threshold returns a nonempty set $S_{j+1}$. 
% Then, for the previous iteration $j$,
% \threshold returns a set $S_j$ that $0 \leq |S_j|<k-|A_{j-1}|$.
% From Theorem \ref{theorem:threshold}, it holds that
% $f(A_{j+1}) - f(A_j)\geq (1-\epsi/3)\tau/(1+\epsi/3) \cdot |A_{j+1}\backslash A_j|$;
% and for any $o \in O$, $\marge{o}{A_j} < \tau / (1-\epsi)$.
% By submodularity and monotonicity,
\color{\rev}
\begin{proof}[Proof of Inequality \ref{ineq:boost1}]
	\begin{align*}
		f(A_{j+1}) - f(A_j) &\geq  (1-2\epsi/3)\tau/(1+\epsi/3)\cdot |A_{j+1}\backslash A_j| \numberthis \label{ineq:boost1-1}\\
		& \geq \frac{(1-2\epsi/3)(1-\epsi)}{(1+\epsi/3)k} |A_{j+1}\backslash A_j| \cdot \sum_{o \in O} \marge{o}{A_j} \numberthis \label{ineq:boost1-2}\\
		& \geq \frac{(1-2\epsi/3)(1-\epsi)}{(1+\epsi/3)k}|A_{j+1}\backslash A_j| \cdot(f(O\cup A_j) - f(A_j))\\
		& \geq \frac{(1-2\epsi/3)(1-\epsi)}{(1+\epsi/3)k} |A_{j+1}\backslash A_j| \cdot (f(O)-f(A_j)),
		\end{align*}
		where Inequalities \ref{ineq:boost1-1}-\ref{ineq:boost1-2} follow from Theorem \ref{theorem:threshold}
		and $0 \leq |S_j|<k-|A_{j-1}|$.
		
\end{proof}
% The above inequality also holds when $A_{j+1}=A_j$. 
% Since $|A|=k$, $\sum_{j=0}^{\ell -1} |A_{j+1}\backslash A_j|=k$. 
% Therefore,
\begin{proof}[Proof of Inequality \ref{ineq:boost2}]
\begin{align*}
f(O)-f(A) &\leq \prod_{j=0}^{\ell-1} \left(1-\frac{(1-2\epsi/3)(1-\epsi)}{(1+\epsi/3)k}|A_{j+1}\backslash A_j|\right) f(O)\\
&\leq  e^{-\frac{(1-2\epsi/3)(1-\epsi)}{(1+\epsi/3)k}\sum_{j=0}^{\ell-1}|A_{j+1}\backslash A_j|} f(O)\\
&= e^{-\frac{(1-2\epsi/3)(1-\epsi)}{1+\epsi/3}} f(O)\\
&\le (1/e+\epsi)f(O),\numberthis \label{ineq:boost-appro}
\end{align*}
where Inequality \ref{ineq:boost-appro} follows from Lemma \ref{lemma:boostExp}.
\end{proof}
% Therefore, the approximation ratio is $1-1/e-\epsi$.

\begin{lemma} \label{lemma:boostExp}
	$e^{-\frac{(1-2\epsi/3)(1-\epsi)}{1+\epsi/3}} \leq 1/e+\epsi$,
	when $0\le \epsi \le 1$
	\end{lemma}
	\begin{proof}
	Let $h(x)=1/e+x-e^{-\frac{(1-2x/3)(1-x)}{1+x/3}}$.
	The lemma holds if $h(x)$ is monotonically 
	increasing on $[0,1]$, since $h(0)=0$.
	The remainder of the proof shows that the first derivative of 
	$h(x)$ satisfies that $h'(x)>0$ on $[0,1]$.
	
	The first and second derivatives of $h(x)$ is as follows:
	\begin{align*}
	h'(x)&=1+\frac{2x^2+12x-18}{(x+3)^2}e^{-\frac{(3-2x)(1-x)}{3+x}},\\
	h''(x)&= 4\frac{18(x+3)-(x^2+6x-9)^2}{(x+3)^4} e^{-\frac{(3-2x)(1-x)}{3+x}}.
	\end{align*}
	Let $g(x)=18(x+3)-(x^2+6x-9)^2$. And, it holds that
	$g'(x)=18-4(x+3)(x^2+6x-9)>0$ when $0 \le x \le 1$.
	Therefore, $g(x)$ is monotonically increasing on $[0,1]$.
	Since $g(0)=-27<0$ and $g(1)=68>0$, $g(x)$ only has one zero $x_0$.
	And when $0 \le x \le x_0$, $g(x) \le 0$; 
	when $x_0 \le x \le 1$, $g(x) \ge 0$.
	Since $(x+3)^4 > 0$ and $e^{-\frac{(1-2x/3)(1-x)}{1+x/3}}>0$,
	when $0 \le x \le x_0$, it holds that $h''(x) \le 0$; 
	when $x_0 \le x \le 1$, it holds that $h''(x) \ge 0$.
	Thus, $x_0$ is the minimum point of $h'(x)$.
	Next, we try to bound the minimum of $h'(x)$.
	
	Since $g(0.22)<0$ and $g(0.23)>0$, the zero of $g(x)$ follows that
	$0.22\le x_0 \le 0.23$.
	From the analysis above,
	we know that $|g(x)| \le \max\{|g(0)|,|g(1)|\}=68$.
	Since $(x+3)^4 \ge 81$ and $e^{-\frac{(1-2x/3)(1-x)}{1+x/3}}\le 1$,
	it holds that $|h''(x)| \le 272/81$.
	Therefore, $$h'(x)\ge h'(x_0)\ge h'(0.22)-\max|h''(x)| \cdot (x_0-0.22) > 0.17>0.$$
	Thus, $h(x)$ is monotonically increasing on $[0,1]$. 
	It holds that $h(x)\ge h(0)=0$, when $0\le x \le 1$.
	\end{proof}
\color{black}

%%% Local Variables:
%%% mode: latex
%%% TeX-master: "main.tex"
%%% End:

\section{Lower Adaptivity Modification of \linearseq} \label{apx:ls-low-ada}
\begin{algorithm}[t]
  \caption{A lower adaptivity version of \linearseq with $(5 + \oh{ \epsi )}^{-1}$ approximation ratio in $\oh{\log(n/k) / \epsi^3}$ adaptive rounds and expected $\oh{ n / \epsi^3 }$ queries.}
  \label{alg:ls-low-ada}
  \begin{algorithmic}[1]
  \Procedure{\low}{$f, \mathcal N, k, \epsi$}

  \State \textbf{Input:} evaluation oracle $f:2^{\mathcal N} 
  \to \reals$, constraint $k$, error $\epsilon$

  \State $a = \arg \max_{u \in \mathcal{N}} f(\{u\})$

  \State Initialize $A \gets \{a\}$ , $V \gets \mathcal N$, 
  $\ell = \lceil 4(1+1/(\beta \epsi))\log(n/k) \rceil$, 
  $\beta = \epsi/(16\log(8/(1-e^{-\epsi/2})))$

  \For{ $j \gets 1$ to $\ell$ } 

        \State Update $V \gets \{ x \in V : \marge{x}{A} \ge 
        f(A) / k \}$ and filter out the rest \label{line:lowUpdateV}

        \State \textbf{if} $|V| \le k$ \textbf{then break} 

      \State $V = \{v_1, v_2, \ldots, v_{|V|}\} \gets $\textbf{random-permutation}$(V)$ 
      \color{\rev}
      \State $\Lamb \gets\left[\left\lceil \frac{1}{\epsi} \right\rceil\right] \cup$
      \color{black}
      $\{\lfloor(1+\epsi)^u \rfloor: 
      1 \leq \lfloor(1+\epsi)^u \rfloor \leq k, 
      u \in \mathbb{N}\} $

      \Statex $\qquad \qquad \qquad \cup 
      \{\lfloor k+u\epsi k \rfloor: \lfloor k+u\epsi k \rfloor \leq 
      |V|, u \in \mathbb{N}\} \cup \{|V|\}$

      \State $B[\lamb_i] = \textbf{false}$, for $\lamb_i \in \Lambda$

      \For{$\lamb_i \in \Lamb$ in parallel }
            %
            %\If{$\lamb_i = 1$}
            %\State $T_{\lamb_{i-1}} \gets \emptyset$
            %\Else
            %
            \State $T_{\lamb_{i-1}} \gets \{v_1, v_2, \ldots, v_{\lamb_{i-1}}\}$
      % \EndIf
            ; $T_{\lamb_i} \gets \{v_1, v_2, \ldots, v_{\lamb_i}\}$ ;
            $T_{\lamb_i}' \gets T_{\lamb_i} \backslash T_{\lamb_{i-1}}$

            \State \textbf{if} $\marge{T_{\lamb_i}'}{A \cup T_{\lamb_{i-1}}}/|T_{\lamb_i}'| \geq 
          (1-\epsi) f(A \cup T_{\lamb_{i-1}})/k$ \textbf{then} $B[\lamb_i] \gets \textbf{true}$ \label{line:lowIf2}       

      \EndFor 

      % \State $\lamb^* \gets \max \{\lamb_i \in \Lamb: 
      % B[\lamb_i] = \textbf{false} $ and
      % %\Statex \hspace*{4cm} and  
      % $\left( (\lamb_i \leq k \text{ and } B[1] \text{ to } B[\lamb_{i-1}] 
      % \text{ are all } \textbf{true}) \right.$
      %     %\Statex \hspace*{4.3cm} or $
      %     or $ (\lamb_i > k \text{ and } \exists m \geq 1 \text{ s.t. } 
      % |\bigcup_{u=m}^{i-1} T_{\lamb_u}'| \geq k $
      % and 
      % $B[\lamb_m] \text{ to } B[\lamb_{i-1}] \text{ are all } \textbf{true})) \}$
      \color{\rev}
      \State $\lamb^* \gets \max \left\{\lamb_i \in \Lamb: 
      \left(\lamb_i \le \left\lceil \frac{1}{\epsi} \right \rceil \text{ and } B[1]\text{ to }B[\lamb_i]\text{ are all true }\right)\right.$
      \Statex \hspace*{8em} or $\left(\left\lceil \frac{1}{\epsi} \right \rceil<\lamb_i \leq k \text{ and } B[\lamb_i] = \textbf{false} \text{ and } B[1] \text{ to } B[\lamb_{i-1}] 
                  \text{ are all } \textbf{true}\right)$
      \Statex \hspace*{8em}or $\left(\lamb_i > k \text{ and } B[\lamb_i] = \textbf{false} \text{ and } \exists m \geq 1 \text{ s.t. } 
            |\bigcup_{u=m}^{i-1} T_{\lamb_u}'| \geq k \right.$
            and 
            $\left.\left.B[\lamb_m] \text{ to } B[\lamb_{i-1}] \text{ are all } \textbf{true}\right) \right\}$
      \color{black}

      \State $A \gets A \cup T_{\lamb^*}$ 

    \EndFor

    \State \textbf{if} { $|V| > k$ } \textbf{then return} \textit{failure}
    \State $A' \gets$ last $k$ elements added to $A$
    \State \textbf{return} $C \gets \arg \max\{f(A'), f(V)\}$

    \EndProcedure
\end{algorithmic}
\end{algorithm}
In this section, we describe and analyze a variant of \linearseq
with lower adaptivity. Pseudocode is given in Alg. \ref{alg:ls-low-ada}.
\begin{theorem} \label{theorem:ls-low-ada}
Let $(f,k)$ be an instance of \sm.
For any constant $0<\epsi< 1/2$,
the algorithm $\low$
has adaptive complexity $\oh{\log(n/k)/\epsi^3}$ and
outputs $C \subseteq \mathcal N$ with $|C| \le k$ such that the following properties hold:
1) The algorithm succeeds with probability at least $1 - k/n$.
2) There are $\oh{(1/(\epsi k)+1)n /\epsi^3}$ 
oracle queries in expectation.
3) If the algorithm succeeds, \textcolor{\rev}{$\left[ 5+\frac{2(5-4\epsi)}{(1-2\epsi)^2}\cdot\epsi \right]f(C) \ge f(O)$}, where $O$ is an optimal
solution to the instance $(f,k)$.
\end{theorem}

\begin{proof}
The main difference between \low and \linearseq is that
the algorithm terminate when $|V| \le k$ and returns the maximum value
between $f(A')$ and $f(V)$. 
Since the two algorithms have the same procedures of selection and 
filtering and the value of $\beta$ does not change, 
the probability of successful filtering of $(\beta \epsi)$-fraction
of $V$ at any iteration of the outer \textbf{for} loop is the same as 
\linearseq which is at least 1/2.

\textbf{Success Probability.}
The algorithm \low will succeed if there are at least 
$m=\lceil \log_{1-\beta \epsi}(k/n)\rceil$ successful iterations.
By modeling the success of the iterations as a sequence of dependent 
Bernoulli random variables, and denoting that $X_\ell$ is the number
of successes up to and including the $\ell$-th iteration
and $Y_\ell$ is the number of successes in $\ell$ independent 
Bernoulli trails with success probability 1/2,
\begin{align*}
\prob{\text{Algorithm \ref{alg:ls-low-ada} succeeds}}
&\geq \prob{X_\ell \geq m}\\
& \ge \prob{Y_\ell \geq m} \numberthis \label{ineq:low-prob1}\\
&\ge 1- \prob{Y_\ell \leq \log(n/k) /(\beta \epsi)} \\
&\ge 1-e^{-\frac{1}{2}\left(\frac{2\beta \epsi +1}{2(\beta \epsi +1)}\right)^2
\cdot 2\left(1+\frac{1}{\beta \epsi}\right) \log(n/k )} \numberthis \label{ineq:low-prob2}\\
& = 1-(\frac{k}{n})^{\frac{(2\beta \epsi+1)^2}{4\beta \epsi(\beta \epsi+1)}} \\
& \geq 1-\frac{k}{n},
\end{align*}
where Inequalities \ref{ineq:low-prob1} and \ref{ineq:low-prob2} follow from Lemma \ref{lemma:indep} and 
Lemma \ref{lemma:chernoff}, respectively.

\textbf{Adaptivity and Query Complexity.}
Oracle queries are made on Line \ref{line:lowUpdateV} 
and \ref{line:lowIf2}.
There is one adaptive round on Line \ref{line:lowUpdateV} and also
one adaptive round of the inner \textbf{for} loop. 
Therefore, the adaptivity is proportional to the number of iterations
of the outer \textbf{for} loop, $\oh{\ell}=\oh{\log(n/k)/\epsi^3}$.

For the query complexity, let $V_j$ be the value of $V$ after
filtering on Line \ref{line:lowUpdateV}, 
and let $Y_i$ be the number of iterations between the $(i-1)$-th
and $i$-th successful iterations of the outer \textbf{for} loop.
By Lemma \ref{lemma:indep} in Appendix \ref{apx:prob}, $\ex{Y_i} \le 2$.
Observe that $|V_0| = |\mathcal{N}| = n$.
Then, the expected number of queries is bounded as follows:
\begin{align*}
&\ex{Queries} \leq \sum_{j=1}^\ell \ex{|V_{j-1}|+2|V_j|/(\epsi k)+
        2\log_{1+\epsi} (k)-2/\epsi + 7} \\
%&\le \left(\frac{2}{\epsi k}+1\right)\sum_{j=1}^{\ell} \ex{|V_j|}+n
%       + \left(2\log_{1+\epsi}(k) -\frac{2}{\epsi}+7\right)\ell \\
       &\le \left(\frac{2}{\epsi k}+1\right)\sum_{i=1}^{\infty} 
       \ex{Y_i(1-\beta \epsi)^i n}+n
       + \left(2\log_{1+\epsi}(k) -\frac{2}{\epsi}+7\right)\ell \\
       & \leq \left(1+\left(\frac{2}{\beta \epsi}-2\right)\left(\frac{2}{\epsi k}+1\right)\right)n + 4\left(1-\frac{1}{\beta \epsi}\right)
       \left(\frac{2(1+\epsi)}{\epsi}\log( k)-\frac{2}{\epsi}+7\right)\log(n/k).
\end{align*}
The total queries are $\oh{ (1/(\epsi k)+1)n /\epsi^3 } = \oh{ n / \epsi^3 }$
in expectation.

\textbf{Approximation Ratio.}
Lemma \ref{lemma:fastSubTwo} still holds for Algorithm \ref{alg:ls-low-ada}, since the selection and filtering procedures
do not change.
Thus, it also holds that:
\begin{equation}
f(A')\ge \textcolor{\rev}{\frac{(1-2\epsi)^2}{1+\epsi+(1-2\epsi)^2}} f(A).
\label{ineq:relationOfA}
\end{equation}
\begin{lemma}\label{lemma:lowSub}
Suppose \low terminates successfully. Then $f(O\backslash V)\le 2f(A)$,
where $O$ is the optimal solution of size $k$.
\end{lemma}
\begin{proof}
For each $o \in O\backslash V$, let $j(o)+1$ be the iteration where
$o$ is filtered out, and $A_{j(o)}$ be the value of $A$ after iteration $j(o)$.
It holds that $\marge{o}{A_{j(o)}} < f(A_{j(o)})/k$. Then,
\begin{align*}
f(O\backslash V)-f(A) 
& \le f\left((O\backslash V)\cup A\right)-f(A) \numberthis \label{ineq:low-appro1}\\
& \le \sum_{o \in (O\backslash V)\cup A}\marge{o}{A_{j(o)}}\numberthis \label{ineq:low-appro2}\\
& \le \sum_{o \in (O\backslash V)\cup A} f(A_{j(o)})/k\\
& \le f(A)\numberthis \label{ineq:low-appro3},
\end{align*}
where Inequality \ref{ineq:low-appro1} follows from monotonicity, and Inequalities \ref{ineq:low-appro2}
and \ref{ineq:low-appro3} follow from submodularity.
\end{proof}
The approximation ratio can be calculated as follows,
\color{\rev}
\begin{align*}
f(O)& \le f(O \cap V) + f(O\backslash V) \numberthis \label{ineq:low-appro4}\\
& \le f(V) + \frac{2(1+\epsi+(1-2\epsi)^2)}{(1-2\epsi)^2}f(A')\numberthis \label{ineq:low-appro5}\\
& \le \left(1+\frac{2(1+\epsi+(1-2\epsi)^2)}{(1-2\epsi)^2}\right)f(C)\\
&= \left(5+ \frac{2(5-4\epsi)}{(1-2\epsi)^2}\cdot\epsi\right) f(C),
\end{align*}
\color{black}
where Inequality \ref{ineq:low-appro4} follows from submodularity, and Inequality \ref{ineq:low-appro5}
follows from monotonicity, Inequality \ref{ineq:relationOfA} and Lemma \ref{lemma:lowSub}.
\end{proof}

%%% Local Variables:
%%% mode: latex
%%% TeX-master: "main.tex"
%%% End:

\section{Practical Optimizations to \linearseq} \label{apx:practical-opt}
In this section, we describe two practical optimizations to \linearseq
that do not compromise its theoretical guarantees. The implementation
evaluated in Section \ref{sec:exp} uses these optimizations for \linearseq.
Full details of the implementation are provided in the source code of the Supplementary Material.

\subsection{Avoidance of Large Candidate Sets $A$}
Most of the applications described in Appendix \ref{subsec:obj}
have runtime that depends at least linearly on $|S|$, the size of the set
$S$ to be evaluated. If \linearseq is implemented as specified
in Alg. \ref{alg:fastlinear}, the initial value of $A$ may be
arbitrarily low. Therefore, the size of $A$ may grow very large
as many low-value elements satisfy the threshold
of $f(A) / k$ for acceptance into the set $A$.
If $A$ becomes very large, the algorithm may slow
down due to the evaluation of many large sets, potentially much
larger than $k$. 

To avoid this issue, we adopt the following strategy. The universe
$\mathcal N$ is partitioned into two sets: $\mathcal N_1$, consisting
of the $5k$ highest value singletons, and $\mathcal N_2 = \mathcal N \setminus \mathcal N_1$. Then the main \textbf{for} loop of
\linearseq is executed twice; the first time,
$V$ is initialized to $\mathcal N_1$.
The second execution sets $V = \mathcal N_2$ and $A$ has the value
obtained after the first execution of the \textbf{for} loop.
After the second execution of the \textbf{for} loop, the algorithm
concludes as specified in Alg. \ref{alg:fastlinear}.

The idea behind the first execution of the \textbf{for} loop
is to obtain an initial set $A$ with a relatively high value;
then, when the rest of the elements $\mathcal N_2$ are considered,
elements with low value with respect to $f(A)$ will be filtered immediately
and the size of $A$ will be limited. In fact, one can show that
it is sufficient to take $|A| < O(k \log(k))$ using this strategy
to ensure that $f(A) \ge \opt$; we omit this proof, but it follows
from an initial value of $f(A)$ of at least the maximum singleton
value, the fact that $\opt$ is at most $k$ times the maximum
singleton value, and the condition on the growth of the value of $f(A)$
when elements are added.

\subsection{Early Termination}
A simple upper bound on \opt may be obtained by the sum of the top $k$
singleton values. Hence, the algorithm may check if its ratio is satisfied
early, \ie, before the \textbf{outer} for loop finishes. If so, the algorithm
can terminate early and still ensure its approximation ratio holds.  

\section{Experiments } \label{apx:exp}

\subsection{Implementation, Environment, and Parameter Settings} \label{subsec:params} 
 The experiments are conducted on a server running Ubuntu 20.04.2 with kernel verison 5.8.0. To efficiently utilize all the available threads, we install the Open-MPI and mpi4py library in the system and implement all the algorithms using Message Passing Interface (MPI)  Using MPI we make all the implementations CPU bound as it allows us to minimize communication between the processors by providing explicit control over the parallel communication architecture and the information exchanged between the processors. The experiments are run with the $mpirun$ command and for all the applications, each algorithm is repeated five times and the average of all the repetitions is used for the evaluation of the objective, value, number of query calls and parallel runtime. Similar to \citet{Breuer2019}, the parallel runtime of the algorithms are obtained by computing the difference in time between two calls to $MPI.Wtime()$, once after all the processors are provided with a copy of the input data and the objective function and the other call at the completion of the algorithm. The objective values used for the evaluation are normalized by the objective value obtained from \plg \cite{minoux1978accelerated}, an accelerated parallel greedy algorithm that avoids re-evaluating samples that are known to not provide the highest marginal gain.
The parameters $\epsi$ for \flsabr is set to $0.1$ and to obtain the $\alpha$ and $\Gamma$ from the preprocessing algorithm \ls, the $\epsi$ was set at 0.21. For \fast, the $\epsi$ and $\delta$ are set to $0.05$ and $0.025$ respectively, same as the parameter settings in the \citet{Breuer2019} evaluation\footnote{Our code is available at \textit{https://gitlab.com/deytonmoy000/submodular-bestofbothworlds.}}. 

\subsection{Application Objectives and Datsets}\label{subsec:obj}
\subsubsection{Max Cover} \label{exp:maxcov}
The objective of this application for a given graph $G$ and a constraint $k$  is to find a set $S$ of size $k$, such that the number of nodes having atleast one neighbour in the set $S$ is maximized. The application run on synthetic random graphs, each consisting of 100,000 nodes generated using the Erdős–Rényi (ER), Watts-Strogatz (WS) and Barabási–Albert (BA) models. The ER graphs were generated with $p=0.0001$, while the WS graphs
were created with $p=0.1$
and 10 edges per node. For the BA model, graphs were generated by adding
$m=5$ edges each iteration.   

\subsubsection{Image Summarization on CIFAR-10 data} \label{exp: imgsumm}
In this application, given a constraint $k$, the objective is to find a subset of size $k$ from a large collection of images which is representative of the entire collection. The objective of the application used for the experiments is a monotone variant of the image summarization from \citet{fahrbach2019non}. For a groundset with $N$ images, it is defined as follows:

\begin{align*}
f(S) = \sum_{i \in N} \max_{j \in S} s_{i,j}
\end{align*}
where $s_{i,j}$ is the cosine similarity of the $32 \times 32$
pixel values between image $i$ and image $j$. The data set used for the image summarization experiments is the CIFAR-10 test set containing 10,000
$32 \times 32$ color images.

\subsubsection{Twitter Feed Summarization} \label{exp:twtsumm}
The objective of this application for a given constraint $k$ is to select $k$ tweets from an entire twitter feed consisting of large number of tweets, that would represent a brief overview of the entire twitter feed and provide all the important information. The objective and the data set used for this experiments is adopted from  twitter stream summarization of \citet{Kazemi2019}. For a given twitter feed of $N$ tweets, where each tweet $s \in N$ consists of a set of keywords $W_s$ and $W$ is the set of all keywords in the feed such that $W$ = $\bigcup_{s = 1}^{N} W_s$, the objective function for the twitter feed summarization is given by:
\begin{align*}
f(S) = \sum_{w \in W} \sqrt{\sum_{s \in S} \mathrm{score}(w, s)}  
\end{align*}
where score($w,s$) is the number of retweets of the tweet $s$ such that $w \in W_s$, otherwise score($w,s$) = 0 if $w \notin W_s$. The experiments for twitter feed summarization uses a data set consisting of 42,104 unique tweets from 30 popular news accounts.

\subsubsection{Influence Maximization on a Social Network.} 
In influence maximization, the objective is to select a set of social network influencers to promote a topic in order to maximise its aggregate influence. The probability that a random user $i$ will be influenced by the set of influencers in $S$ is given by:
\begin{align*}
    f_i(S) & = 1 \quad \textrm{for $i \in S$} \\
    f_i(S) & = 1 - (1 - p)^{|N_S(i)|} \quad \textrm{for  $i \notin S$ }
\end{align*}

where |$N_S(i)$| is the number of neighbors of node $i$ in $S$. We use the  Epinions data set with 27,000 users from \citet{rossi2015network} for the influence maximization experiments.

\subsubsection{Revenue Maximization on YouTube.} 
Based on the objective function and data set from \citet{mirzasoleiman2016fast}, the objective for this application objective is to maximise product revenue by choosing set of Youtubers $S$, who will each advertise a different product to their network neighbors. For a given set of users $X$ and $w_{i,j}$ as the influence between user $i$ and $j$, the objective function can be defined by:
  \begin{align*}
      f(S) & = \sum_{i \in X} V \left(\sum_{j \in S} w_{i,j}\right) \\
      V(y) & = y^{\alpha}
  \end{align*}
  where $V(S)$, the expected revenue from an user is a function of the sum of influences from neighbors who are in $S$ and $\alpha : 0 < \alpha < 1 $ is a rate of diminishing returns parameter for increased cover. We use the Youtube data set consisting of 18,000 users and value of $\alpha$ set to 0.9 for this application.

\subsubsection{Traffic Speeding Sensor Placement.} 
The objective of this application is to install speeding sensors on a select set of locations in a highway network to maximize the observed traffic by the sensors. This application uses the data from CalTrans PeMS system \cite{caltrans} consisting of 1,885 locations from Los Angeles region.

% \subsubsection{Movie Recommendation.} 
% Based on the user movie ratings, the objective of this application is to recommend a set of movies from a large collection so that the set is small, diverse and consist of highly rated movies. The objective function is adopted from \citet{Balkanski2018a}, and for select set of movies $S$, the objective is to sum the user ratings of the movies in the set $S$ and use the diversity term to create a diverse set that covers the set of movie genres in the data. For the set of users $U$, the objective function is defined by:
%     \begin{align*}
%         f(S) = \sum_{i \in U}\sum_{j \in S} r_{i,j} + \alpha C(S) + \beta D(S)
%     \end{align*}

% where $r_{i,j}$ is user $i$'s predicted rating of movie $j$; $C(S)$ is the coverage function that counts the number of different genres covered by $S$ and D(S) is the coverage function that counts the number of users with at least one highly rated film in $S$. Parameters $\alpha$ and $\beta$ are corresponding relative weight for coverage functions $C(S)$ and $D(S)$ respectively. MovieLens 1$m$ data set consisting of 6000 users and 4000 movies is used for the experiments, with weight parameters $\alpha = 0.5 \max_j (\sum_i r_j)$ where $r_j > 4.5$ and $\beta$ = 1.
\subsection{Additional Evaluation Results (continued from Section \ref{sec:exp})}

Figure \ref{fig:evalApndxObj}, \ref{fig:evalApndxQry}, \ref{fig:evalApndxAda} and \ref{fig:evalApndxTim} illustrates the performance comparison with \fast across the MaxCover(WS), MaxCover(ER), MaxCover(BA), TweetSumm, InfluenceMax and TrafficMonitor applications. In Fig. \ref{fig:evalApndxObj}, we show the mean objective value of \fast and \flsabr 
across all the applications and datasets. The objective value is normalized by that of \plg. Both algorithms perform very similarly with objective value higher than $80$\% on most instances, however, as demonstrated in Fig. \ref{fig:Apndx_objA}, \ref{fig:Apndx_objB}, \ref{fig:Apndx_objC} the objective value obtained by \fast is not very stable. Overall as shown in the table \ref{table:cmp-exp}, \flsabr either maintains or outperforms the objective obtained by \fast across these applications with the TrafficMonitor and MaxCover (BA) being the instances where it exceeds the average objective value of FAST by 6\% and 5\% respectively.

\begin{figure}[t]
  \subfigure[]{
    \includegraphics[width=0.31\textwidth, height=0.14\textheight]{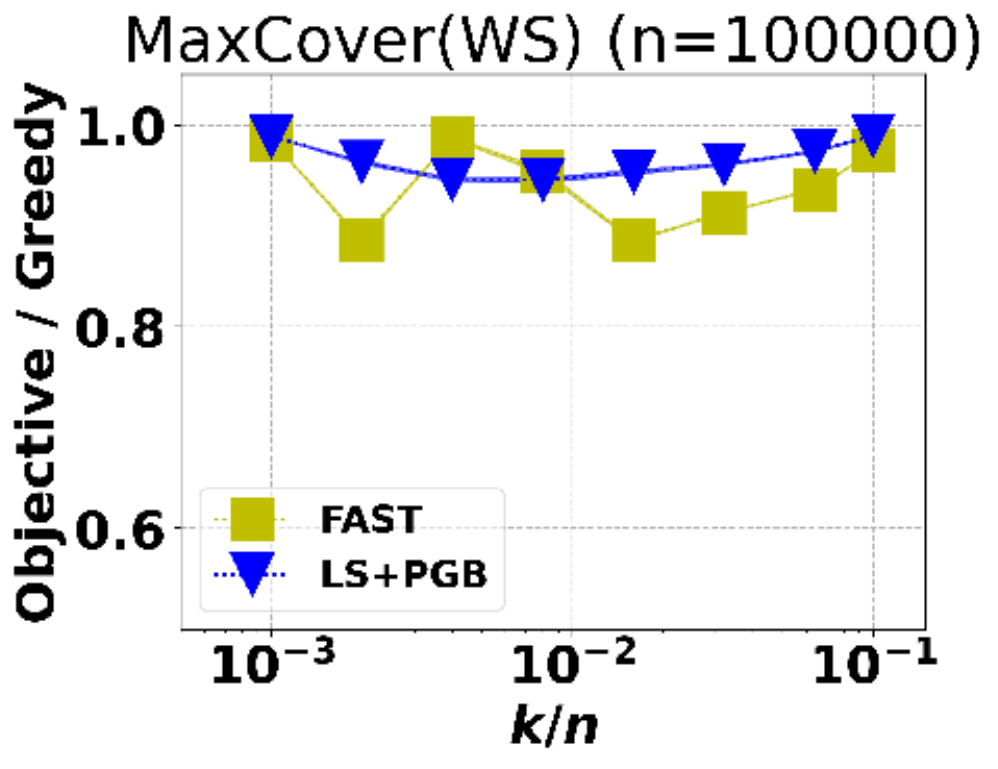}\label{fig:Apndx_objA}
  }
  \subfigure[]{
    \includegraphics[width=0.31\textwidth, height=0.14\textheight]{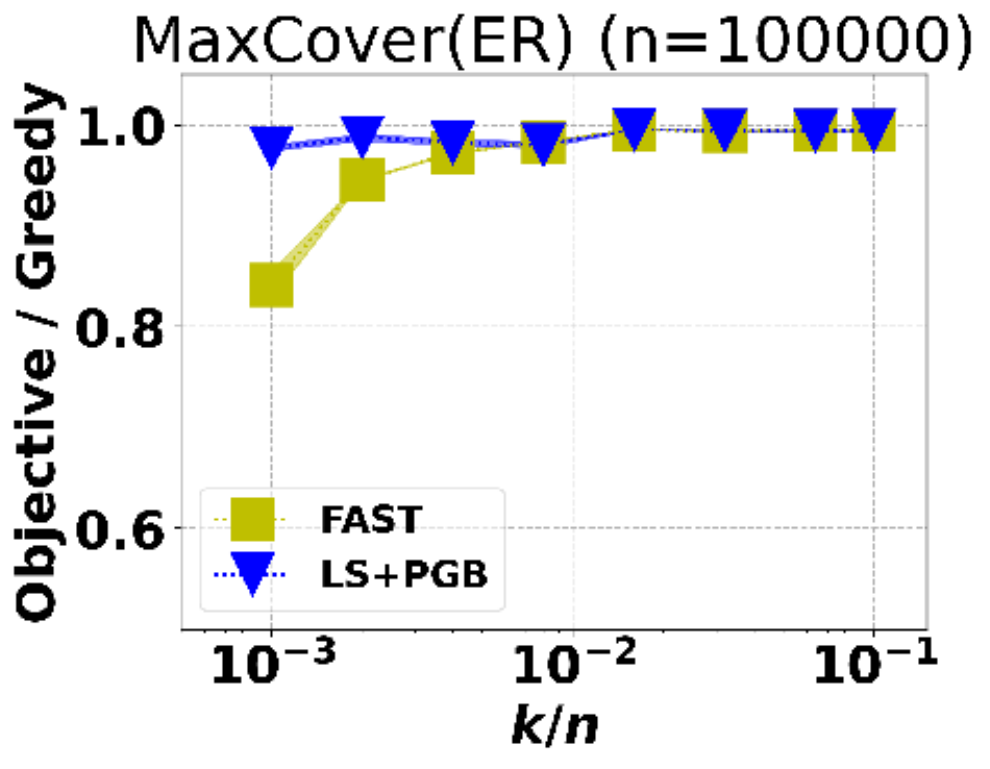} \label{fig:Apndx_objB}
  }
  \subfigure[]{
    \includegraphics[width=0.31\textwidth, height=0.14\textheight]{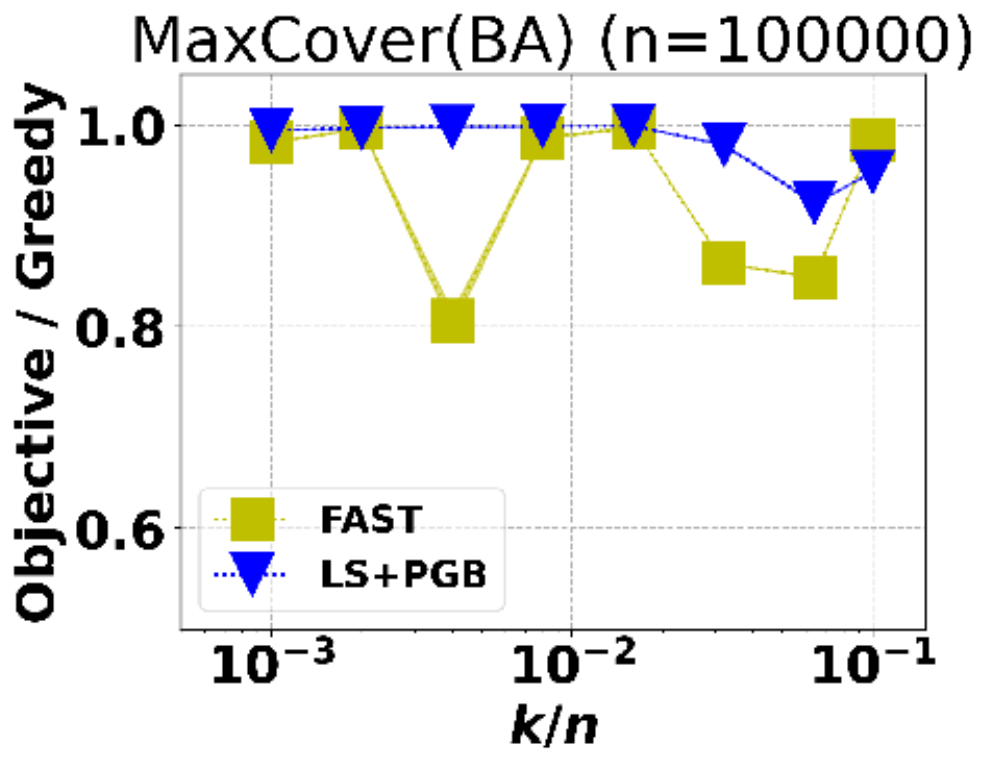} \label{fig:objA}
  }
  \subfigure[]{
    \includegraphics[width=0.31\textwidth, height=0.14\textheight]{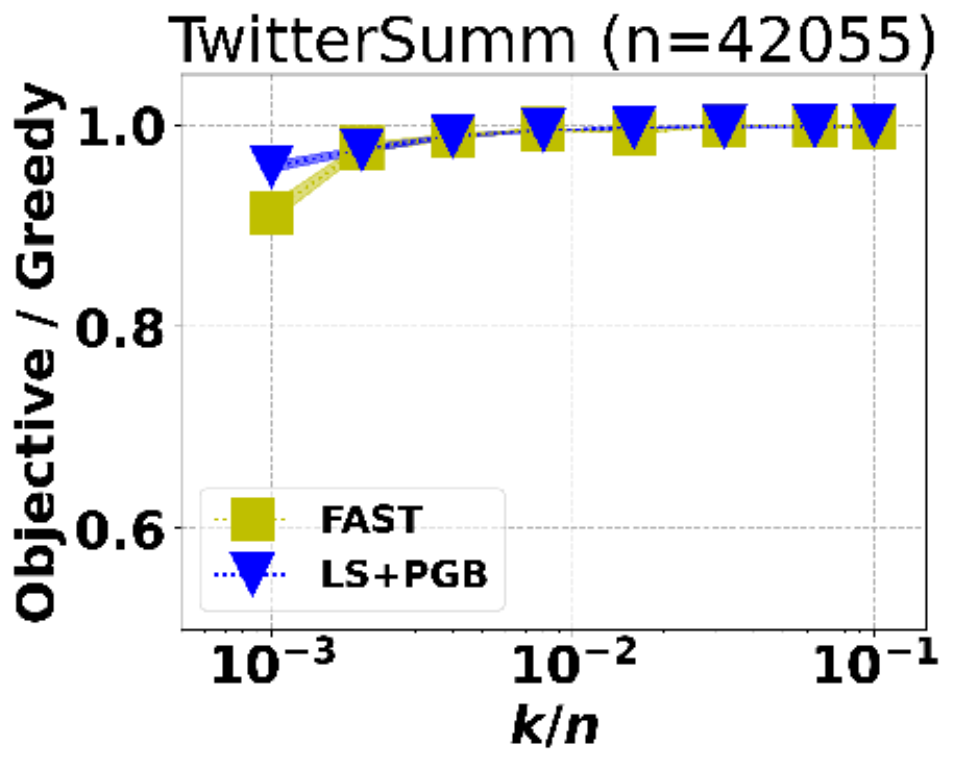}
  }
  \subfigure[]{
    \includegraphics[width=0.31\textwidth, height=0.14\textheight]{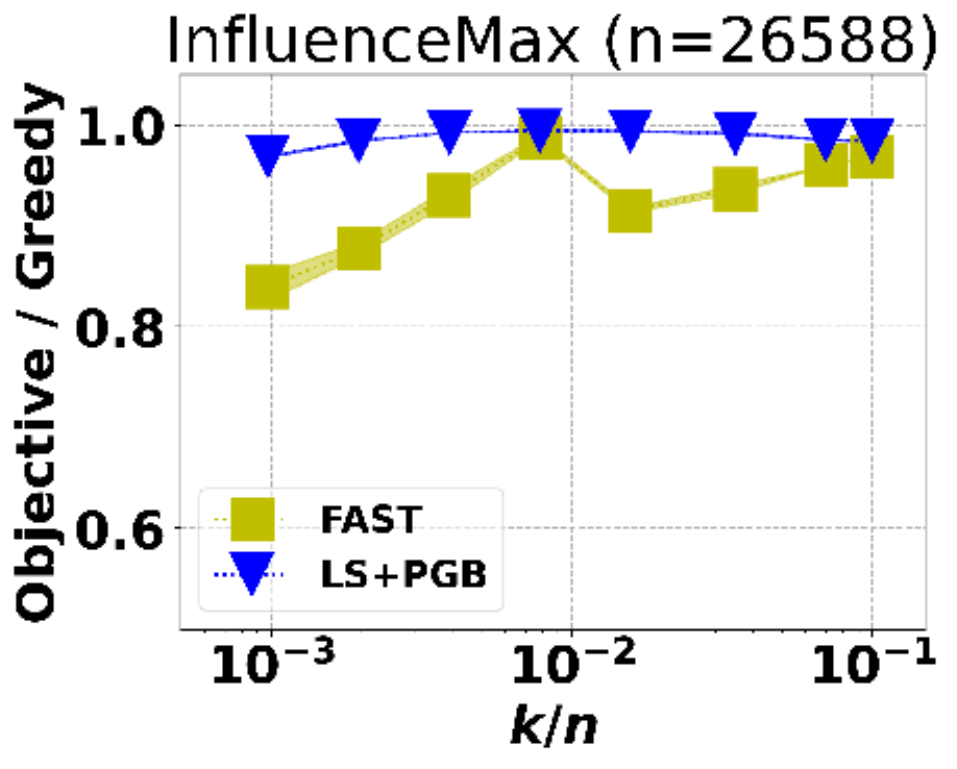} \label{fig:Apndx_objC}
  }
  \subfigure[]{
    \includegraphics[width=0.31\textwidth, height=0.14\textheight]{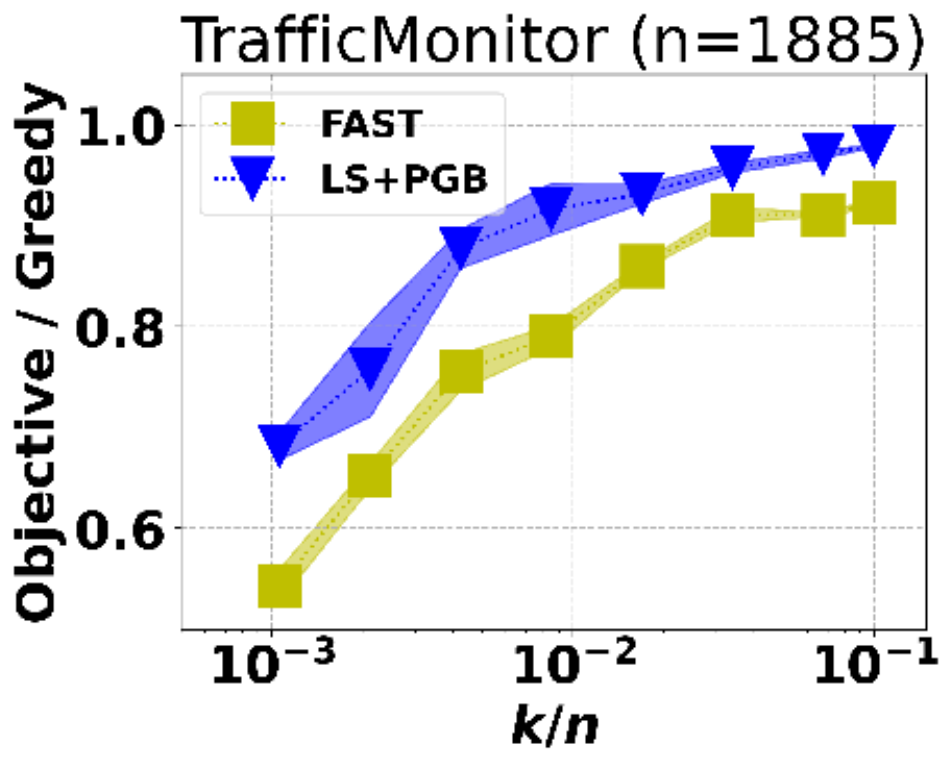}\label{fig:Apndx_objD}
  } 

\caption{Objective value vs. $k/n$. The objective value is normalized by the standard greedy value. The ($k/n$)-axis is log-scaled.} \label{fig:evalApndxObj}
\end{figure}

Fig.\ref{fig:evalApndxQry} demonstrates the mean total queries needed by \flsabr and \fast for all  applications with both \fast and \flsabr exhibiting a linear scaling behavior with the increasing $k$ values with the magnitude of rise in total queries with $k$ is less than 5 folds even with 100 folds increase in $k$.  Overall as shown in table \ref{table:cmp-exp}, \flsabr achieves the objective in less than half the total queries required by \fast for the MaxCover and the TwitterSumm objective. Whereas for TrafficMonitor and InfluenceMax, \fast requires 1.5 and 1.9 times the queries needed by \flsabr for the same objective. Very similar in nature to the number of query calls, as shown in fig. \ref{fig:Apndx_timeA} -  \ref{fig:Apndx_timeD}, \flsabr either maintains or outperforms \fast across all the applications.
\begin{figure}[t]
  \subfigure[]{
    \includegraphics[width=0.31\textwidth, height=0.14\textheight]{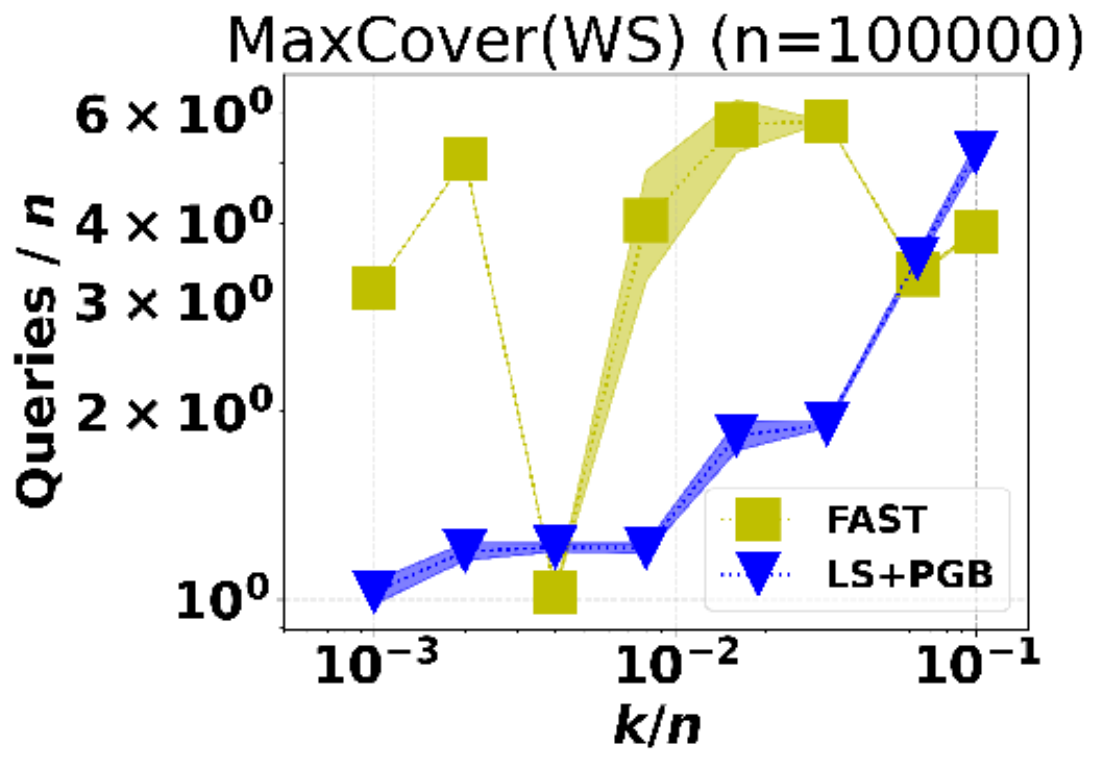} \label{fig:Apndx_qryA}%\label{fig:queryB}
  }
  \subfigure[]{
    \includegraphics[width=0.31\textwidth, height=0.14\textheight]{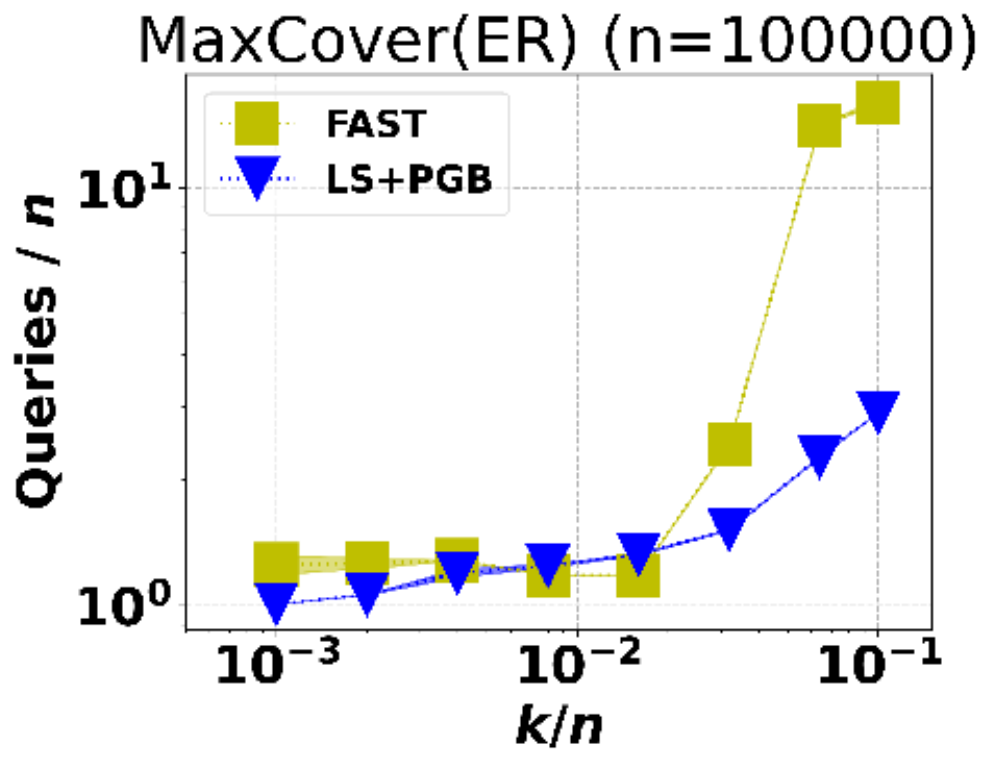} \label{fig:Apndx_qryB}
  }
  \subfigure[]{
    \includegraphics[width=0.31\textwidth, height=0.14\textheight]{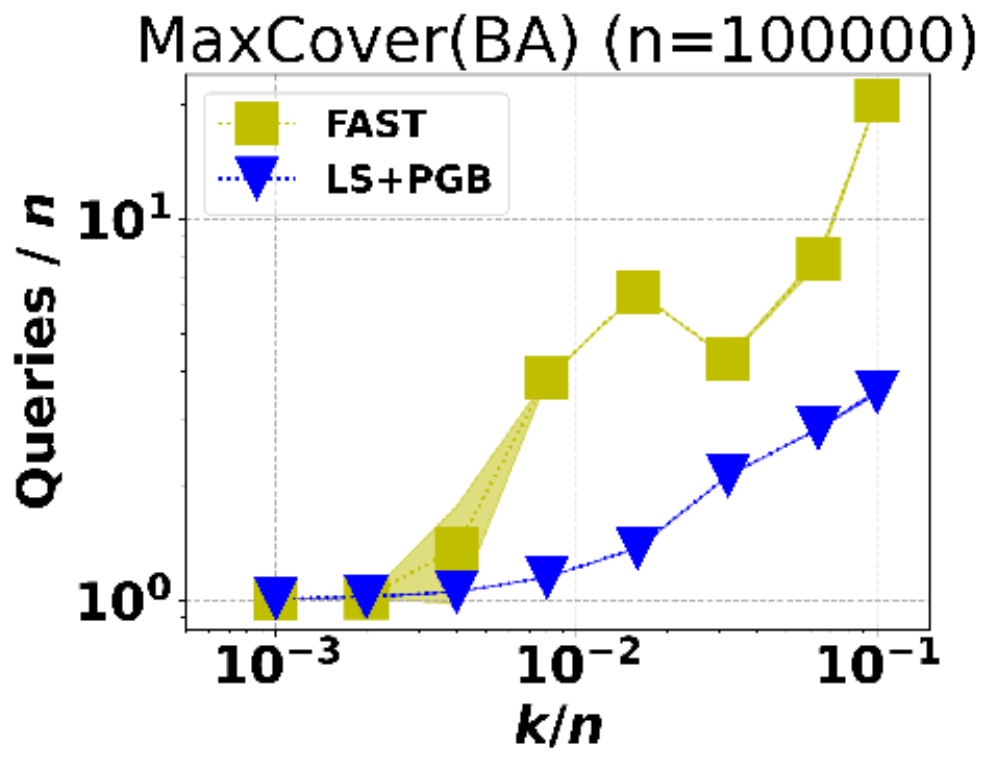}\label{fig:qryA}
  }
  \subfigure[]{
    \includegraphics[width=0.31\textwidth, height=0.14\textheight]{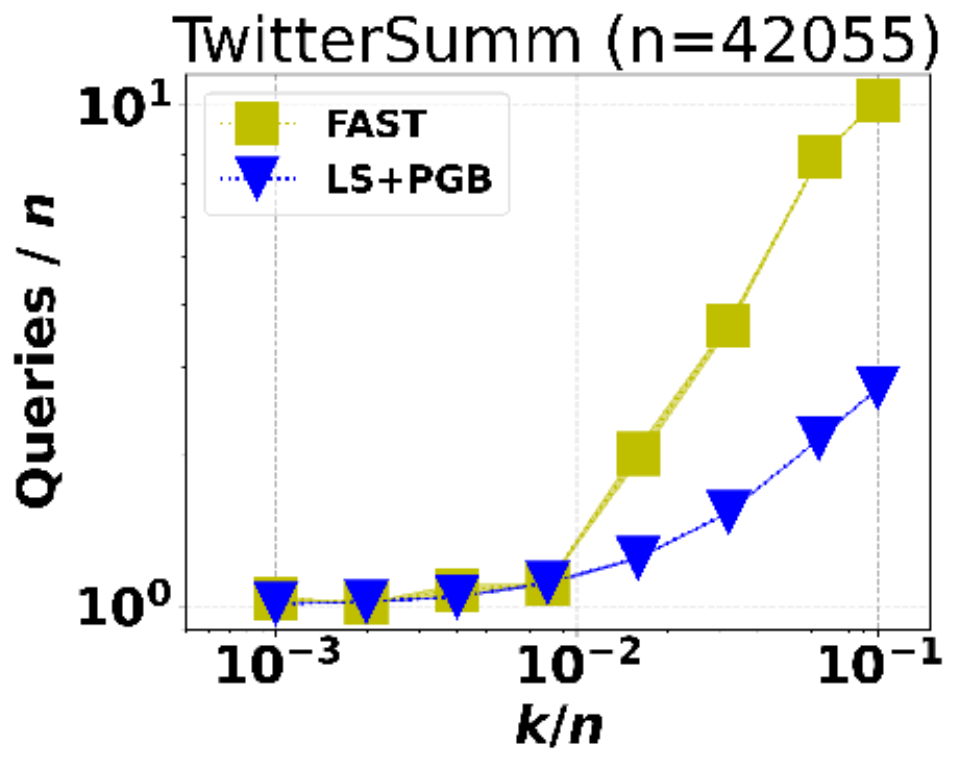}
  }
  \subfigure[]{
    \includegraphics[width=0.31\textwidth, height=0.14\textheight]{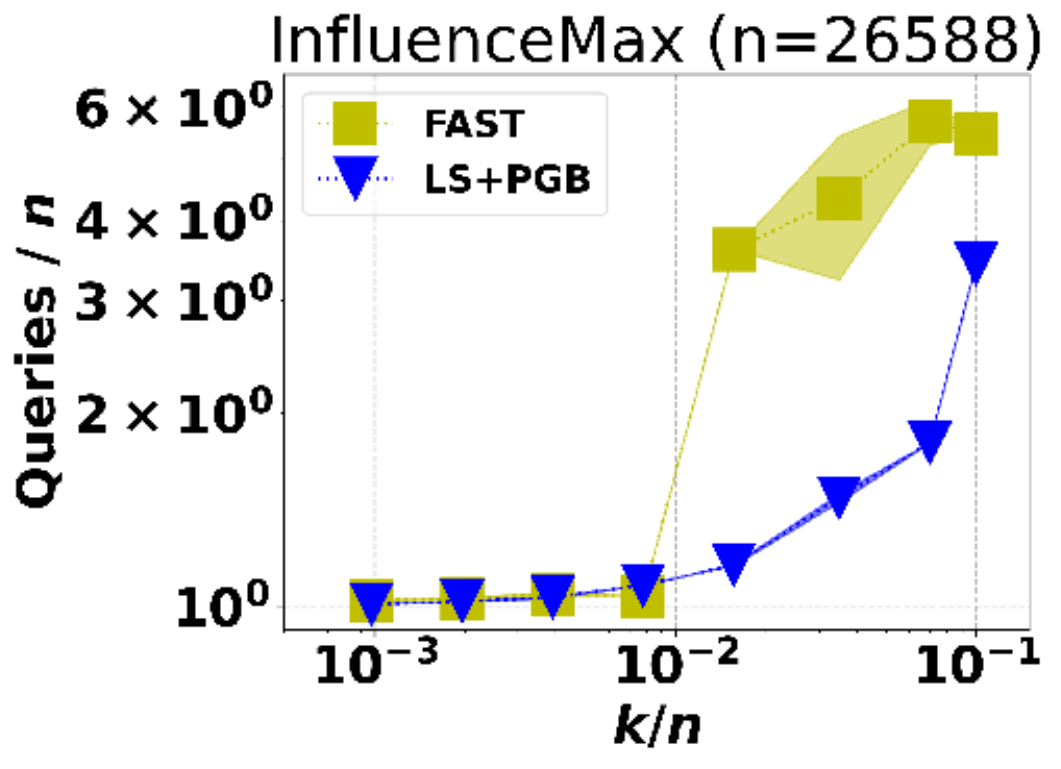}\label{fig:Apndx_qryC}
  }
  \subfigure[]{
    \includegraphics[width=0.31\textwidth, height=0.14\textheight]{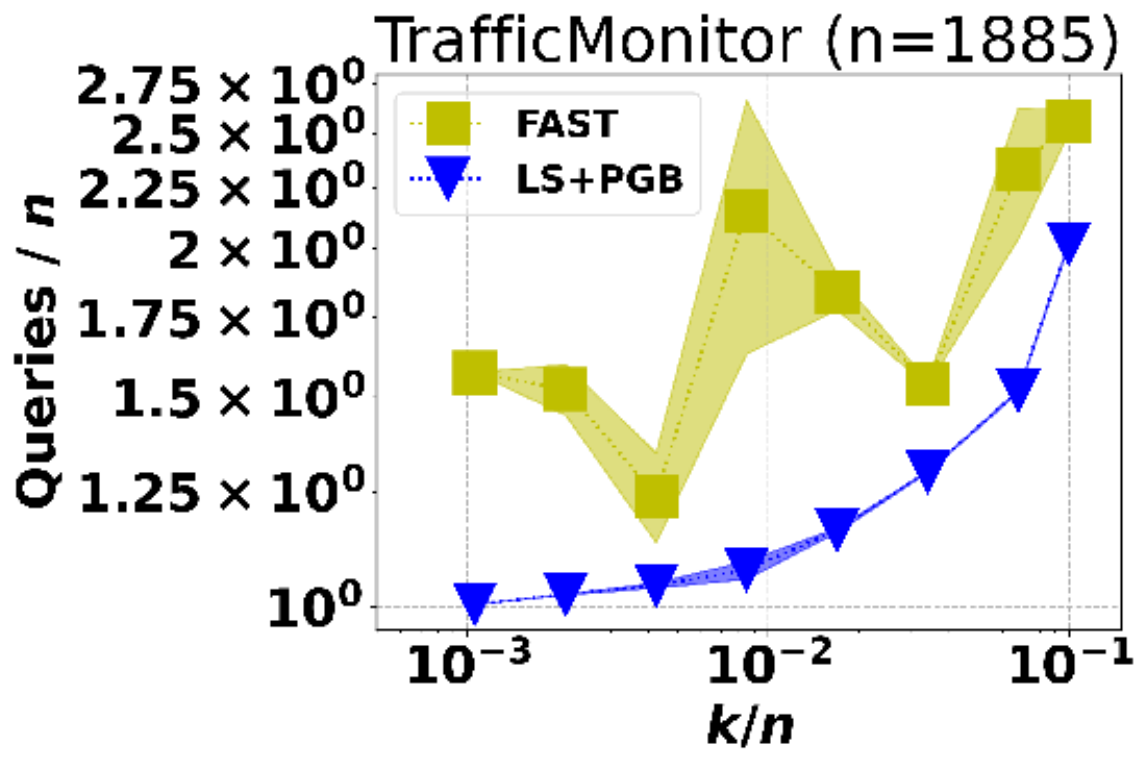} \label{fig:Apndx_qryD}%\label{fig:queryH}
  }
\caption{Total Queries/ $n$ vs. $k/n$. Both axes are log-scaled.} \label{fig:evalApndxQry}
\end{figure}

Fig. \ref{fig:evalApndxAda} and \ref{fig:evalApndxTim} illustrates the adaptivity and the parallel runtime of \flsabr and \fast across the six datasets.  As shown in fig. \ref{fig:evalApndxTim}, both algorithms exhibit linear scaling of runtime with $k$. Overall on an average over the six datasets, FAST requires more than 3 times the time needed by LS+PGB to achieve the objective with TwitterSumm being the objective where overall \flsabr is over 4 times quicker than \fast. In terms of adaptivity, as demonstarted in fig. \ref{fig:evalApndxAda}, especially for larger values $k$, \fast requires more than double the adaptive rounds needed by \flsabr for the MaxCover and TwitterSumm application. For InfluenceMax and TrafficMonitor, \flsabr either maintans or outperforms \fast for larger $k$ values.
 
\begin{figure}[t]
  \subfigure[]{
    \includegraphics[width=0.31\textwidth, height=0.14\textheight]{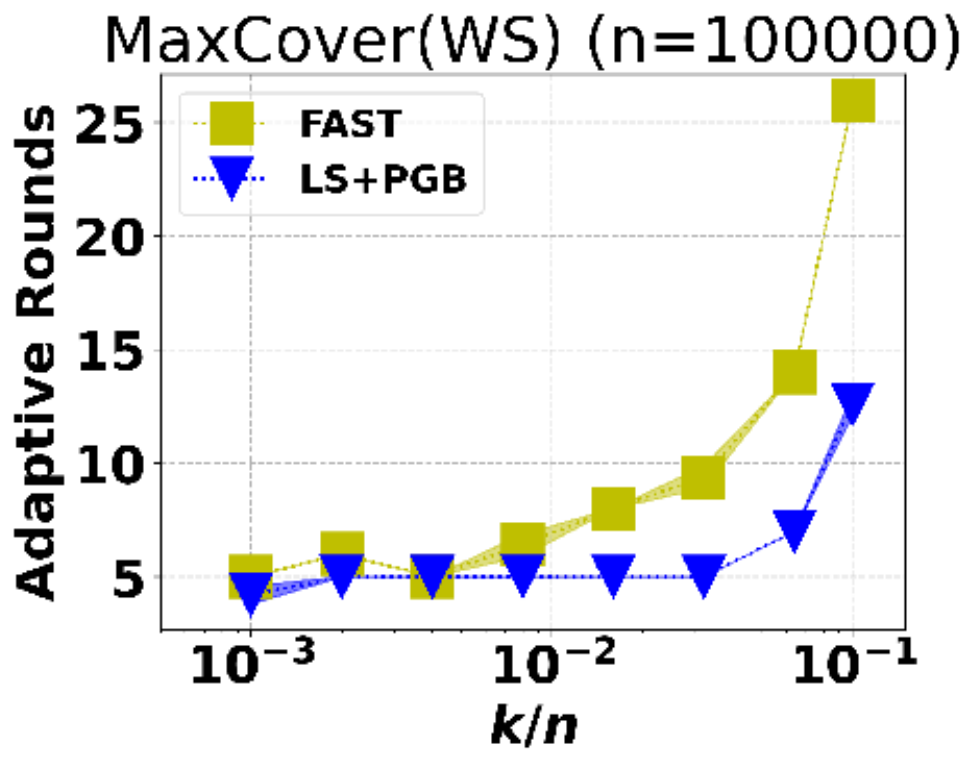} \label{fig:Apndx_adaA}
  }
  \subfigure[]{
    \includegraphics[width=0.31\textwidth, height=0.14\textheight]{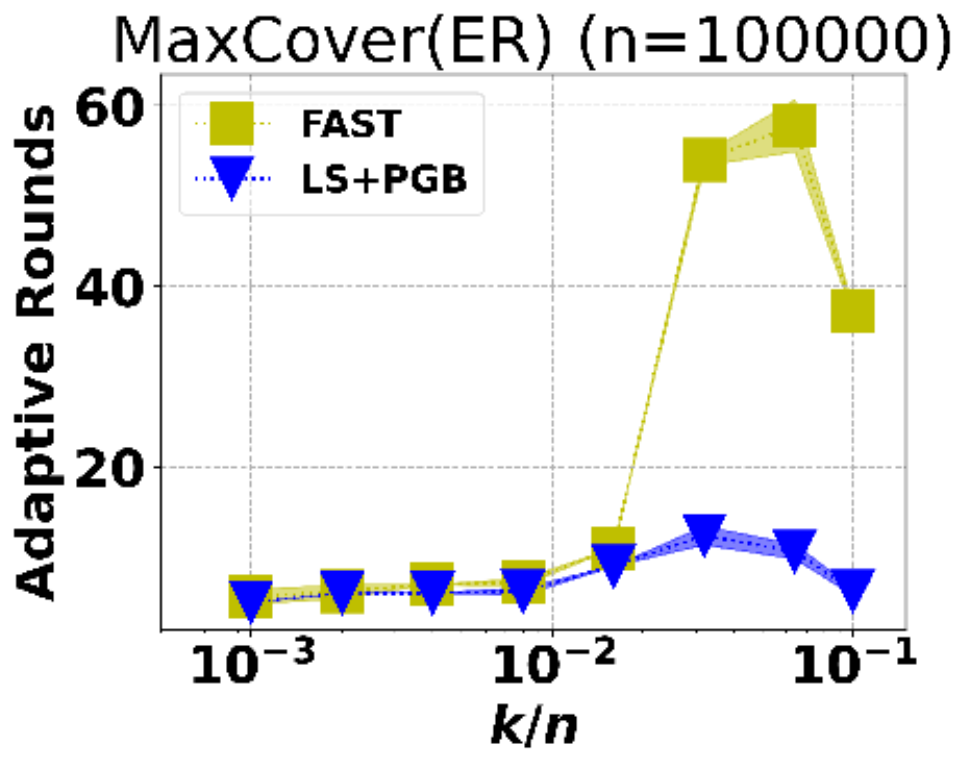}
  }
  \subfigure[]{
    \includegraphics[width=0.31\textwidth, height=0.14\textheight]{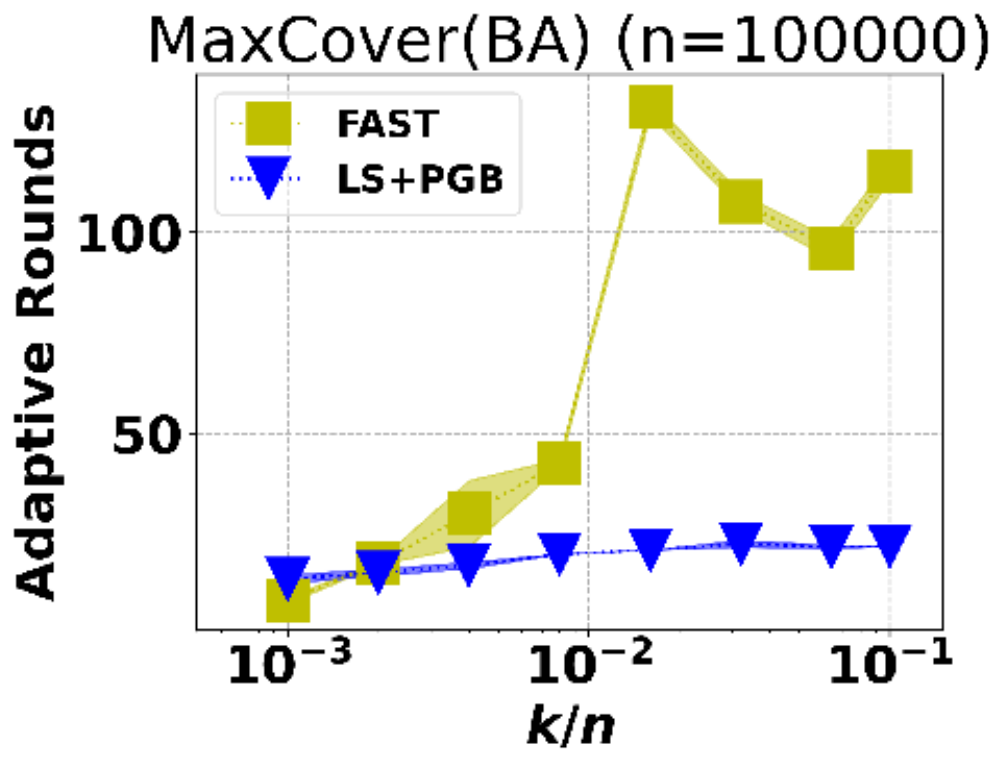}\label{fig:adaA}
  }
  \subfigure[]{
    \includegraphics[width=0.31\textwidth, height=0.14\textheight]{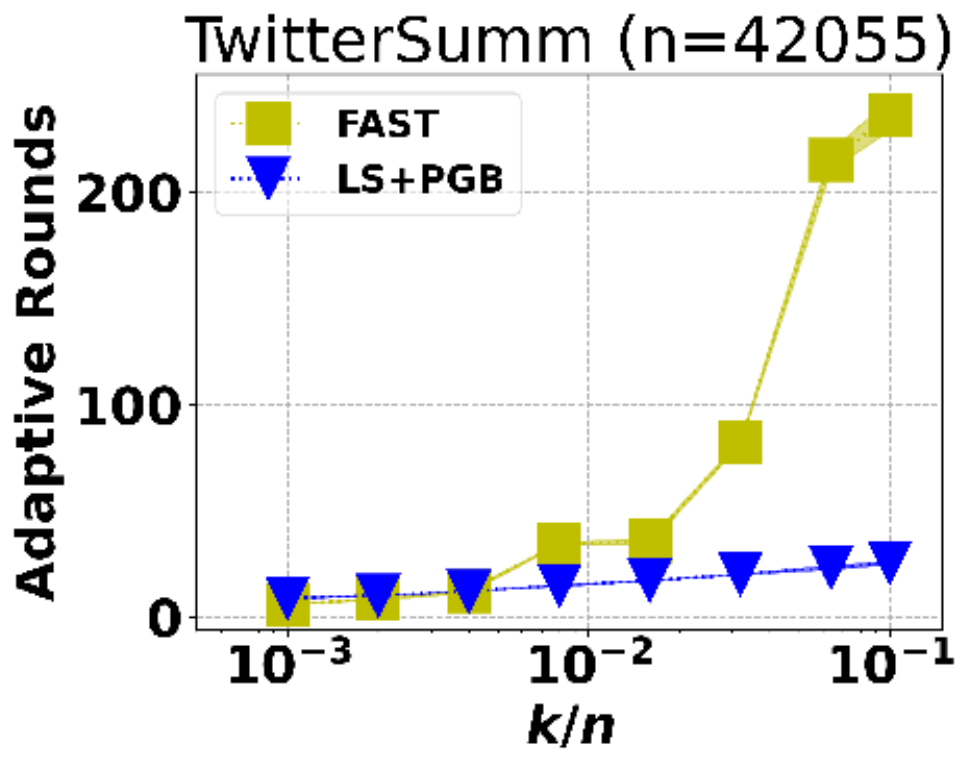}
  }
  \subfigure[]{
    \includegraphics[width=0.31\textwidth, height=0.14\textheight]{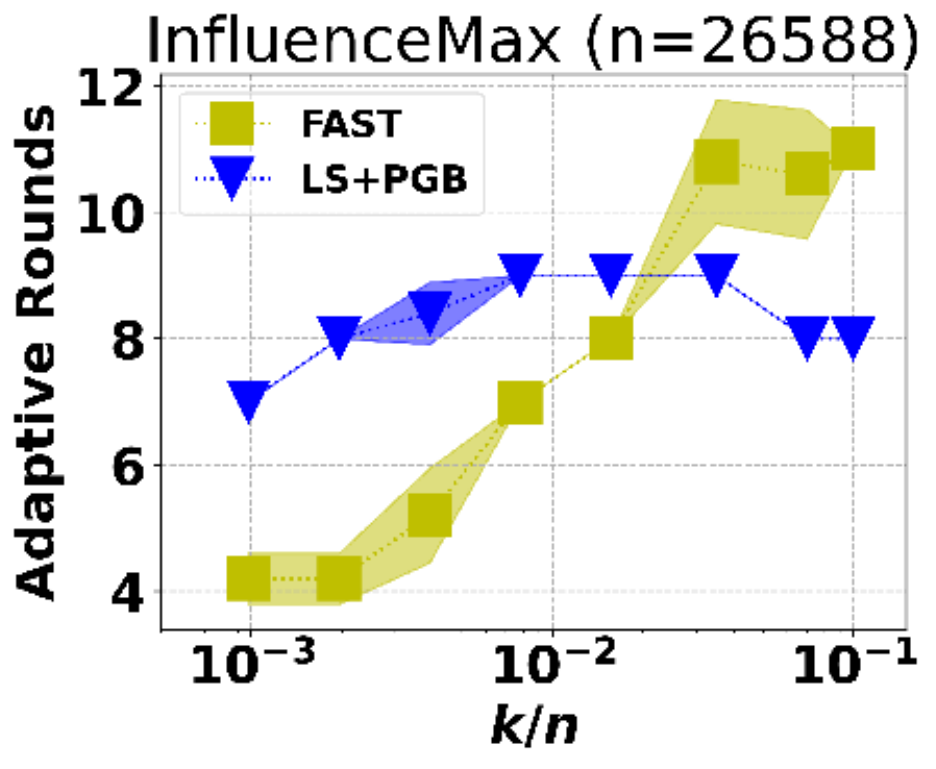}
  }
  \subfigure[]{
    \includegraphics[width=0.31\textwidth, height=0.14\textheight]{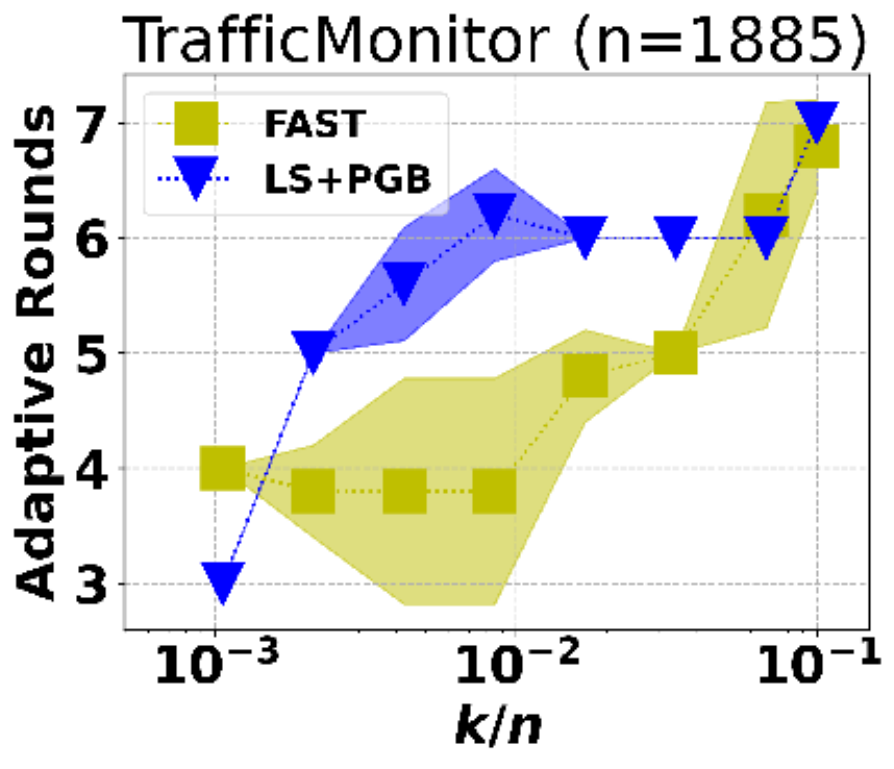} \label{fig:Apndx_adaD}
  }
\caption{Adaptive Rounds vs. $k/n$. The ($k/n$)-axis is log-scaled.} \label{fig:evalApndxAda}
\end{figure}
\begin{figure}[t]
  \subfigure[]{
    \includegraphics[width=0.31\textwidth, height=0.14\textheight]{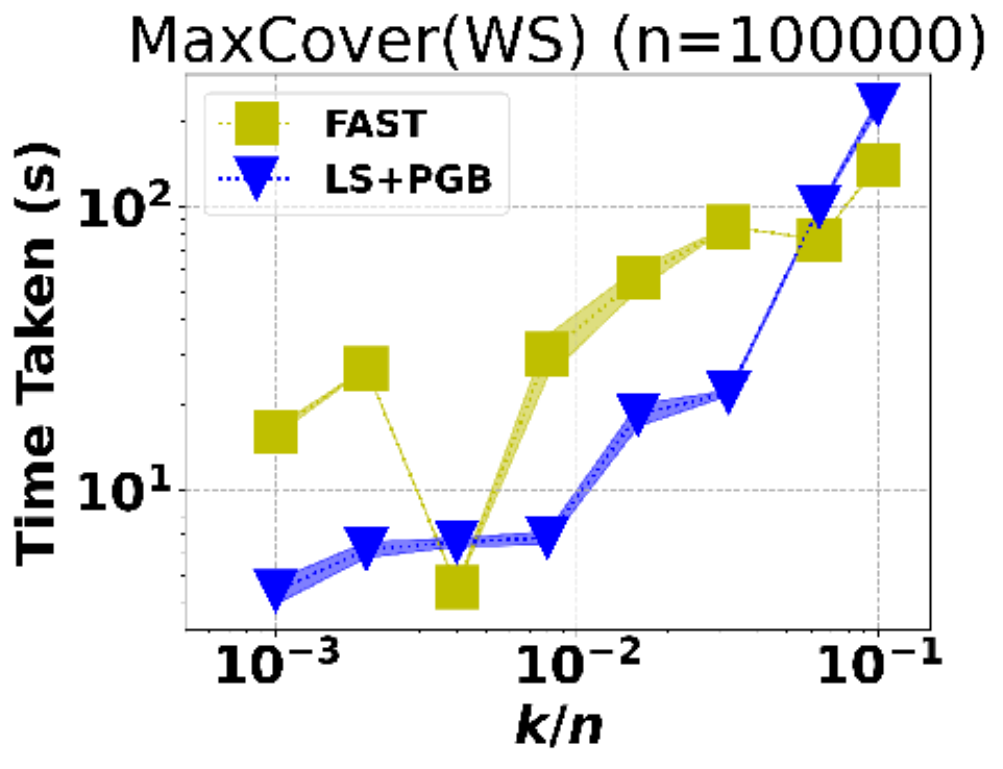}\label{fig:Apndx_timeA}
  }
  \subfigure[]{
    \includegraphics[width=0.31\textwidth, height=0.14\textheight]{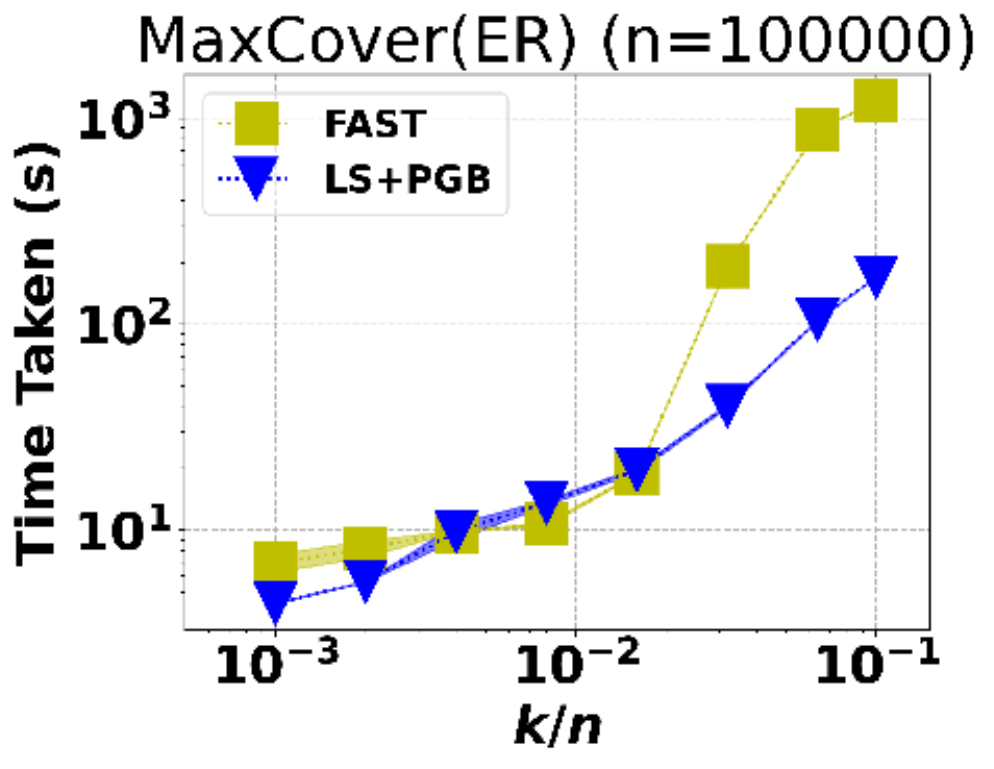}
  }
  \subfigure[]{
    \includegraphics[width=0.31\textwidth, height=0.14\textheight]{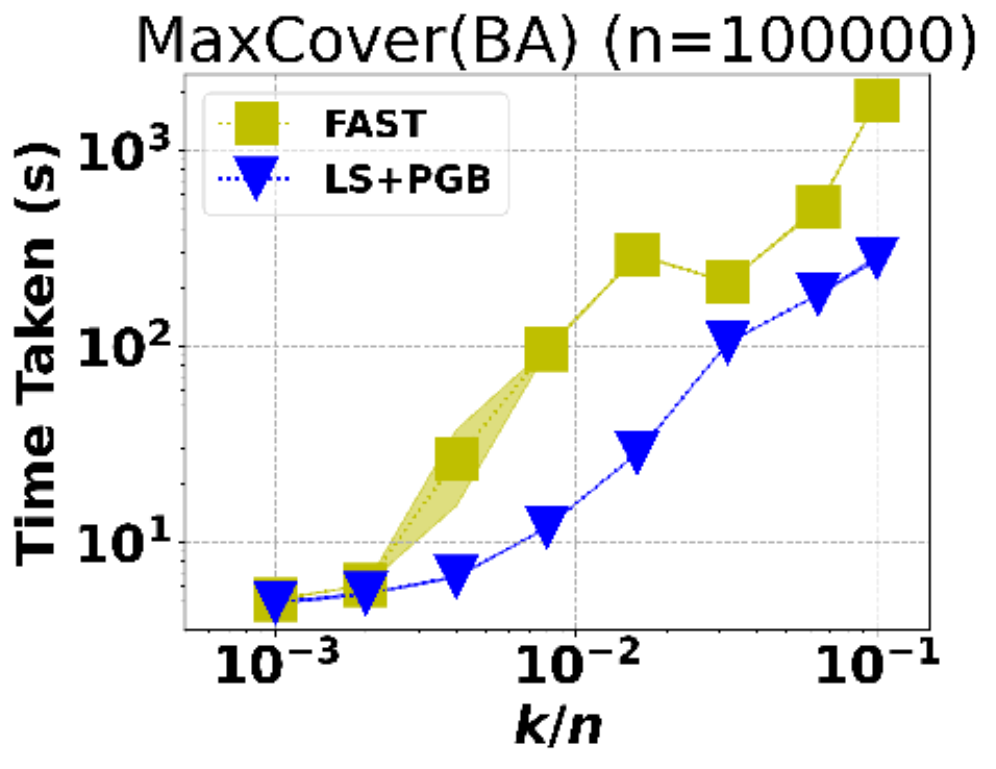}\label{fig:timeAapx}
  }
  \subfigure[]{
    \includegraphics[width=0.31\textwidth, height=0.14\textheight]{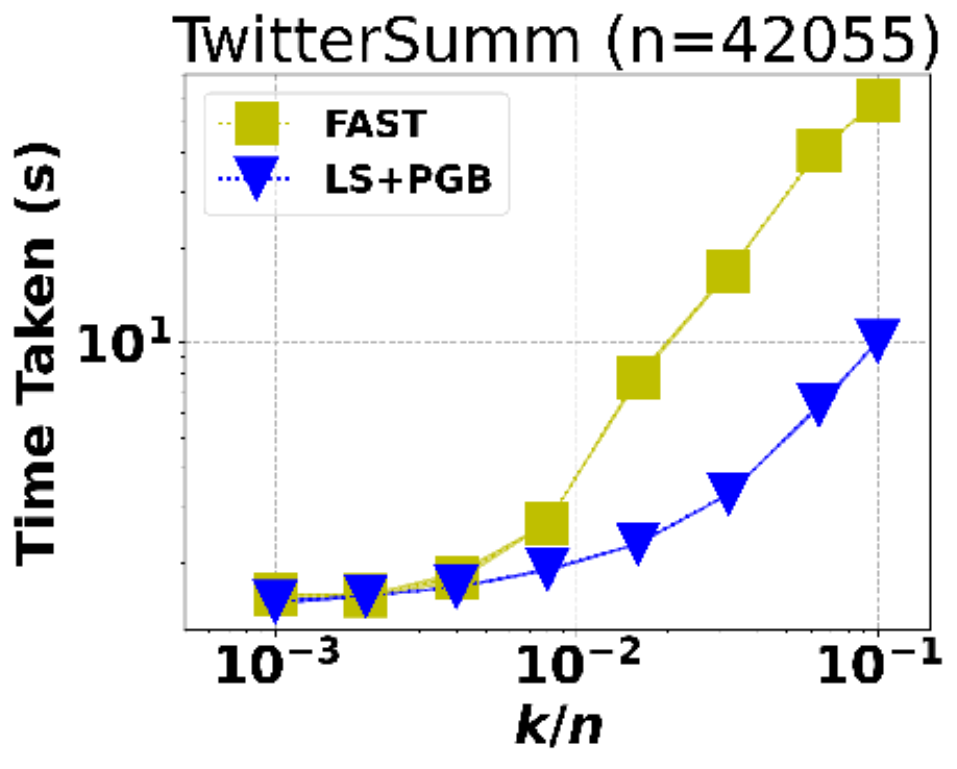}
  }
  \subfigure[]{
    \includegraphics[width=0.31\textwidth, height=0.14\textheight]{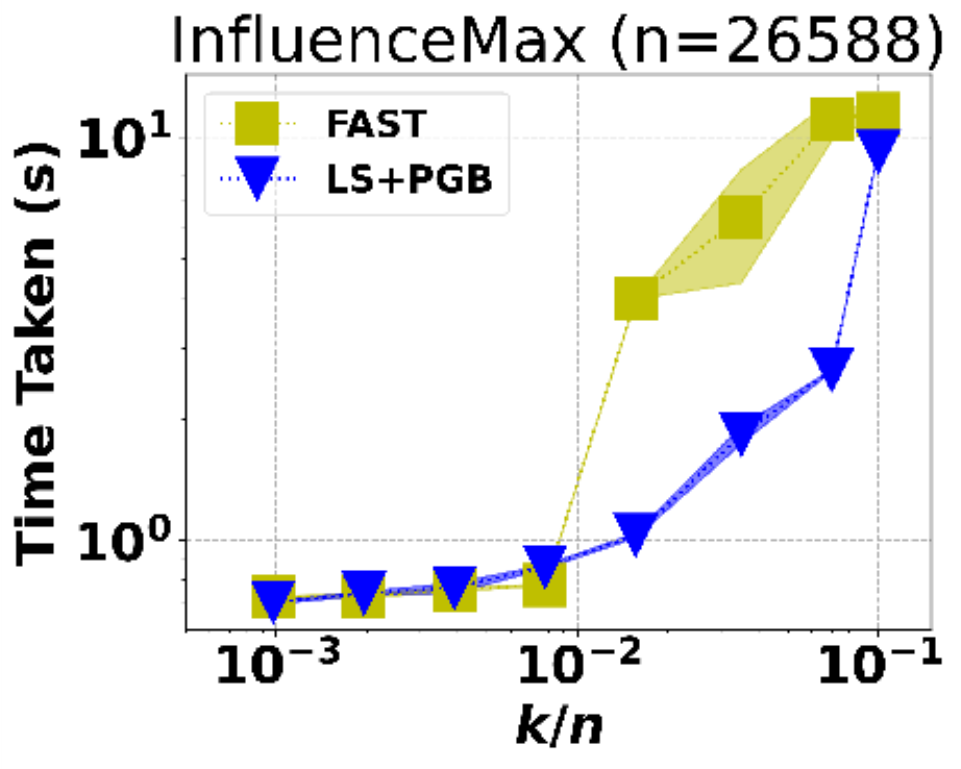}
  }
  \subfigure[]{
    \includegraphics[width=0.31\textwidth, height=0.14\textheight]{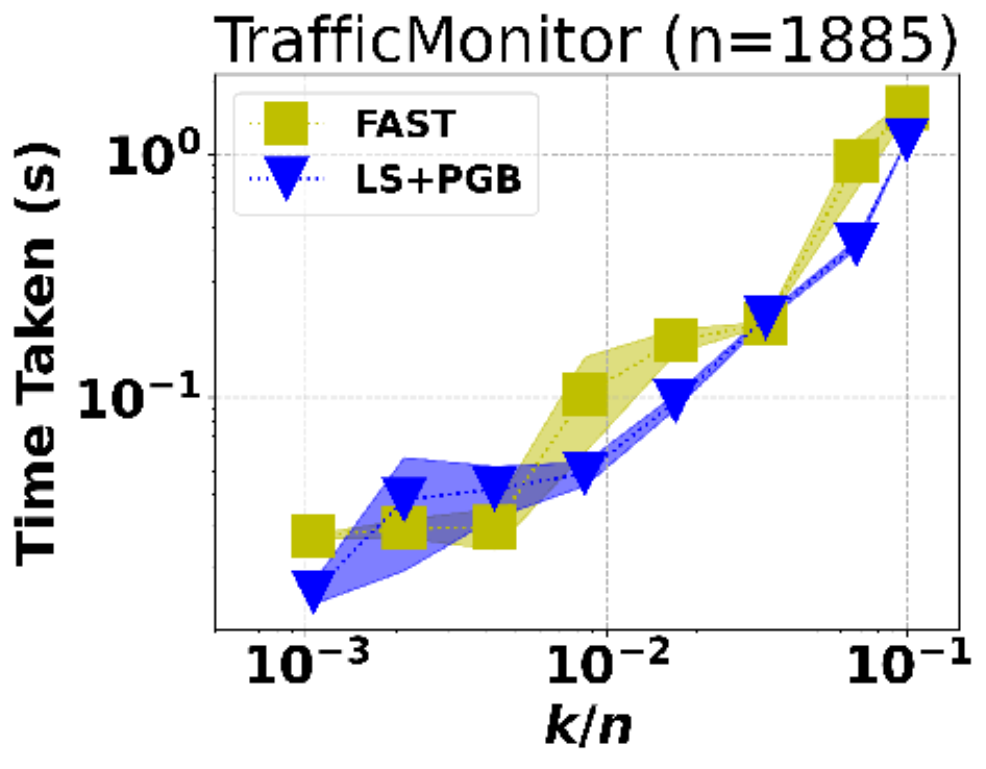}\label{fig:Apndx_timeD}
  }
  \caption{Parallel Runtime vs. $k/n$. Both axes are log-scaled.}\label{fig:evalApndxTim}
\end{figure}

% \subsection{Empirical Queries}

\subsection{Time vs. Number of Threads}

\begin{figure}[t]
  \subfigure[]{
    \includegraphics[width=0.23\textwidth, height=0.11\textheight]{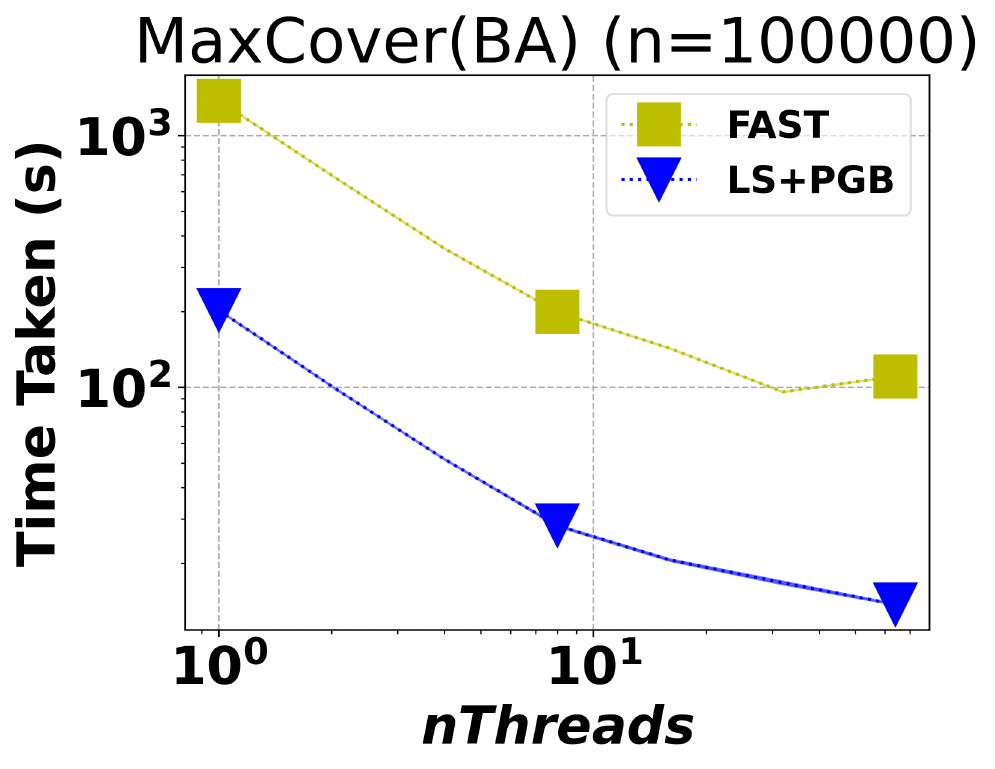}
  }
  \subfigure[]{
    \includegraphics[width=0.23\textwidth, height=0.11\textheight]{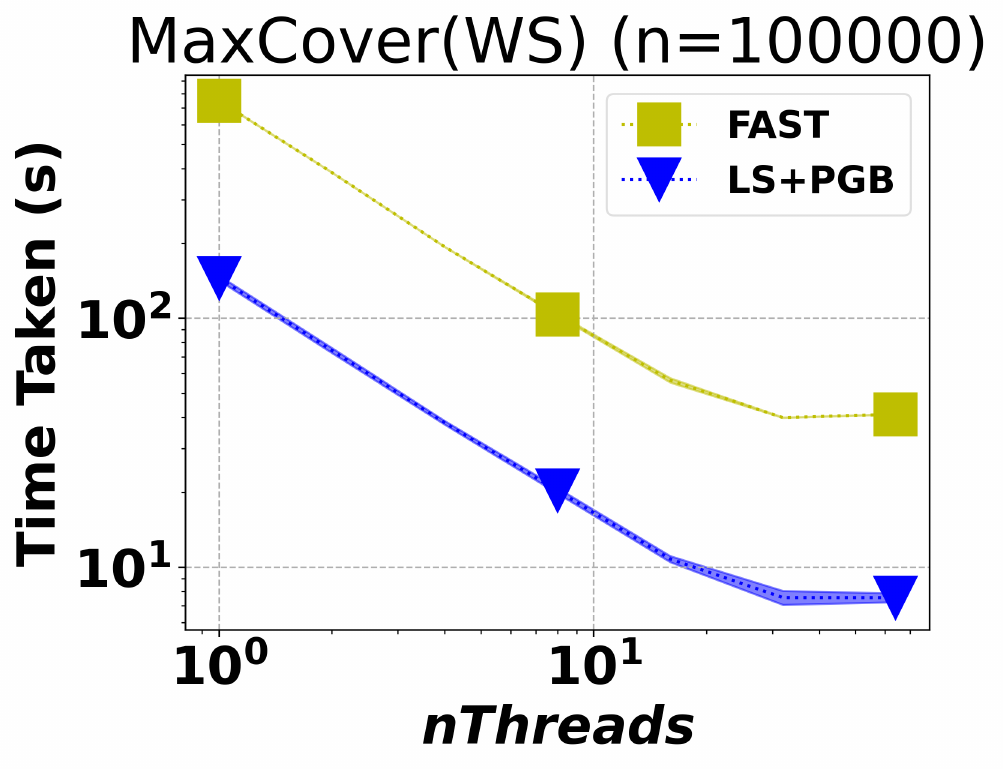}
  }
  \subfigure[]{
    \includegraphics[width=0.23\textwidth, height=0.11\textheight]{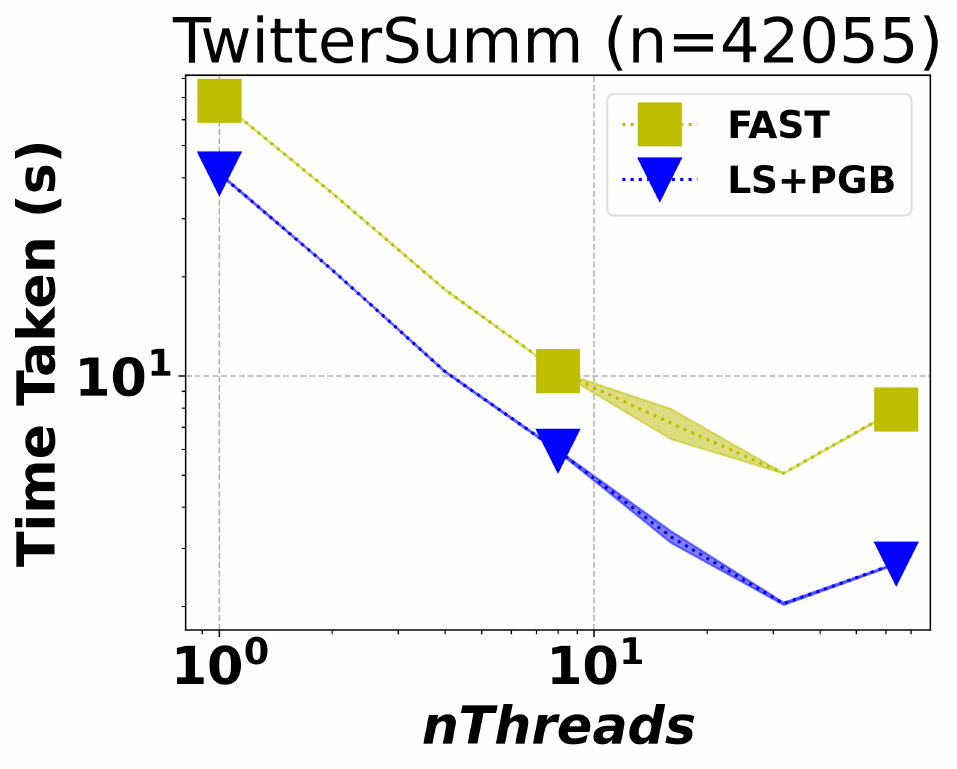}
  }
  \subfigure[]{
    \includegraphics[width=0.23\textwidth, height=0.11\textheight]{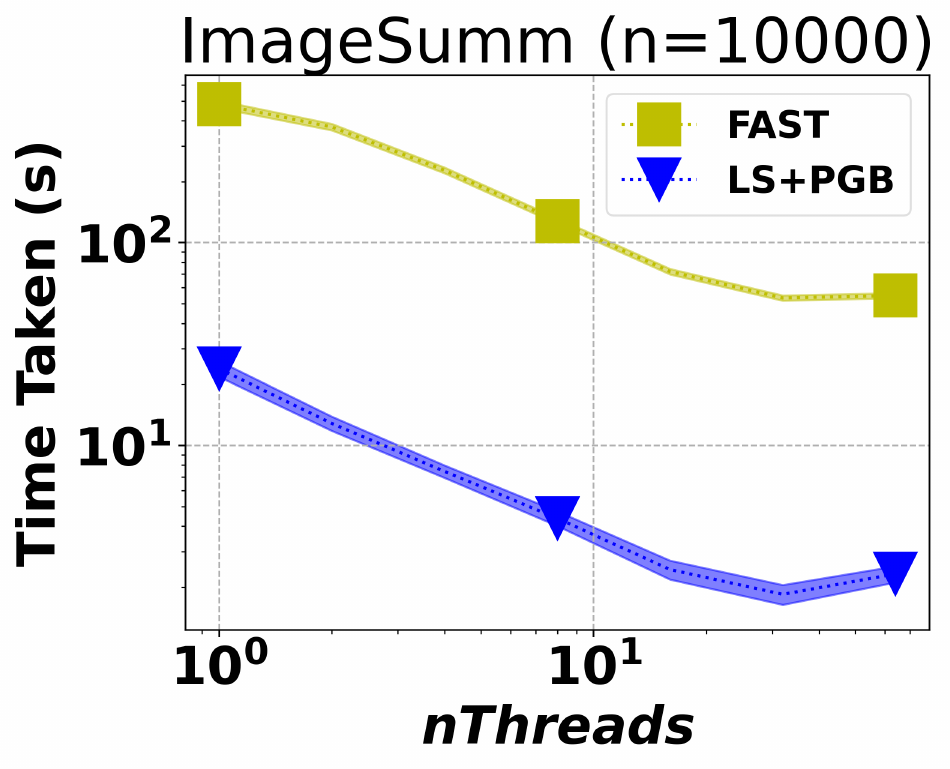}
  }

  \caption{Runtime (s) vs. number of processors. Both axes are log-scaled.} \label{fig:parallel}
\end{figure}  

This experiment set aims to demonstrate the improvement in parallel run time taken by the algorithms to solve the application with increasing number of available threads. We use the \sm: maximum cover on random graphs (MaxCover),
twitter feed summarization (TweetSumm), image summarization (ImageSumm). See
Appendix \ref{subsec:obj} for the definition of the objectives. All experiments were conducted with a constant $k$ value of 1,000 and the number of available threads provided to the algorithms ranged from 1 to 64 threads, doubling the number of threads for each interval in between. The parallel run time of the experiments is measured in seconds. As shown in fig. \ref{fig:parallel}, the scaling behavior of both  \fast and \flsabr are very similar with the number of processors employed, each algorithm exhibiting a linear speedup initially with the number of processors, which
plateaus past a certain number of processors. \flsabr outperforms \fast across all instances with an average speedup of 10 over \fast.

\subsection{Comparison vs adaptive algorithms from \citet{Fahrbach2018}}

\begin{figure}[t]
  \subfigure[]{
    \includegraphics[width=0.23\textwidth, height=0.11\textheight]{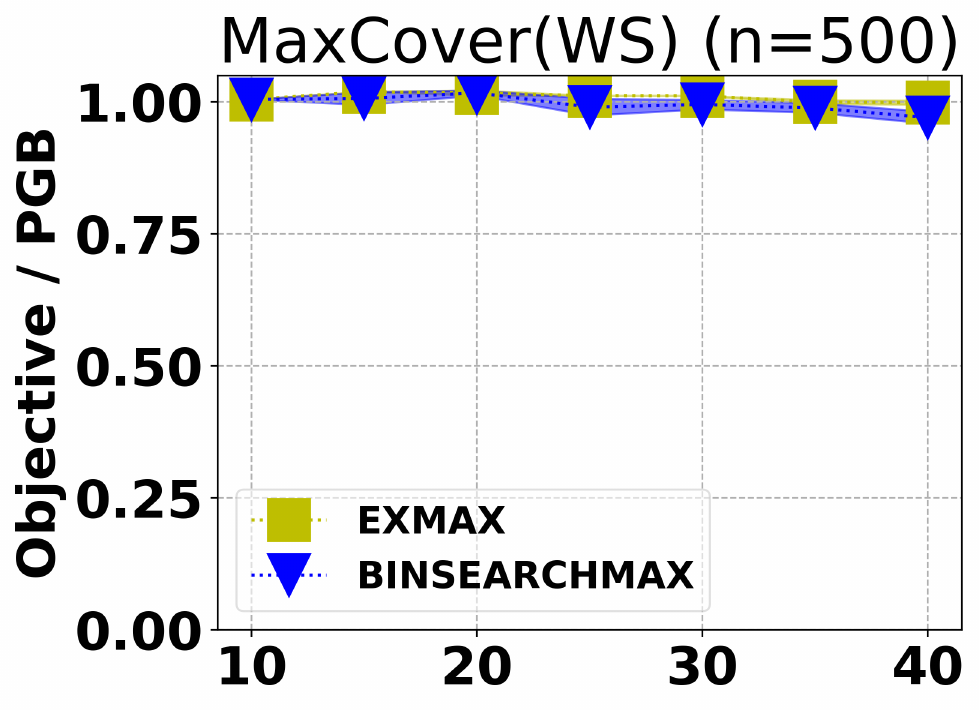}\label{fig:SODA_objA}
  }
  \subfigure[]{
    \includegraphics[width=0.23\textwidth, height=0.11\textheight]{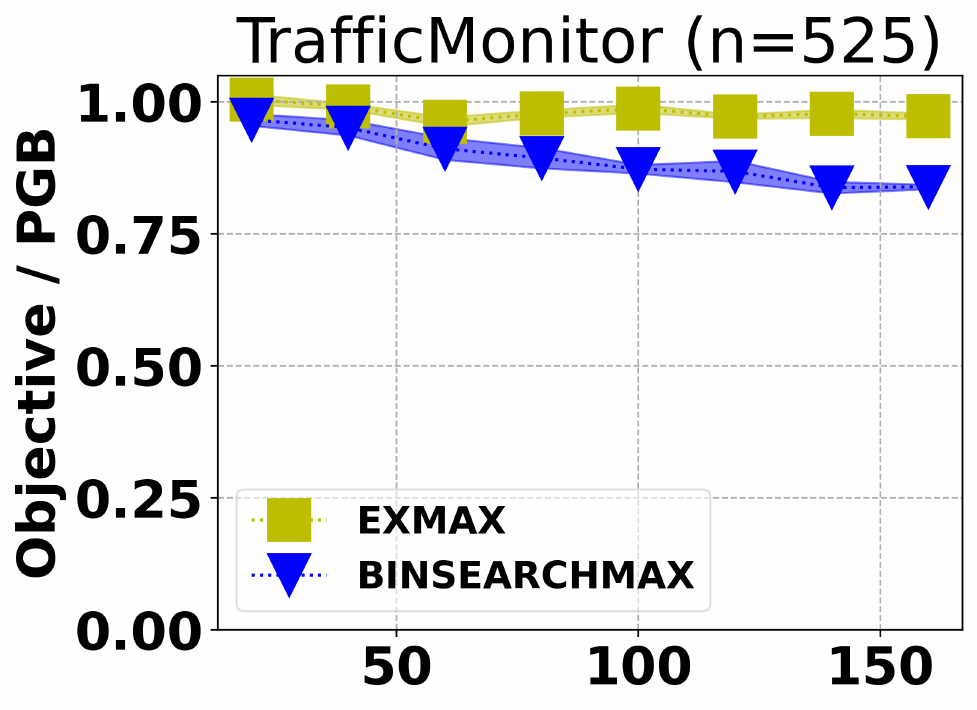}
  }
  \subfigure[]{
    \includegraphics[width=0.23\textwidth, height=0.11\textheight]{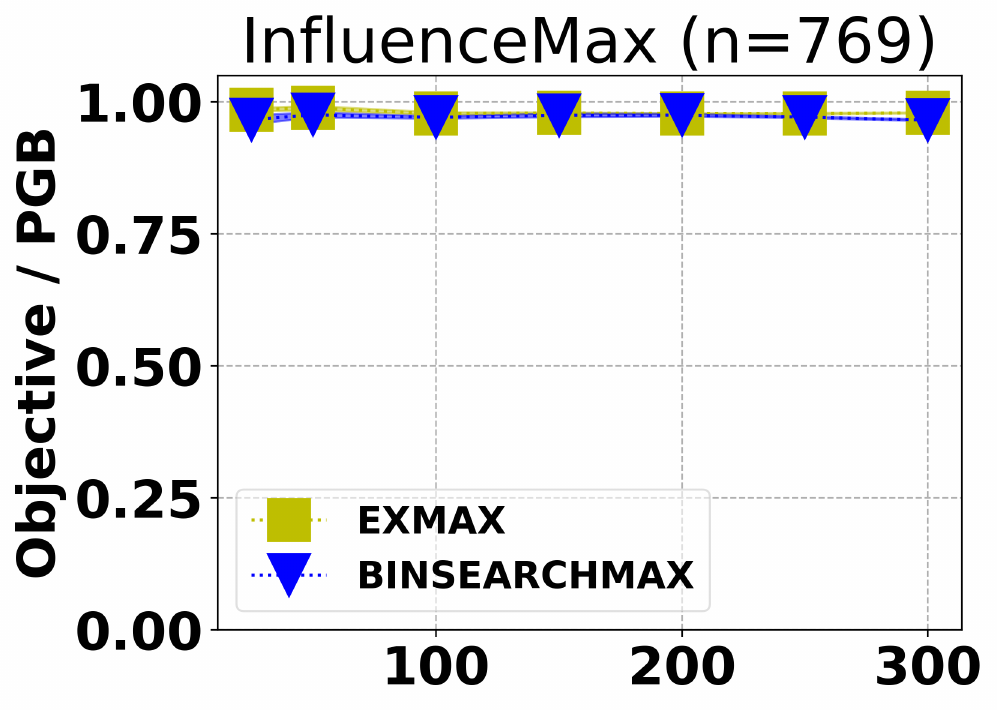} \label{fig:objF}
  }
  \subfigure[]{
    \includegraphics[width=0.23\textwidth, height=0.11\textheight]{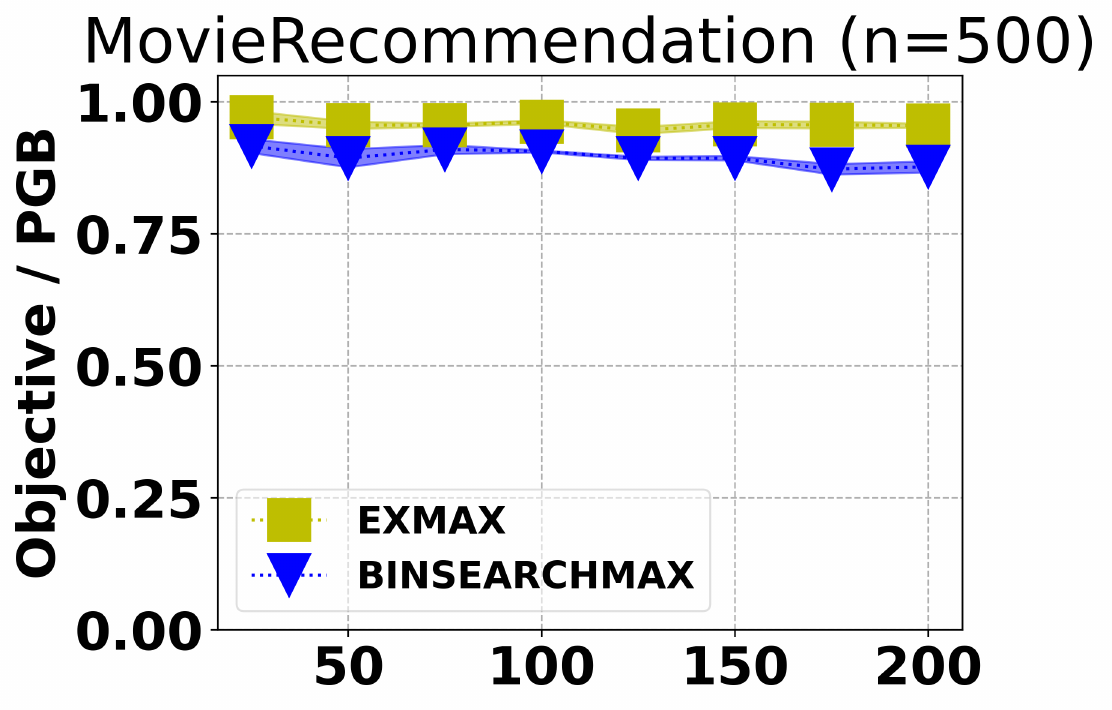}\label{fig:SODA_objD}
  }
  \subfigure[]{
    \includegraphics[width=0.23\textwidth, height=0.11\textheight]{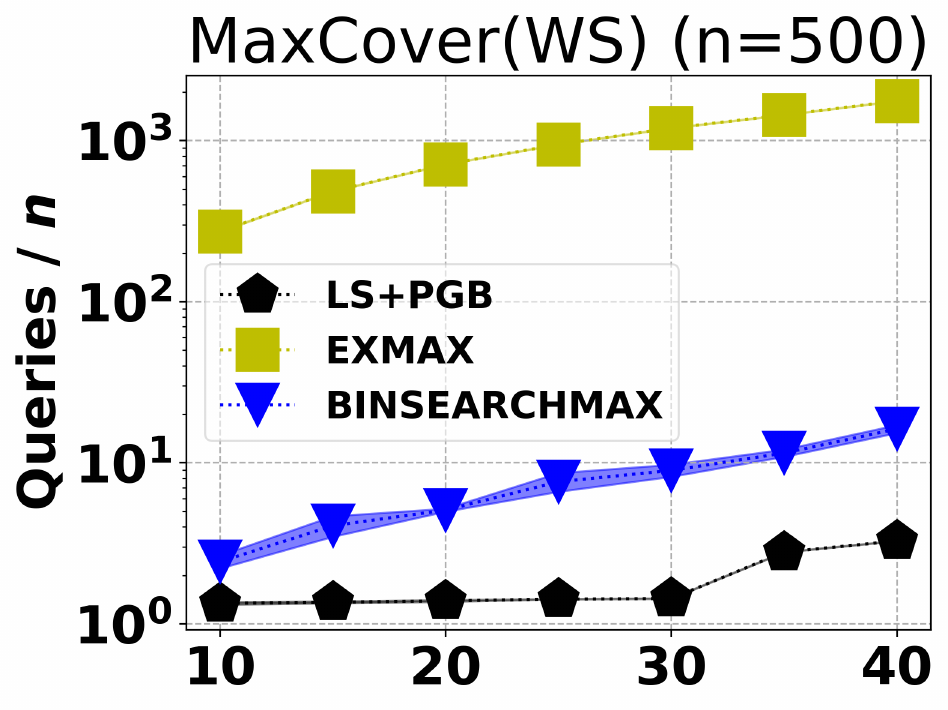}\label{fig:SODA_qryA}
  }
  \subfigure[]{
    \includegraphics[width=0.23\textwidth, height=0.11\textheight]{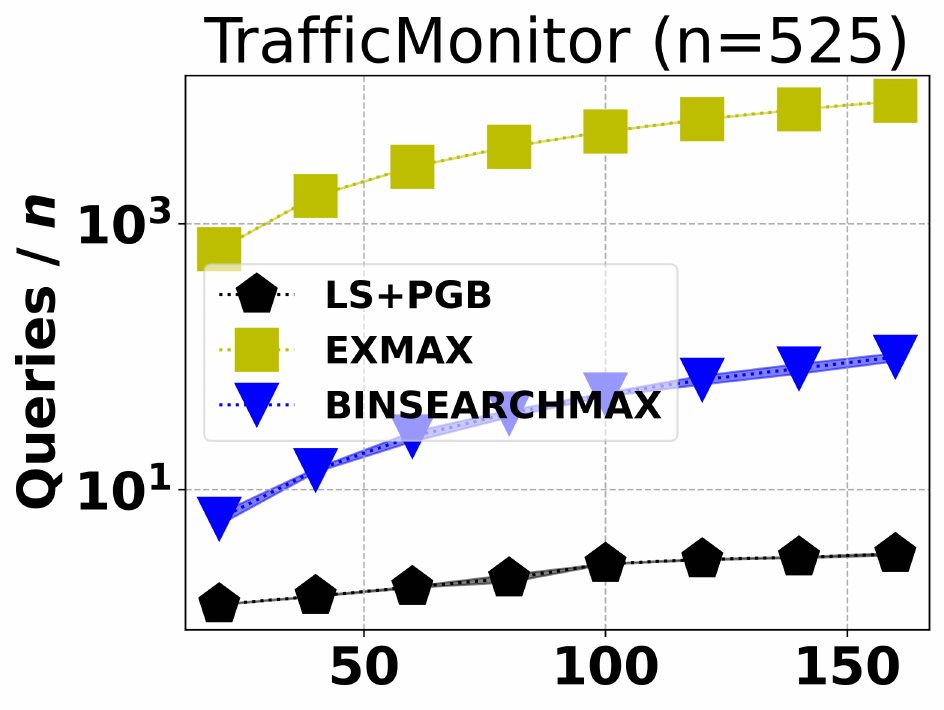}
  }
  \subfigure[]{
    \includegraphics[width=0.23\textwidth, height=0.11\textheight]{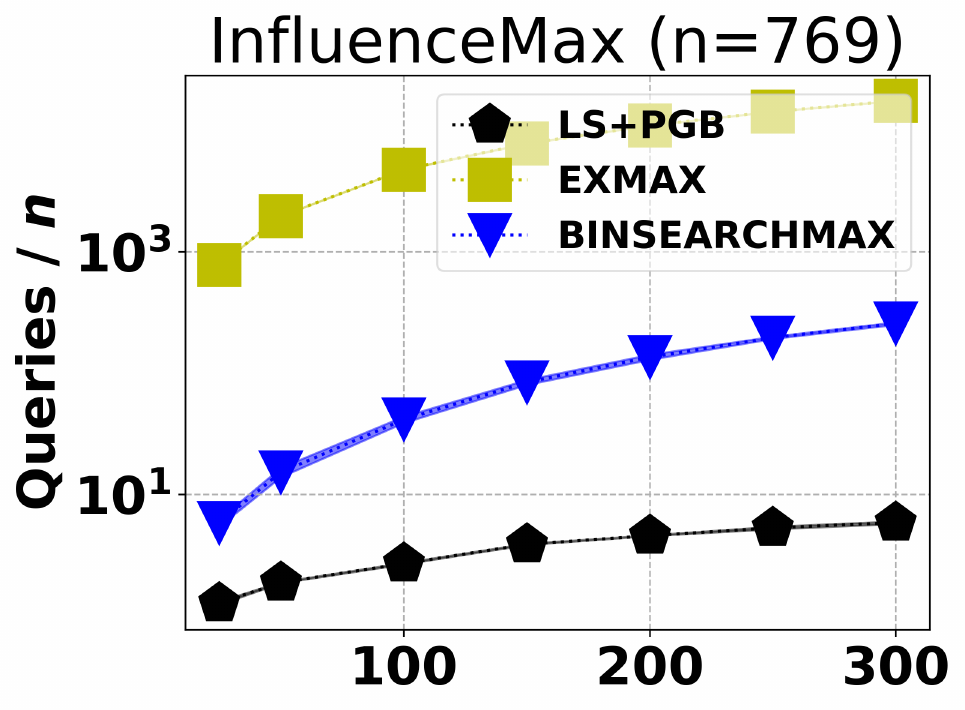} \label{fig:objF}
  }
  \subfigure[]{
    \includegraphics[width=0.23\textwidth, height=0.11\textheight]{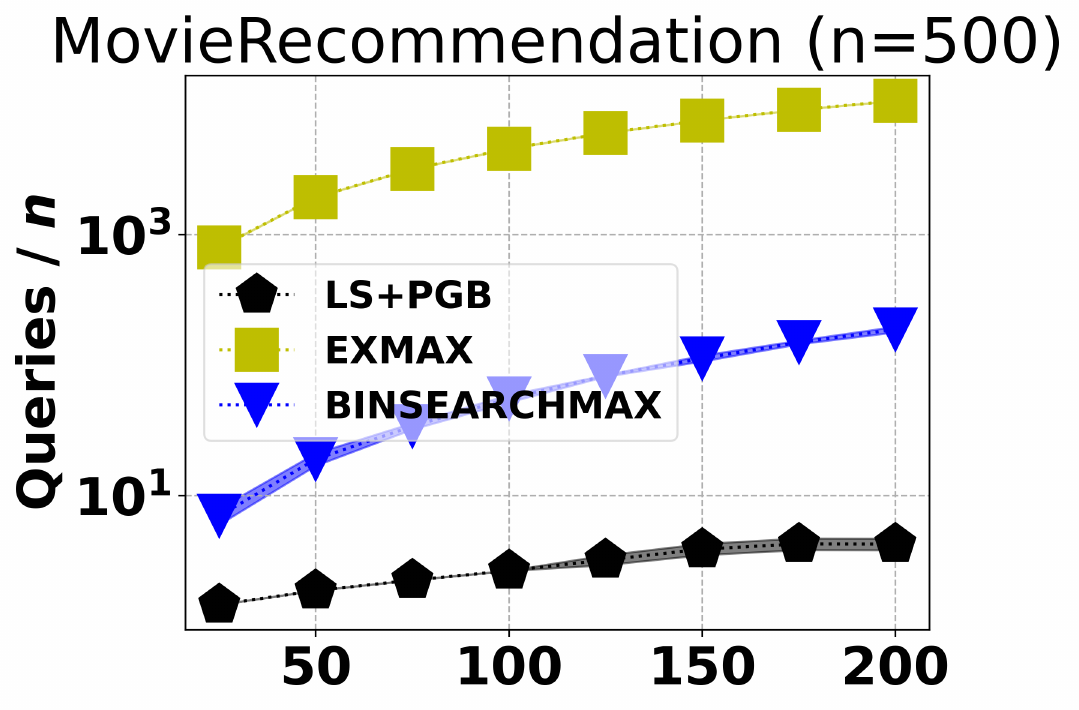}\label{fig:SODA_qryD}
  }
  \subfigure[]{
    \includegraphics[width=0.23\textwidth, height=0.11\textheight]{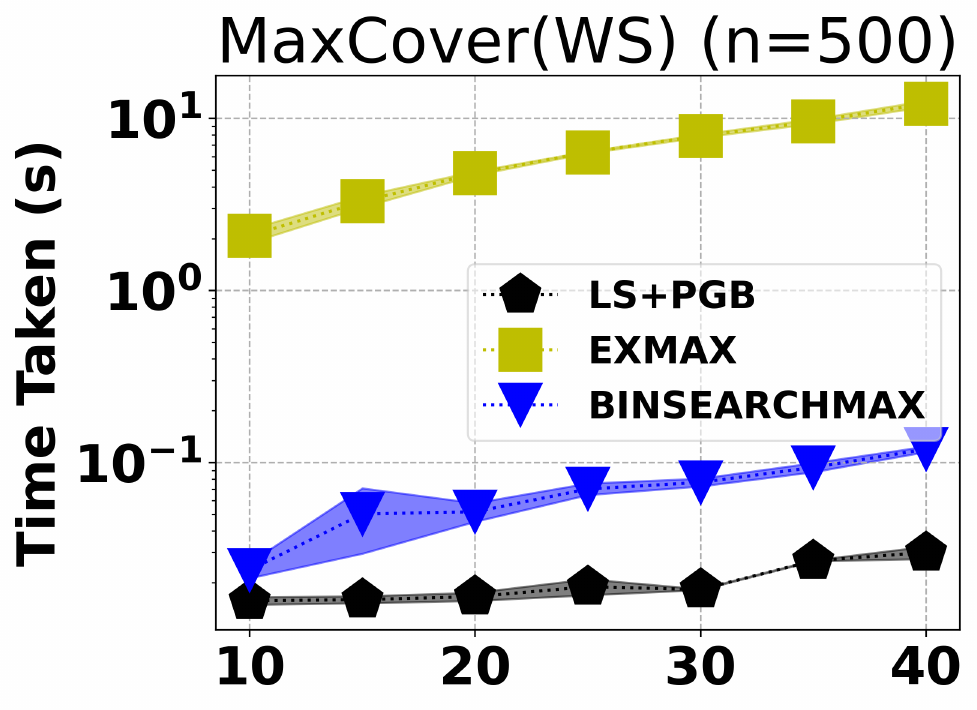}\label{fig:SODA_timeA}
  }
  \subfigure[]{
    \includegraphics[width=0.23\textwidth, height=0.11\textheight]{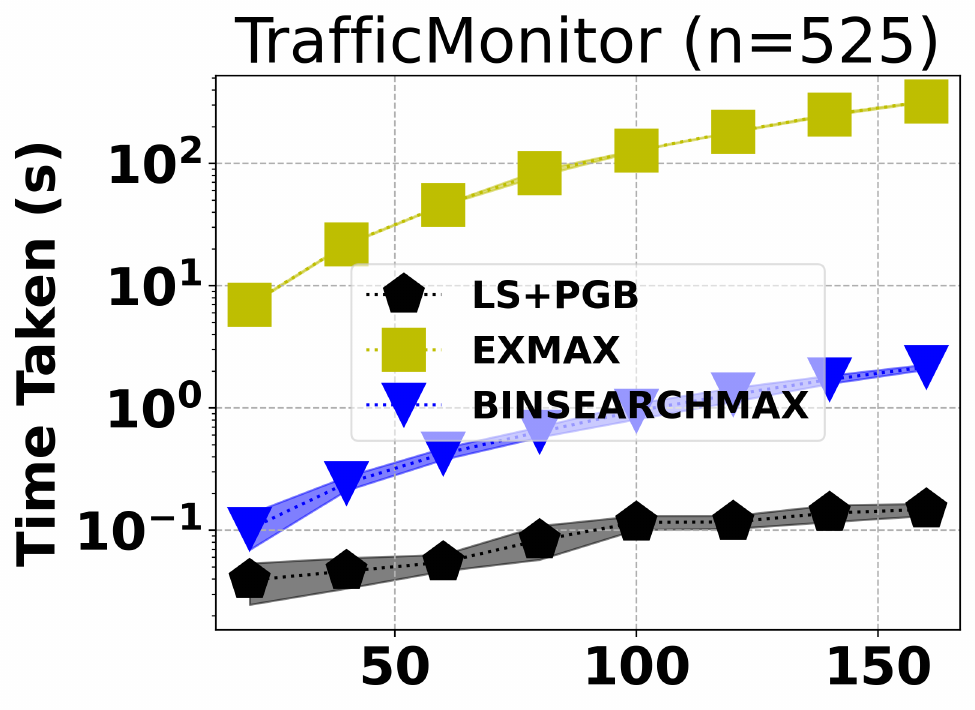}
  }
  \subfigure[]{
    \includegraphics[width=0.23\textwidth, height=0.11\textheight]{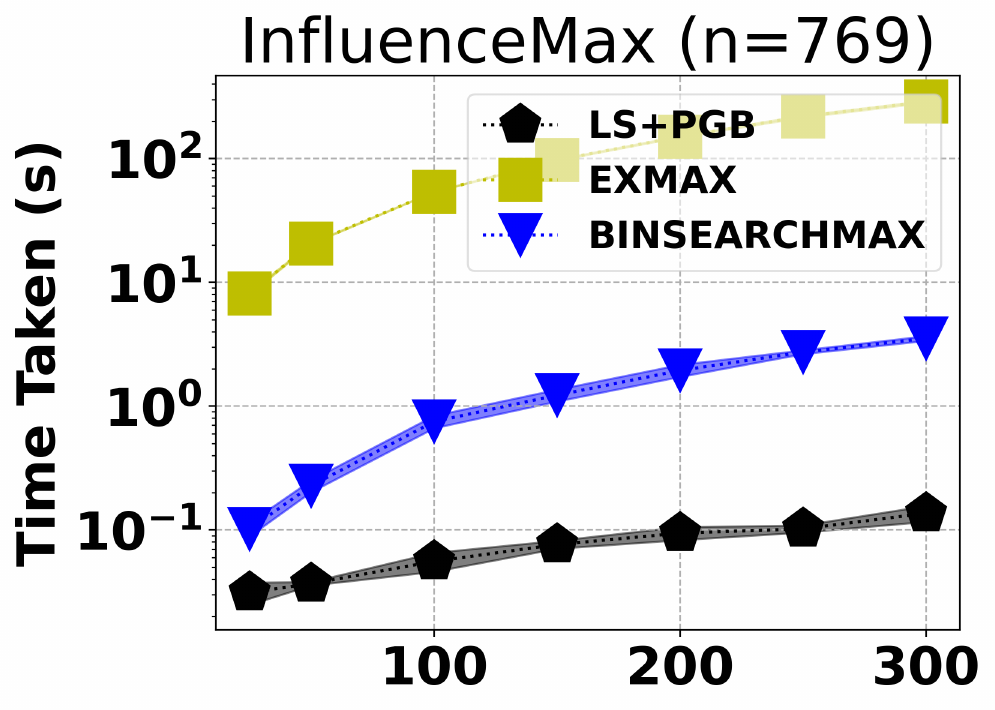} \label{fig:objF}
  }
  \subfigure[]{
    \includegraphics[width=0.23\textwidth, height=0.11\textheight]{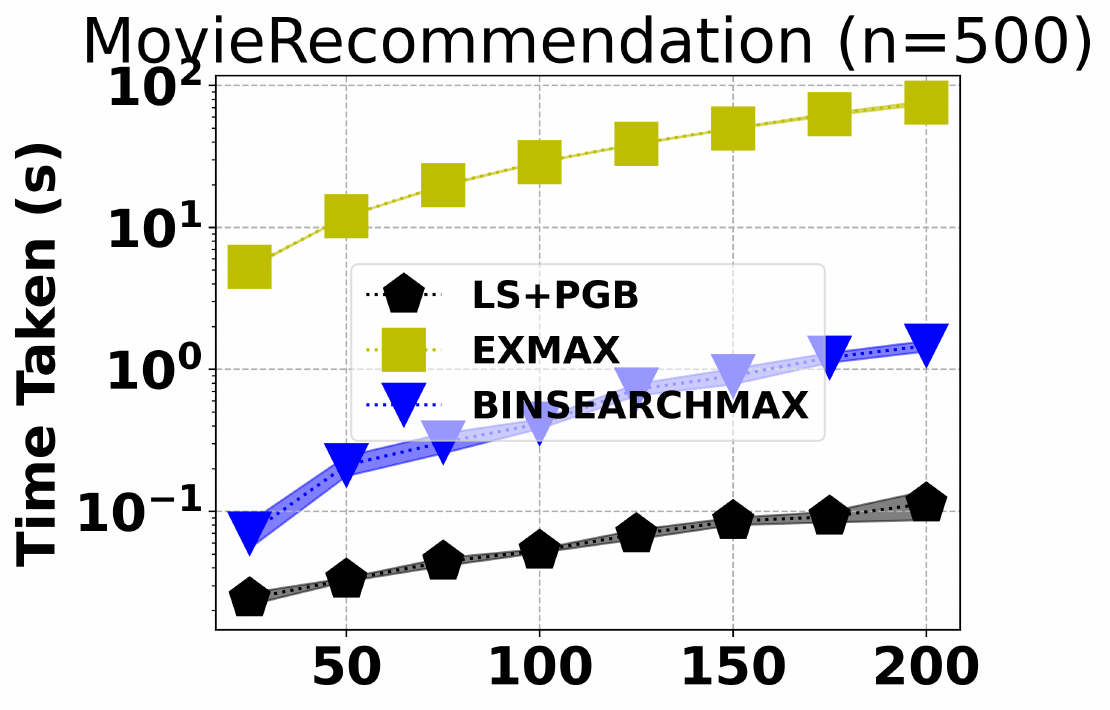}\label{fig:SODA_timeD}
  }
  \caption{Evaluation of adaptive algorithms from \citet{Fahrbach2018} on MaxCover(WS), TrafficMonitor, InfluenceMax and MovieRecommendation in terms of objective value normalized by the standard greedy value (Figure \ref{fig:SODA_objA} -  \ref{fig:SODA_objD}), total number of queries (Figure \ref{fig:SODA_qryA} -  \ref{fig:SODA_qryD}) and the time required by each algorithm (Figure \ref{fig:SODA_timeA} -  \ref{fig:SODA_timeD})} \label{fig:query}
\end{figure} 

 Since \exm and the \bsm algorithm of \citet{Fahrbach2018} has better theoretical guarantee than \fast, we ran additional experiments comparing against both the \exm and the \bsm algorithm, as implemented by the authors of \fast. In this highly optimized implementation, a single query per processor is used in each adaptive round instead of the number of queries required theoretically as description of the implementation in \citet{Breuer2019}. The evaluation was performed on MaxCover(WS), TrafficMonitor, InfluenceMax and MovieRecommendation in terms of objective value normalized by the standard greedy value (Figure \ref{fig:SODA_objA} -  \ref{fig:SODA_objD}), total number of queries (Figure \ref{fig:SODA_qryA} -  \ref{fig:SODA_qryD}) and the time required by each algorithm. In summary, our algorithm \flsabr is faster by roughly an order of magnitude and uses roughly an order of magnitude fewer queries over \bsm.

%%% Local Variables:
%%% mode: latex
%%% TeX-master: "main.tex"
%%% End:

% \input{FastExp1.tex}
\end{document}